\documentclass{article}
\usepackage{geometry}

\usepackage{kr}

\usepackage{times}
\usepackage{soul}
\usepackage{url}
\usepackage[hidelinks]{hyperref}
\usepackage[utf8]{inputenc}
\usepackage[small]{caption}
\usepackage{graphicx}
\usepackage{amsmath}
\usepackage{amsthm}
\usepackage{booktabs}
\usepackage{algorithm}
\usepackage{algorithmic}
\usepackage{amsfonts}
\usepackage{amsmath}
\usepackage{multirow}
\usepackage{todonotes}
\usepackage{xspace}
\usepackage{amsthm}
\usepackage{changepage}
\usepackage{stmaryrd}

\newtheorem{theorem}{Theorem}

\newtheorem{lemma}[theorem]{Lemma}
\newtheorem{proposition}[theorem]{Proposition}
\newtheorem{fact}[theorem]{Fact}

\theoremstyle{definition} 
\newtheorem{definition}{Definition}

\newcommand{\FormatEntitySet}[1]{\ensuremath{\mathsf{#1}}\xspace}
\newcommand{\FormatFormulaSet}[1]{\ensuremath{\mathcal{#1}}\xspace}
\newcommand{\FormatFormulaSetOfSets}[1]{\ensuremath{\mathfrak{#1}}\xspace}
\newcommand{\FormatFunction}[1]{\ensuremath{\mathit{#1}}\xspace}
\newcommand{\FF}[1]{\FormatFunction{#1}}
\newcommand{\FormatPredicate}[1]{\ensuremath{\mathtt{#1}}\xspace}
\newcommand{\FP}[1]{\FormatPredicate{#1}}

\newcommand{\Tuple}[1]{\ensuremath{\langle #1 \rangle}\xspace}
\newcommand{\Vector}[1]{\ensuremath{\vec{#1}}\xspace}
\newcommand{\Vx}{\Vector{x}}

\newcommand{\Vz}{\Vector{z}}

\newcommand{\Vv}{\Vector{v}}
\newcommand{\Vt}{\Vector{t}}
\newcommand{\VT}{\Vector{T}}
\newcommand{\Vu}{\Vector{u}}
\newcommand{\Va}{\Vector{a}}
\newcommand{\Vb}{\Vector{b}}
\newcommand{\VC}{\smash{\Vector{C}}}

\newcommand{\FirstItem}{(i)\xspace}
\newcommand{\SecondItem}{(ii)\xspace}
\newcommand{\ThirdItem}{(iii)\xspace}

\newcommand{\Variables}{\FormatEntitySet{Vars}}
\newcommand{\Nulls}{\FormatEntitySet{Nulls}}
\newcommand{\Predicates}{\FormatEntitySet{Preds}}

\newcommand{\Terms}{\FormatEntitySet{Terms}}

\newcommand{\EntitiesIn}[2]{\ensuremath{#1(#2)}\xspace}
\newcommand{\EI}[2]{\EntitiesIn{#1}{#2}}

\newcommand{\Arity}{\FF{Ar}}
\newcommand{\MaxArity}{{\bar{m}}} 
\newcommand{\PredsPerArity}{{\bar{n}}} 

\newcommand{\Body}{\ensuremath{\beta}\xspace}
\newcommand{\Head}{\ensuremath{\eta}\xspace}
\newcommand{\Rule}{\ensuremath{\rho}\xspace}

\newcommand{\FactSet}{\FormatFormulaSet{F}}

\newcommand{\Database}{\FormatFormulaSet{D}}
\newcommand{\D}{\Database}

\newcommand{\DatabaseAux}{\FormatFormulaSet{E}}
\newcommand{\DA}{\DatabaseAux}
\newcommand{\RuleSet}{\FormatFormulaSet{R}}
\newcommand{\R}{\RuleSet}

\newcommand{\Interpretation}{\FormatFormulaSet{I}\xspace}
\newcommand{\I}{\Interpretation}
\newcommand{\InterpretationJ}{\FormatFormulaSet{J}\xspace}
\newcommand{\J}{\InterpretationJ}
\newcommand{\Model}{\FormatFormulaSet{M}}
\newcommand{\M}{\Model}
\newcommand{\ISet}{\FormatFormulaSetOfSets{I}}
\newcommand{\IS}{\ISet}

\newcommand{\ModelSet}{\FormatFormulaSetOfSets{M}}
\newcommand{\MS}{\ModelSet}

\newcommand{\Fragment}{\FormatFormulaSetOfSets{F}}

\newcommand{\Homomorphism}{\ensuremath{\textit{h}}\xspace}
\newcommand{\Hom}{\Homomorphism}
\newcommand{\HomomorphismAux}{\ensuremath{\textit{g}}\xspace}
\newcommand{\HomA}{\HomomorphismAux}

\newcommand{\Query}{\FormatFormulaSetOfSets{Q}}

\newcommand{\Blank}{\text{\textvisiblespace}}
\newcommand{\Separator}{\mathord{\parallel}}
\newcommand{\StartingState}{\ensuremath{q_S}\xspace}
\newcommand{\AcceptingState}{\ensuremath{q_A}\xspace}
\newcommand{\RejectingState}{\ensuremath{q_R}\xspace}

\newcommand{\States}{\ensuremath{Q}\xspace}
\newcommand{\Alphabet}{\ensuremath{\Gamma}\xspace}
\newcommand{\TransitionFunction}{\ensuremath{\delta}\xspace}
\newcommand{\TM}{\ensuremath{M}\xspace}
\newcommand{\ETM}{\ensuremath{P}\xspace}

\newcommand{\PrGoal}{\FormatPredicate{Goal}}
\newcommand{\PrWorldGoal}{\FormatPredicate{Acc}}
\newcommand{\PrHalt}{\FormatPredicate{Halt}}
\newcommand{\PrBrake}{\FormatPredicate{Brake}}
\newcommand{\PrReal}{\FormatPredicate{Real}}
\newcommand{\PrBrakeBody}{\FormatPredicate{B}}
\newcommand{\PrDomain}{\FormatPredicate{DbDom}}
\newcommand{\PrFreshIDBIn}[1]{\ensuremath{\FP{In}_{#1}}\xspace}
\newcommand{\PrFreshIDBNotIn}[1]{\ensuremath{\FP{NIn}_{#1}}\xspace}
\newcommand{\PrFreshIDBInTape}[1]{\ensuremath{\FP{In}_{#1}'}\xspace}
\newcommand{\PrFreshIDBNotInTape}[1]{\ensuremath{\FP{NIn}_{#1}'}\xspace}
\newcommand{\PrFirst}{\FormatPredicate{First}}
\newcommand{\PrLast}{\FormatPredicate{Last}}
\newcommand{\PrOrder}{\FormatPredicate{LT}}
\newcommand{\PrRoot}{\FormatPredicate{Root}}
\newcommand{\PrLeaf}{\FormatPredicate{Leaf}}

\newcommand{\PrHasChild}{\FormatPredicate{Chi}}
\newcommand{\PrLink}{\FormatPredicate{Rep}} 
\newcommand{\PrEncoding}{\FormatPredicate{Enc}}
\newcommand{\PrCopy}{\FormatPredicate{Cpy}}
\newcommand{\PrCopyPlusOne}{\ensuremath{\PrCopy_{\!+1}}\xspace}
\newcommand{\PrNext}{\FormatPredicate{Nxt}}
\newcommand{\PrNextTrans}{\ensuremath{\FormatPredicate{Nxt}^{\!+}}\xspace}
\newcommand{\PrEndTape}{\FormatPredicate{End}}
\newcommand{\PrStep}{\FormatPredicate{Stp}}
\newcommand{\PrCell}[1]{\ensuremath{\FormatPredicate{S}_{#1}}}
\newcommand{\PrHead}[1]{\ensuremath{\FormatPredicate{Hd}_{#1}}}
\newcommand{\PrLoad}[1]{\ensuremath{\FormatPredicate{Ld}_{#1}}\xspace}
\newcommand{\PrLoadEncoding}{\ensuremath{\FormatPredicate{LdE}}\xspace}
\newcommand{\PrReady}[1]{\ensuremath{\FormatPredicate{Rdy}_{#1}}\xspace}
\newcommand{\PrEq}{\ensuremath{\FormatPredicate{Eq}}\xspace}
\newcommand{\PrNeq}{\ensuremath{\FormatPredicate{NEq}}\xspace}
\newcommand{\PrP}{\ensuremath{\FormatPredicate{p}}\xspace}
\newcommand{\PrQ}{\ensuremath{\FormatPredicate{q}}\xspace}
\newcommand{\PrR}{\ensuremath{\FormatPredicate{R}}\xspace}
\newcommand{\PrEdge}{\ensuremath{\FormatPredicate{ed}}\xspace}

\newcommand{\PrIns}[1]{\ensuremath{\FormatPredicate{Ins_{#1}}}\xspace}
\newcommand{\PrHatted}[1]{\ensuremath{\FormatPredicate{#1'}}\xspace}
\newcommand{\PrInit}{\ensuremath{\FormatPredicate{Init}}\xspace}
\newcommand{\PrDone}{\ensuremath{\FormatPredicate{Done}}\xspace}
\newcommand{\PrEmpty}{\ensuremath{\FormatPredicate{Empty}}\xspace}
\newcommand{\PrSubs}{\ensuremath{\FormatPredicate{Subs}}\xspace}

\newcommand{\FirstNull}{\ensuremath{u_\alpha}\xspace}
\newcommand{\LastNull}{\ensuremath{u_\omega}\xspace}

\newcommand{\FunStartingConf}{\FF{StConf}}
\newcommand{\FunComputation}{\FF{Comp}}

\newcommand{\FunOrders}{\FF{Ords}}
\newcommand{\FunCompletions}{\FF{Compls}}
\newcommand{\FunSeed}{\FF{Seed}}
\newcommand{\FunOrder}{\FF{Order}}

\newcommand{\FunInterpretation}[1]{\ensuremath{\FF{Int}_{#1}}\xspace}

\newcommand{\FunDatabase}{\FF{Db}}

\makeatletter
\newsavebox{\@brx}
\newcommand{\llangle}[1][]{\savebox{\@brx}{\(\m@th{#1\langle}\)}%
  \mathopen{\copy\@brx\kern-0.5\wd\@brx\usebox{\@brx}}}
\newcommand{\rrangle}[1][]{\savebox{\@brx}{\(\m@th{#1\rangle}\)}%
  \mathclose{\copy\@brx\kern-0.5\wd\@brx\usebox{\@brx}}}
\makeatother

\title{Capturing Homomorphism-Closed Decidable Queries with Existential Rules}

\author{
Camille~Bourgaux$^1$\and
David~Carral$^2$\and
Markus~Kr\"{o}tzsch$^3$\and
Sebastian~Rudolph$^3$\and
Micha\"{e}l~Thomazo$^1$
\affiliations
$^1$ DI ENS, ENS, CNRS, PSL University \& Inria, Paris, France \\
$^2$ LIRMM, Inria, University of Montpellier, CNRS, Montpellier, France \\
$^3$ Technische Universit\"at Dresden, Dresden, Germany \\
\emails
\{camille.bourgaux, david.carral, michael.thomazo\}@inria.fr,\\ \{markus.kroetzsch,  sebastian.rudolph\}@tu-dresden.de
}

\begin{document}

\maketitle

\begin{abstract}
Existential rules are a very popular ontology-mediated query language for which the chase represents a generic computational approach for query answering. 
It is straightforward that existential rule queries exhibiting chase termination are decidable and can only recognize properties that are preserved under homomorphisms. 
In this paper, we show the converse: every decidable query that is closed under homomorphism can be expressed by an existential rule set for which the standard chase universally terminates. Membership in this fragment is not decidable, but we show via a diagonalisation argument that this is unavoidable.
%
%
%
\end{abstract}

\section{Introduction}
\label{section:introintro}


At the core of contemporary logic-based knowledge representation is the concept of \emph{querying} data sources, often using elaborate query formalisms that allow for taking background knowledge into account. The classical decision problem related to such knowledge-aware querying is \emph{Boolean query entailment}. From an abstract point of view, a Boolean query identifies a class of \emph{databases} $\Database$ -- those that satisfy the query, i.e., to which the query ``matches''. This view allows us to define and investigate properties of (abstract) queries independently from the syntax used to specify them. Such properties can be structural (morphisms, closure properties) or computational (decidability, complexity).

A very popular querying formalism are \emph{existential rules}, also referred to as \emph{tuple-generating dependencies}.
It is straightforward that the class of databases satisfying some existential rule query is closed under homomorphisms and recursively enumerable.
Conversely, it was established that \textsl{every} homomorphism-closed query that is recursively enumerable can be expressed using existential rules \cite{DBLP:conf/ijcai/RudolphT15}.
That is, plain existential rules already realize their full potential; further syntactic extensions within these boundaries do not enhance expressivity.

For questions related to automated deduction, however, a more restricted requirement than recursive enumerability is of central interest: decidability.
Therefore, the crucial question we tackle in this paper is:

\smallskip

\noindent\emph{Can we characterize an existential rules fragment capable of expressing \textsl{every} decidable homomorphism-closed query?}

\smallskip

\noindent The generic computational paradigm for existential rules, the \emph{chase} \cite{BeeriVardi:Chase84}, is based on repetitive, for\-ward-chaining rule application, starting from the database.
As this may cause the iterated introduction of new domain elements, this procedure is not guaranteed to terminate -- yet, termination is a crucial criterion for decidability.  
The chase comes in several variants, mainly differing in their (increasingly thorough) mechanisms to prevent unnecessary rule applications: 
While the \emph{Skolem} chase \cite{Marnette09:superWA} essentially just avoids duplicate rule applications, the \emph{standard} \cite{FaginKMP05} and the \emph{core} chase \cite{DNR08:corechase} check for redundancy on a local and global level, respectively.

The class of existential rule sets with terminating\footnote{We always mean universal termination, i.e., for every database.} Skolem chase has already been weighed and found wanting:
it only comprises those queries that are already expressible in plain Datalog -- and hence can be evaluated in polynomial time \cite{Marnette09:superWA,KR11:jointacyc,ZZY15:finitechase}. 
For the standard-chase-terminating and the core-chase-terminating existential rules classes, on the other hand, we only know that the former is contained in the latter \cite{DBLP:journals/fuin/GrahneO18}, but little more than that \cite{DBLP:conf/icdt/KrotzschMR19}.
In this paper, we clarify the situation significantly by showing the following:

\smallskip

\noindent\emph{Standard-chase-terminating existential rules capture the class of all decidable homomorphism-closed queries.}

\smallskip

Notably, this implies that standard-chase-terminating and core-chase-terminating existential rule queries are equally expressive and no decidable enhancement of this formalism that preserves homomorphism-closedness (e.g. by allowing disjunction in rule heads) can be strictly more expressive. 

As a downside, the existential rules fragment thus identified is not even semi-decidable, but we show via a diagonalisation argument that this downside is, in fact, unavoidable.


Additional proofs and details are given in the appendix. 
\section{Preliminaries}\label{sec_prelims}

\paragraph*{Rules}
We consider first-order formulas over countably infinite sets \Variables of variables
and $\Predicates$ of \emph{predicates}, where
each $\PrP \in \Predicates$ has an arity $\Arity(\PrP) \geq 0$.
Lists of variables are denoted $\vec{x}=x_1,\ldots,x_k$ and will be
treated like sets when order is not relevant.
An \emph{atom} is an expression $\PrP(\vec{x})$ with $\PrP \in \Predicates$
and $|\vec{x}|=\Arity(\PrP)$.

The fragment of \emph{disjunctive existential rules} consists of formulae of the form:
\begin{align}
\forall \vec{x}.\Big(\Body[\vec{x}] &\to \bigvee\nolimits_{\!\!i=1}^k \exists\vec{y}_i. \Head_i[\vec{x}_i,\vec{y}_i]\Big), \label{eq_rule}
\end{align}
where $\Body[\vec{x}]$ and $\Head_i[\vec{x}_i,\vec{y}_i]$ ($i=1,\ldots,k$) are conjunctions of atoms with variables
$\vec{x}$ and $\vec{x}_i\cup\vec{y}_i$, respectively. We call $\Body$ \emph{body} and $\bigvee\nolimits_{i=1}^k \exists\vec{y}_i.\Head_i$ \emph{head}. Bodies can be empty (we then omit $\to$), but heads must be non-empty.
We require that $\vec{x}$ and $\vec{y_i}$ ($i=1,\ldots,k$) are mutually disjoint and that
$\vec{x}_i\subseteq \vec{x}$ for all $i=1,\ldots,k$.
We single out the fragment of \emph{existential rules} by disallowing disjunction, i.e.\ requiring $k=1$, and \emph{Datalog rules} by disallowing existential quantifiers.
We often omit the universal quantifiers from rules and treat conjunctions of atoms as sets of atoms.

\paragraph*{Databases, Interpretations, and Entailment}
The semantics of formulas is based on logical interpretations, which we define as relational structures over a
countably infinite set $\Nulls$ of \emph{nulls}.
A \emph{schema} $\mathcal{S}$ is a finite set of predicates.
An \emph{interpretation} $\mathcal{I}$ over schema $\mathcal{S}$ is a set of expressions
$\PrP(\vec{n})$ with $\PrP \in \mathcal{S}$ and $\vec{n}$ a list of nulls of length $\Arity(\PrP)$.
We write $\EI{\Nulls}{\mathcal{I}}$ for the set of nulls in $\mathcal{I}$.
A \emph{database} is a finite interpretation.
See also the remarks on this notation below.

A \emph{homomorphism} $h:\mathcal{I}_1\to\mathcal{I}_2$ between interpretations $\mathcal{I}_1$ and $\mathcal{I}_2$
is a mapping $h$ from the nulls in $\mathcal{I}_1$ to the nulls in $\mathcal{I}_2$, such that
$\PrP(\vec{n})\in\mathcal{I}_1$ implies $\PrP(h(\vec{n}))\in\mathcal{I}_2$,
where $h(\vec{n})$ denotes the list of $h$-values over $\vec{n}$.
We also write $h(\mathcal{I})$ for $\{\PrP(h(\vec{n}))\mid \PrP(\vec{n})\in\mathcal{I}\}$.

A \emph{substitution} $\sigma$ is a mapping from variables to nulls, which we extend to lists of variables and formulas
as usual.
A rule $\rho$ as in \eqref{eq_rule} is \emph{satisfied} by interpretation $\mathcal{I}$
if every substitution $\sigma:\vec{x}\to\Nulls$ with $\sigma(\Body)\subseteq\mathcal{I}$
can be extended to a substitution $\sigma':\vec{x}\cup\vec{y}_i\to\Nulls$ for some $i\in\{1,\ldots,k\}$
such that $\sigma'(\Head_i)\subseteq\mathcal{I}$.
Otherwise, if $\sigma(\Body)\subseteq\mathcal{I}$ but no extension $\sigma'$ of $\sigma$ verifies $\sigma'(\Head_i)\subseteq\mathcal{I}$ for some $i\in\{1,\ldots,k\}$, 
then
$\Tuple{\rho,\sigma}$ is \emph{applicable} to $\mathcal{I}$.

$\mathcal{I}$ \emph{satisfies} a set $\Sigma$ of rules if it satisfies every rule in $\Sigma$.
An interpretation $\mathcal{J}$ is \emph{satisfied} by an interpretation $\mathcal{I}$ if there is a homomorphism
$h:\mathcal{J}\to\mathcal{I}$.
$\mathcal{I}$ is a \emph{model} of a rule/rule set/interpretation/database $\mathcal{X}$ if $\mathcal{X}$ is satisfied by $\mathcal{I}$,
written $\mathcal{I}\models\mathcal{X}$.
As usual, we also write $\mathcal{X}\models\mathcal{Y}$ (``$\mathcal{X}$ entails $\mathcal{Y}$'') if every model of $\mathcal{X}$ is a model of $\mathcal{Y}$, 
where $\mathcal{X}$ and $\mathcal{Y}$ might be rules, rule sets, databases, or lists of several such elements.
Note that the semantics of a database $\mathcal{D}$ in this context corresponds to the semantics of
a \emph{Boolean conjunctive query} $\exists\vec{x}.\bigwedge \{\PrP(x_{n_1},\ldots,x_{n_\ell})\mid \PrP(n_1,\ldots,n_\ell)\in\mathcal{D}\}$ -- 
we will therefore not introduce such queries as a separate notion.
Also note that entailment and satisfaction between interpretations/databases coincide.

\paragraph*{Abstract Queries, Expressivity, and Decidability}
An \emph{(abstract) query} $\Query$ over a schema $\mathcal{S}$ is a set of databases over $\mathcal{S}$
that is \emph{closed under isomorphism}, i.e., such that whenever $\mathcal{D}\in\Query$ and $\mathcal{D}'$ is
obtained from $\mathcal{D}$ by bijective renaming of nulls, then $\mathcal{D}'\in\Query$.
$\Query$ is further \emph{closed under homomorphisms} if, for all $\mathcal{D}\in\Query$ and 
homomorphisms $h:\mathcal{D}\to\mathcal{D}'$, we have $\mathcal{D}'\in\Query$.

\begin{definition}\label{def_express}
Let $\PrGoal$ be a nullary predicate. 
A query $\Query$ over $\mathcal{S}$ is \emph{expressed by} a set $\Sigma$ of rules if,
for every database $\mathcal{D}$ over $\mathcal{S}$, 
we have $\mathcal{D}\in\Query$ if and only if $\Sigma,\mathcal{D}\models\PrGoal$.
\end{definition}
 
To discuss decidability of queries, we need to conceive databases as Turing machine inputs over a fixed alphabet.
A \emph{serialisation} for a schema $\mathcal{S}$ is a word $s\in(\{0,1,\Separator\}\cup\mathcal{S})^*$
of the form $e_1\cdots e_n$ where $n\geq 0$ and 
$e_i=\PrP_i\Separator w_{i1}\Separator\cdots\Separator w_{i\Arity(\PrP_i)}\Separator$ for $w_{ij}\in\{0,1\}^+$
and $\PrP_i\in\mathcal{S}$.
Given $s$ of this form and an injection $\eta: \{0,1\}^+\to\Nulls$, let
$\eta(s)$ denote the database $\{\PrP_i(\eta(w_{i1}),\ldots,\eta(w_{i\Arity(\PrP_i)}))\mid 1\leq i\leq n\}$.
Then $s$ \emph{corresponds to} a database $\mathcal{D}$ if $\eta(s)$ is isomorphic to $\mathcal{D}$;
note that this does not depend on the choice of $\eta$.

A query $\Query$ with schema $\mathcal{S}$ is \emph{decidable} if the set of all serialisations
for $\mathcal{S}$ that correspond to some $\mathcal{D}\in\Query$ is a decidable language.

\paragraph*{Remarks on our Notation}
Many works consider constants to appear in databases (not just nulls), but complexity and expressivity is usually studied
for queries that are closed under isomorphisms, a.k.a.\ \emph{generic} \cite[Ch. 16]{DBLP:books/aw/AbiteboulHV95},
and nulls are more natural there. One can admit finitely many exceptions (elements that must not be renamed),
but such ``constants'' can be simulated by marking them with dedicated unary predicates.

Specifying logical interpretations as sets of ``atoms'' that may use nulls is a notational
convenience with some side effects: our interpretations cannot contain elements that do not stand in
any relation, but they can have an empty domain. Both aspects do not change the notion
of logical entailment \emph{on the formulas we consider}.

\paragraph*{Universal Models and the Chase}

Entailment of databases (corresponding to Boolean conjunctive queries) can be decided by considering
only a subset of all models.
Given sets $\mathfrak{I}$ and $\mathfrak{K}$ of interpretations, $\mathfrak{I}$ is \emph{universal}
for $\mathfrak{K}$ if, for all $\mathcal{K}\in\mathfrak{K}$, there is $\mathcal{I}\in\mathfrak{I}$
and a homomorphism $\mathcal{I}\to\mathcal{K}$.
Consider a rule set $\Sigma$ and database $\mathcal{D}$, and let $\mathfrak{M}$ be the set of all models
of $\Sigma,\mathcal{D}$. Then $\mathfrak{I}$ is a \emph{universal model set} for $\Sigma$ and $\mathcal{D}$
if $\mathfrak{I}\subseteq\mathfrak{M}$ and $\mathfrak{I}$ is universal for $\mathfrak{M}$.

\begin{fact}\label{fact_uniModSetCorrect}
If $\mathfrak{I}$ is a universal model set for $\Sigma$ and $\mathcal{D}$ then, for every database $\mathcal{C}$, we have
$\Sigma,\mathcal{D}\models\mathcal{C}$ iff $\mathcal{I}\models\mathcal{C}$ for all $\mathcal{I}\in\mathfrak{I}$.
\end{fact}

Universal model sets can be computed with the \emph{chase} algorithm. Here,
we consider a variation of the \emph{standard} (or \emph{restricted}) chase for rules with disjunction, introduced by \cite{DBLP:conf/ijcai/CarralDK17}.

\begin{definition}\label{definition_chase}
A \emph{chase tree} for a rule set $\Sigma$ and database $\mathcal{D}$ is a (finite or infinite) tree 
where each node is labelled by a database, such that:
\begin{enumerate}
\item The root is labelled with $\mathcal{D}$.
\item For every node with label $\mathcal{E}$ that has $\ell$ children labelled $\mathcal{C}_1,\ldots,\mathcal{C}_\ell$, there is a rule $\Rule\in \Sigma$ of the form \eqref{eq_rule} and a substitution $\sigma:\vec{x}\to\Nulls$ such that
(i) $\Tuple{\Rule, \sigma}$ is applicable to $\mathcal{E}$,
(ii) $\Rule$ has $k=\ell$ head disjuncts, and
(iii) $\mathcal{C}_i=\mathcal{E}\cup\sigma_i(\Head_i)$ where $\sigma_i$ extends $\sigma$ by mapping each variable $y\in\vec{y}_i$ to a fresh null.
\item For each rule $\Rule\in\Sigma$ and each substitution $\sigma$,
there is $i\geq 1$ such that $\Tuple{\Rule, \sigma}$ is not applicable to the label of any node of depth $\geq i$.
%
\end{enumerate}
The \emph{result} that corresponds to a chase tree is the set of all interpretations that
can be obtained as the union of all interpretations along a path in the tree.
%
\end{definition}

Condition (3) ensures fair, exhaustive rule application, but different orders of
application can lead to different chase trees, which can also have different results.
Nevertheless, every result is semantically correct in the following sense:

\begin{fact}\label{fact_chase-univ-model-sets}
Every result of a chase on a rule set $\Sigma$ and database $\mathcal{D}$ is a universal model set for $\Sigma$ and $\mathcal{D}$.
\end{fact}

The pair $\Tuple{\Sigma,\mathcal{D}}$ is \emph{chase-terminating} if all its chase trees are finite -- by König's Lemma, this is equivalent to 
all chase results for $\Tuple{\Sigma,\mathcal{D}}$ containing 
only finite interpretations; this corresponds to \emph{all-strategy termination}. 
$\Sigma$ is \emph{chase-terminating} if $\Tuple{\Sigma, \Database}$ is chase-terminating for every database $\mathcal{D}$;
 this corresponds to \emph{universal termination}.

\paragraph*{Turing Machines}

We will use (deterministic) Turing machines (TM), denoted as a tuple $\TM=\Tuple{\States, \Alphabet, \TransitionFunction}$,
with states $\States$, tape alphabet $\Alphabet$ with blank $\Blank\in\Alphabet$, and transition function $\TransitionFunction$.
$\TM$ has a distinguished initial state $\StartingState\in\States$, and accepting and rejecting halting states
$\AcceptingState,\RejectingState\in\States$. For all states $q\in\States\setminus \{\AcceptingState, \RejectingState\}$
and tape symbols $a\in\Alphabet$,
there is exactly one transition $(q, a)\mapsto (r, b, D) \in \TransitionFunction$.
We assume that TM tapes are unbounded to the right but bounded to the left, and
that TMs will never attempt to move left on the first
position of the tape (this is w.l.o.g., since one can modify any TM to insert a marker at the tape start
to recognise this case).

\section{On the Expressivity of Disjunctive Rules}\label{sec_disRules}

In this section, we show how to express a homomorphism-closed, decidable query with a disjunctive rule set.
This construction will be the basis for finding a chase-terminating set of (deterministic) existential rules.
Throughout this section, $\Query$ is a fixed but arbitrary homomorphism-closed query over signature $\mathcal{S}$,
and $\TM = \Tuple{\States, \Alphabet, \TransitionFunction}$ is a TM that decides $\Query$.

To express $\Query$, we specify five rule sets
$\R_1\subseteq\R_2\subseteq\R_3\subseteq\R_4\subseteq\R_5$ (see Figures~\ref{figure:rule-set-1}--\ref{figure:rule-set-5}),
which will be explained later. We want to show the following result:

\begin{theorem}\label{theo_expressivity}
The set $\R_5$ of disjunctive existential rules expresses the query $\Query$.
\end{theorem}

To show this, we fix an arbitrary database $\D$ over $\mathcal{S}$. Theorem~\ref{theo_expressivity} then follows from 
Fact~\ref{fact_uniModSetCorrect} and the next lemma:

\begin{lemma}\label{lemma:expressivity-universal-model}
There is a universal model set \MS of $\R_5$ and $\D$ such that $\D\in\Query$ iff $\PrGoal\in\mathcal{I}$ for every $\mathcal{I}\in\MS$.
\end{lemma}

The universal model set \MS is a complicated structure that we describe step by step,
by specifying five sets of interpretations -- $\IS_1$, $\IS_2$, $\IS_3$, $\IS_4$, and $\IS_5$ --,
such that
(1) $\IS_i$ is a universal model set of $\R_i, \D$ for each $1\leq i\leq 5$,
(2) $|\IS_1| = |\IS_2| = |\IS_3| = |\IS_4|= |\IS_5|$, and
(3) for each $1\leq i<j\leq 5$ and each $\mathcal{I}\in\IS_i$, there is exactly one $\mathcal{J}\in\IS_j$ with $\mathcal{I}\subseteq\mathcal{J}$. 
%
Lemma~\ref{lemma:expressivity-universal-model} can then be shown using $\MS=\IS_5$. 

Before dwelling into the details of each rule set and universal model, we give an overview of the construction: 
$\R_1$ constructs all possible linear orders over the nulls in $\D$, as well as all possible completions of $\D$ with facts built using these nulls; 
$\R_2\setminus\R_1$ extracts successor relations from the linear orders;  
$\R_3\setminus\R_2$ associates to nulls representations of their positions in successor relations; 
$\R_4\setminus\R_3$ encodes all initial TM configurations corresponding to some linear order and completion, and $\R_5\setminus\R_4$ simulates the run of the TM on these configurations.

\begin{figure}[t]%
\setlength{\abovedisplayskip}{0pt}\setlength{\belowdisplayskip}{0pt}
\begin{align}%
& \to \exists y. \PrFirst(y) \wedge \PrDomain(y) \label{rul_createFirst}\\
& \to \exists z.\PrLast(z) \wedge \PrDomain(z) \label{rul_createLast}\\
\PrP(\Vx) &\textstyle \to \PrFreshIDBIn{\PrP}(\Vx) \wedge \bigwedge_{x \in \Vx} \PrDomain(x) \label{rul_copyPrP}\\
\PrDomain(x) &\to \PrEq(x, x) \label{rul_DomToEq}\\
\PrEq(x, y) &\to \PrEq(y, x)  \label{rul_ReflEq}\\
\PrNeq(x, y) &\to \PrNeq(y, x) \\
\PrR(\Vx) \wedge \PrEq(x_i, y) &\to \PrR(\Vx_{x_i\mapsto y}) \label{rul_EqCongruence}\\
\PrDomain(x) \wedge \PrDomain(y) &\to \PrEq(x, y) \vee \PrNeq(x, y) \label{rul_EqOrNeq}\\
\PrOrder(x,y)\wedge\PrOrder(y,z) & \to \PrOrder(x,z) \label{rul_TransLT}\\
\PrFirst(x)\wedge\PrNeq(x, y) &\to \PrOrder(x,y) \label{rul_FirstLessThan}\\
\PrNeq(x, y)\wedge\PrLast(y) &\to \PrOrder(x,y)\label{rul_LessThanLast}\\
\PrNeq(x, y) &\to \PrOrder(x, y) \vee \PrOrder(y, x) \label{rul_LTchoice}
 \\
\textstyle
\bigwedge_{x \in \Vx} \PrDomain(x) &\to \PrFreshIDBIn{\PrP}(\Vx) \vee \PrFreshIDBNotIn{\PrP}(\Vx)\label{rul_DbCompletion}
\end{align}
\caption{The rule set $\RuleSet_1$, where rules \eqref{rul_copyPrP} and \eqref{rul_DbCompletion} are instantiated
for each $\PrP\in\mathcal{S}$, and rules \eqref{rul_EqCongruence} are instantiated for 
each $\PrR\in\{\PrFirst,\PrLast,\PrEq,\PrNeq,\PrOrder\} \cup \{\PrFreshIDBIn{\PrP},\PrFreshIDBNotIn{\PrP} \mid \PrP\in\mathcal{S}\}$
and $1\leq i\leq\Arity(\PrR)$, and $\Vx_{x_i\mapsto y}$ denotes $\Vx$ with $x_i$ replaced by $y$.}
\label{figure:rule-set-1}
\end{figure}

\begin{figure}[t]%
\setlength{\abovedisplayskip}{0pt}\setlength{\belowdisplayskip}{0pt}
\begin{align}%
\PrFirst(x) &\to \exists u.\PrRoot(u)\wedge\PrLink(x,u) \label{rul_RootCreation}\\
\PrLink(x,v) \wedge \PrOrder(x,z) &\to \exists w.\PrHasChild(v,w)\wedge\PrLink(z,w) \label{rul_NextTreeRep}\\
\PrLast(x) \wedge \PrLink(x,u) &\to \PrLeaf(u) \label{rul_LeafCreation}\\
\PrLink(x,u) \wedge \PrEq(x,y) &\to\PrLink(y,u) \label{rul_RepEq}\\
\textstyle
\PrFreshIDBIn{\PrP}(\vec{x}) \wedge \bigwedge_{i=1}^{|\vec{x}|}& \PrLink(x_i,u_i) \to \PrFreshIDBInTape{\PrP}(\vec{u})\label{rul_makeTreeIn}\\
\textstyle
\PrFreshIDBNotIn{\PrP}(\vec{x}) \wedge \bigwedge_{i=1}^{|\vec{x}|}& \PrLink(x_i,u_i) \to \PrFreshIDBNotInTape{\PrP}(\vec{u})\label{rul_makeTreeNotIn}
\end{align}
\caption{The rule set $\RuleSet_2$ contains $\RuleSet_1$ (see Figure~\ref{figure:rule-set-1}) and all above rules, where
\eqref{rul_makeTreeIn} and \eqref{rul_makeTreeNotIn} are instantiated for each $\PrP\in\mathcal{S}$.
}
\label{figure:rule-set-2}
\end{figure}
\begin{figure*}[t]%
\setlength{\abovedisplayskip}{0pt}\setlength{\belowdisplayskip}{0pt}
\begin{align}%
\PrRoot(u) &\to \exists y_1, y_2 . \PrEncoding(u, y_1, y_2) \wedge \PrCell{0}(y_1) \wedge \PrNext(y_1, y_2) \wedge \PrCell{1}(y_2) \label{rul_encRoot}\\
\PrEncoding(u, y_1, y_\_) \wedge \PrHasChild(u,v) &\to \exists z_1, z_\_ . \PrEncoding(v, z_1, z_\_)  \wedge \PrCopyPlusOne(y_1, y_\_, z_1, z_\_) \label{rul_encChild}\\
\PrCopyPlusOne(y_1, y_2, z_1, z_\_) \wedge \PrCell{0}(y_1) \wedge \PrNext(y_1, y_2) &\to \PrCell{1}(z_1) \wedge \PrNext(z_1, z_\_) \wedge \PrCell{1}(z_\_) \label{rul_CpyPlusOneCell0}\\
\PrCopyPlusOne(y_1, y_2, z_1, z_\_) \wedge \PrCell{1}(y_1) \wedge \PrNext(y_1, y_2) &\to \exists z_2 . \PrCell{0}(z_1) \wedge \PrNext(z_1, z_2) \wedge \PrCell{0}(z_2) \wedge \PrNext(z_2, z_\_) \wedge \PrCell{1}(z_\_) \label{rul_CpyPlusOneCell1}\\
\!\!\PrCopyPlusOne(y_1, y_\_, z_1, z_\_) \wedge \PrCell{0}(y_1) \wedge \PrNext(y_1, y_2) \wedge \PrNext(y_2, y_3)
	&\to \exists z_2. \PrCopy(y_2, y_\_, z_2, z_\_) \wedge \PrCell{1}(z_1) \wedge \PrNext(z_1, z_2) \label{rul_CpyPlusOneCell0Next}\\
\!\!\PrCopyPlusOne(y_1, y_\_, z_1, z_\_) \wedge \PrCell{1}(y_1) \wedge \PrNext(y_1, y_2) \wedge \PrNext(y_2, y_3)
	&\to \exists z_2 . \PrCopyPlusOne(y_2, y_\_, z_2, z_\_) \wedge \PrCell{0}(z_1) \wedge \PrNext(z_1, z_2) \label{rul_CpyPlusOneCell1Next}\\
\PrCopy(y_1, y_2, z_1, z_2) \wedge \PrCell{*}(y_1) \wedge \PrNext(y_1, y_2)
	&\to \PrCell{*}(z_1) \wedge \PrNext(z_1, z_2) \wedge \PrCell{1}(z_2) \label{rul_copyBase}\\
\PrCopy(y_1, y_\_, z_1, z_\_) \wedge \PrCell{*}(y_1) \wedge \PrNext(y_1, y_2) \wedge \PrNext(y_2, y_3)
	&\to \exists z_2 . \PrCopy(y_2, y_\_, z_2, z_\_) \wedge \PrCell{*}(z_1) \wedge \PrNext(z_1, z_2) \label{rul_copyRec}
\end{align}
\caption{The rule set $\RuleSet_3$ contains $\RuleSet_2$ (Figure~\ref{figure:rule-set-2}) and all of the above rules, where \eqref{rul_copyBase} and \eqref{rul_copyRec} are instantiated for each $* \in \{0, 1\}$.
}
\label{figure:rule-set-3}
\end{figure*}

\paragraph{$\mathbf{\R_1}$: Linear Order and Database Completion}

$\R_1$ serves two distinct purposes:
(1) predicates $\PrFirst$, $\PrLast$, $\PrEq$ (``$=$''), $\PrNeq$ (``$\neq$''), and
$\PrOrder$ (``$<$'') encode representations of possible linear orders over nulls in $\D$
(collected in predicate $\PrDomain$); and 
(2) predicates $\PrFreshIDBIn{\PrP}$ and $\PrFreshIDBNotIn{\PrP}$ for each $\PrP\in\mathcal{S}$ explicitly
encode positive and negative (absent) facts in $\D$.
Both purposes require disjunctive reasoning. Possible models include representations of strict, total linear orders (1) and
the exact database completion (2), but also models for collapsed orders and inconsistent completions.
The latter is not problematic since we consider homomorphism-closed queries.

We define the interpretation set $\IS_1$ on the set $\Delta=\EI{\Nulls}{\D}\cup\{\FirstNull,  \LastNull\}$,
for fresh nulls $\FirstNull, \LastNull\notin\EI{\Nulls}{\D}$.
A (partially collapsed) linear order can be represented by an ordered partition $\VT = T_1, \ldots, T_k$ ($k\geq 1$)
of $\Delta$ where $\FirstNull\in T_1$ and $\LastNull\in T_k$.
Let $\FunOrders$ be the set of all such $\VT$, and let $\FunCompletions(\VT)$ 
be the set of all interpretations with nulls $\Delta$ that are set-minimal among the models
of $\RuleSet_1$ and $\D$, and that contain the database
\begin{align*}
\{&\PrFirst(u) \mid u \in T_1\} \cup \{\PrLast(u) \mid u \in T_k\} \cup {}\\ 
\{&\PrEq(t, u) \mid 1 \leq i \leq k; t, u \in T_i\} \cup{} \\\{&\PrOrder(t,\! u),
  \PrNeq(t,\! u), \PrNeq(u,\! t) \,{\mid}\, 1 \,{\leq}\, i \,{<}\, j \,{\leq}\, k; t \,{\in}\, T_i; u \,{\in}\, T_j\}.
\end{align*}
By minimality, each $\mathcal{I}\in\FunCompletions(\VT)$ contains exactly one of 
$\{\PrFreshIDBIn{\PrP}(\Vt), \PrFreshIDBNotIn{\PrP}(\Vt)\}$ for every $\Vt\subseteq\Delta$ with $|\Vt| = \Arity(\PrP)$.

\begin{lemma}\label{lem:univ-model-set-1}
$\IS_1=\bigcup_{\VT \in \FunOrders}\FunCompletions(\VT)$ is a universal model set of $\R_1$ and $\D$.
\end{lemma}
\begin{proof}[Proof Sketch]
We can construct a chase tree for $\R_1$ and $\D$ that prioritises the application of
rules in the order of their appearance in Figure~\ref{figure:rule-set-1}.
The result of that chase $\mathfrak{K}$ is a finite universal model set (Fact~\ref{fact_chase-univ-model-sets})
and we can show that every $\mathcal{K}\in\mathfrak{K}$ is isomorphic to a unique $\mathcal{I}\in\IS_1$.
\end{proof}

Each set $\IS_i$ for some $2 \leq i \leq 5$ is obtained as
$\IS_i = \{\FunInterpretation{i}(\mathcal{I}) \mid \mathcal{I} \in \IS_{i-1}\}$ for a function
$\FunInterpretation{i}$ as defined below.
Every $\mathcal{I}\in\IS_i$ with $1 \leq i \leq 5$ contains a unique interpretation $\FunSeed(\mathcal{I})\in\IS_1$,
and there is a unique ordered partition $\FunOrder(\mathcal{I})\in\FunOrders$ such that
$\FunSeed(\mathcal{I}) \in \FunCompletions(\FunOrder(\mathcal{I}))$.

\paragraph{$\mathbf{\R_2}$: Representative Tree}
The purpose of $\R_2$ is to extract successor relations from the transitive
linear order $\PrOrder$.
Given an ordered partition $\FunOrder(\mathcal{I})=T_1,\ldots,T_k$ of some model $\mathcal{I}$ of $\R_1$,
$\R_2$ constructs a finite tree structure -- defined using predicates $\PrRoot$, $\PrHasChild$ (``child''), and $\PrLeaf$ --
where each path represents a sub-sequence $T_{z_1}, \ldots, T_{z_p}$ of $T_1, \ldots, T_k$ with $z_1=1$ and $z_p=k$.
Unavoidably, some paths will skip some $T_i$, but it suffices for our purposes if one path is complete.
The elements of any set $T_j$ are related to the tree nodes that represent $T_j$ by a predicate $\PrLink$.
Moreover, nodes are related via predicates $\PrFreshIDBInTape{\PrP}$ and $\PrFreshIDBNotInTape{\PrP}$
that reflect the relations for $\PrFreshIDBIn{\PrP}$ and $\PrFreshIDBNotIn{\PrP}$ that hold between the represented elements
(in the database completion of the considered model $\mathcal{I}$ of $\R_1$).

Given $\mathcal{I}$ with $|\FunOrder(\mathcal{I})|=k$ as above,
let $\mathbf{Z}$ be the set of all words $z_1\cdots z_p\in\{1,\ldots,k\}^*$ such that $z_1=1$ and $z_i<z_{i+1}$ for all
$i\in\{1,\ldots,p-1\}$. Moreover, let $\FormatFunction{end}(z_1\cdots z_p):=z_p$.
Using fresh nulls $\{u_w\mid w\in\mathbf{Z}\}$, we define
$\FormatFunction{Tree}(\mathcal{I})$ to be the interpretation
\begin{align*}
&\{\PrRoot(u_1)\} \cup
\{\PrHasChild(u_w,u_{wz})\mid wz\in\mathbf{Z}, 2\leq z\leq k\} \cup{} \\
&\{\PrLeaf(u_w)\mid w\in\mathbf{Z}; \FormatFunction{end}(w)=k\} \cup{}\\
&\{\PrLink(x,u_w)\mid w\in\mathbf{Z}; x\in T_{\FormatFunction{end}(w)}\}.
\end{align*}
Now $\FunInterpretation{2}(\mathcal{I})$ is the least interpretation that 
contains $\mathcal{I}$ and $\FormatFunction{Tree}(\mathcal{I})$, and that 
further satisfies rules \eqref{rul_makeTreeIn} and \eqref{rul_makeTreeNotIn}.
We can extend Lemma~\ref{lem:univ-model-set-1} as follows. The required
universal model set is obtained by any chase that prioritises rule \eqref{rul_RepEq}.

\begin{lemma}\label{lem:univ-model-set-2}
$\IS_2 = \{\FunInterpretation{2}(\mathcal{I}) \mid \mathcal{I} \in \IS_1\}$ is a universal model set of $\R_2$ and $\D$.
\end{lemma}

\begin{figure*}[t]%
\setlength{\abovedisplayskip}{0pt}\setlength{\belowdisplayskip}{0pt}
\begin{align}%
\PrLoad{\PrP}(u,t,\Vv) \wedge \PrFreshIDBInTape{\PrP}(\Vv) &\to\textstyle \exists \Vx, y . \PrCell{\PrP}(t) \wedge \PrNext(t, x_1) \wedge \bigwedge_{i=1}^{|\Vv|} \PrLoadEncoding(v_i, x_i, x_{i+1}) \wedge \PrNext(x_{|\Vv|+1}, y) \wedge \PrReady{\PrP}(u, y, \Vv) \label{rul_loadAllFacts}\\
\PrLoadEncoding(v, x_s, x_e) \wedge \PrEncoding(v, y_1, y_\_) &\to \exists z_1, z_\_ . \PrCell{\Separator}(x_s) \wedge \PrNext(x_s, z_1) \wedge \PrCopy(y_1, y_\_, z_1, z_\_) \wedge \PrNext(z_\_, x_e) \wedge \PrCell{\Separator}(x_e)  \label{rul_encodeFact}
\end{align}
\caption{The rule set $\RuleSet_4$ contains $\RuleSet_3$ (see Figure~\ref{figure:rule-set-3}), the rules from Figure~\ref{figure:rule-set-4-part}, and the above rules instantiated for all $\PrP\in\mathcal{S}$.}
\label{figure:rule-set-4}
\end{figure*}

\begin{figure}[t]%
\setlength{\abovedisplayskip}{0pt}\setlength{\belowdisplayskip}{0pt}
\begin{align}%
%
\PrLeaf(u) \to \exists t.\PrLoad{1}(u,&t,u) \wedge \PrHead{\StartingState}(t) \label{rul_startTape}\\
\PrLoad{\ell}(u, t, \Vv) &\to \PrLoad{\PrP^\ell_1}(u, t, \Vv) \label{rul_firstPredEll}\\
\PrReady{\PrP_{j}^\ell}(u, t, \Vv) &\to \PrLoad{\PrP_{j+1}^{\ell}}(u, t, \Vv) \label{rul_nextPredEll}\\
\PrReady{\PrP_{\PredsPerArity}^\ell}(u, t, \Vv) &\to \PrReady{\ell}(u, t, \Vv) \label{rul_lastPredEll}\\
\begin{split}
\textstyle
\!\!\!\PrReady{\ell}(u, t, \Vv) \wedge \bigwedge_{i=k+1}^{\ell} \PrRoot&(v_i) \wedge \PrHasChild(w,v_k) \\
	\to \PrLoad{\ell}(u, t, &v_1,{\scriptstyle\cdots}, v_{k-1}, w, u,{\scriptstyle\cdots}, u)\!
\end{split} \label{rul_tapeNextVector}\\
\textstyle
\!\!\!\PrReady{\ell}(u, t, \Vv) \,{\wedge}\bigwedge_{i=1}^\ell \PrRoot(v_i) &\to \PrLoad{\ell+1}(u, t, u,{\scriptstyle\cdots},u )\!  \label{rul_tapeNextLevel}\\[0.5mm]
\textstyle
\!\!\!\PrReady{\MaxArity}(u, t, \Vv)  \,{\wedge}\bigwedge_{i=1}^{\MaxArity} \PrRoot(v_i) &\to \PrCell{\Blank}(t)\wedge\PrEndTape(t) \label{rul_endTape}\\
\PrLoad{\PrP}(u, t, \Vv) \wedge \PrFreshIDBNotInTape{\PrP}(\Vv) &\to \PrReady{\PrP}(u, t, \Vv) \label{rul_copyNotInPred}
\end{align}
\caption{Some rules of $\R_4$, to be instantiated for all $1\leq j\leq \PredsPerArity-1$, $1\leq k\leq\ell\leq \MaxArity$, and $\PrP\in\mathcal{S}$.}
\label{figure:rule-set-4-part}
\end{figure}

\paragraph{$\mathbf{\R_3}$: Position Binary Encodings}
The purpose of $\R_3$ is to associate each node in the tree of $\R_2$ with a binary encoding
of its distance from the root (the root starts with ``distance'' $2$ for technical reasons).
Encodings start at the least significant bit and always end in $1$ (i.e., have no leading $0$s).
To simplify upcoming steps, encodings take the form of little TM tapes, represented
by a $\PrNext$-connected chain of nulls with unary predicates $\PrCell{0}$ and $\PrCell{1}$
encoding the symbol at each position. Nodes $u$ relate to the first and last null $t_s$ and $t_e$ of their ``tape''
through facts $\PrEncoding(u,t_s,t_e)$.
Facts $\PrCopy(a_s, a_e, b_s, b_e)$ are used to create a tape between $b_s$ and $b_e$ that contains
a copy of the information on the tape between $a_s$ and $a_e$. Predicate $\PrCopyPlusOne$ is analogous, but
creates a representation of the successor of the number that is copied.

Consider a model $\mathcal{I}$ of $\R_2$ and define the set of sequences $\mathbf{Z}$
as before.
For $w\in\mathbf{Z}$ of length $|w|$,
and $b_1\cdots b_\ell$ the binary representation of $|w|+1$, let $\FormatFunction{EncPos}(w)$ be the database
\begin{align*}
\{&\PrEncoding(u_w, e_w^1, e_w^\ell)\} \cup \{\PrCell{b_i}(e_w^{i}) \mid 1 \leq i \leq \ell \} \cup{} \\
\{&\PrNext(e_w^{i-1}, e_w^i) \mid 2 \leq i \leq \ell\}.
\end{align*}
Let $\mathcal{J}=\mathcal{I}\cup\bigcup_{w\in\mathbf{Z}}\FormatFunction{EncPos}(w)$.
We define $\FunInterpretation{3}(\mathcal{I})$ as the smallest superset
of $\mathcal{J}$ that satisfies all rules in $\R_3$ while including only the nulls in 
$\mathcal{J}$.
$\FunInterpretation{3}(\mathcal{I})$ extends $\mathcal{J}$ only by missing
$\PrCopy$ and $\PrCopyPlusOne$ relations, which can be inferred by slightly rewritten
rules. For example, rule \eqref{rul_encChild} is satisfied when applying the following rule to $\mathcal{J}$:
\begin{align*}
\PrEncoding(u, y_1, y_\_) \wedge \PrHasChild(u,v) &\wedge \PrEncoding(v, z_1, z_\_) \\
&\to \PrCopyPlusOne(y_1, y_\_, z_1, z_\_).
\end{align*}
All other rules can be rewritten analogously, since every existentially quantified variable is used 
in unique ways with predicates other than $\PrCopy$ and $\PrCopyPlusOne$.

%
%
%
%
%
For every $\mathcal{I}\in\IS_2$, we show that $\FunInterpretation{3}(\mathcal{I})$ is isomorphic to a result of the chase on $\R_3$ and $\mathcal{I}$. The next result then follows from Lemma~\ref{lem:univ-model-set-2}.

\begin{lemma}\label{lem:univ-model-set-3}
$\IS_3 = \{\FunInterpretation{3}(\mathcal{I})\mid\mathcal{I}\in\IS_2\}$ is a universal model set of $\R_3$ and $\D$.
\end{lemma}

\paragraph{$\mathbf{\R_4}$: Initial TM Configuration}
For each leaf in the tree of completions, $\R_4$ creates the representation of an
initial TM configuration. The tape is again represented by a $\PrNext$-chain,
using further unary predicates $\PrCell{\Separator}$, $\PrCell{\Blank}$, and $\PrCell{\PrP}$ (for all $\PrP\in\mathcal{S}$)
for additional tape symbols. 
$\PrHead{\StartingState}$ marks the TM's starting position and initial state $\StartingState$, and $\PrEndTape$ the end of the tape.

Let $\MaxArity$ be the maximal arity of predicates in $\mathcal{S}$.
We require that there is some $\PredsPerArity>0$ such that
$\mathcal{S}$ contains exactly $\PredsPerArity$ predicates $\PrP^i_1, \ldots, \PrP^i_\PredsPerArity$ of arity $i$, for every $1\leq i\leq \MaxArity$.
This is without loss of generality, except for the exclusion of nullary predicates.
Our results do not depend on this restriction, but it helps to simplify the presentation of our main ideas.

To serialise the data as a tape, we iterate over all predicate arities $\ell=1,\ldots,\MaxArity$ and over
all lists $\vec{v}$ of tree nodes with length $\ell$.
In this process, $\PrLoad{\ell}(u, t, \vec{v})$ expresses that, while encoding the leaf $u$, after constructing the tape
until position $t$, we continue serialising $\ell$-ary predicate data for arguments $\vec{v}$.
Analogously, $\PrReady{\ell}(u, t, \vec{v})$ means that this was completed at tape position $t$.
Similar predicates $\PrLoad{\PrP}$ and $\PrReady{\PrP}$ are used to consider a specific predicate $\PrP\in\mathcal{S}$
during this process. The rules in Figure~\ref{figure:rule-set-4-part} start the serialisation \eqref{rul_startTape},
proceed over all predicates \eqref{rul_firstPredEll}--\eqref{rul_lastPredEll},
iterate over parameter vectors \eqref{rul_tapeNextVector} and arities \eqref{rul_tapeNextLevel},
and finally end the tape \eqref{rul_endTape}.

Absent facts do not need to be serialised \eqref{rul_copyNotInPred},
while present facts can be treated by copying the encodings for each of their parameters \eqref{rul_loadAllFacts} and \eqref{rul_encodeFact}.
In the latter, $\PrLoadEncoding$ states that a specific argument is serialised between two given tape positions.

The resulting TM tapes serialise facts $\PrFreshIDBInTape{\PrP}(\vec{u})$ as introduced by $\R_2$, i.e., where $\vec{u}$ 
are nodes in the representative tree.
Given a model $\mathcal{I}\in\IS_3$ with some $\PrLeaf(u_w)\in\mathcal{I}$,
let $\FF{branch}(u_w)$ be the set of all nodes $u_{w'}$ on the branch of $u_w$, i.e.
all nulls $u_{w'}$ where $w'$ is a prefix of $w$.
Elements of $\FF{branch}(u_w)$ are totally ordered by setting $u_{w_1}\prec u_{w_2}$ if $|w_1|>|w_2|$.
Predicates are totally ordered by setting $\PrP^a_i\prec\PrP^b_j$ if either $a<b$, or both $a=b$ and $i<j$.
We can then order facts as $\PrP(\vec{u})\prec\PrQ(\vec{v})$ if $\vec{u},\vec{v}\subseteq\FF{branch}(u_w)$
and $\Tuple{\PrP,\vec{u}}$ is lexicographically before $\Tuple{\PrQ,\vec{v}}$.

Now let $\FF{branchDb}(\mathcal{I},u_w)=\{\PrP(\vec{u})\mid \PrFreshIDBInTape{\PrP}(\vec{u})\in\mathcal{I}, \vec{u}\subseteq\FF{branch}(u_w)\}$
denote the set of all facts on the branch with leaf $u_w$,
and let $\FF{branchTape}(\mathcal{I},u_w)$ denote the TM tape serialisation (as defined in Section~\ref{sec_prelims})
of $\FF{branchDb}(\mathcal{I},u_w)$ according to the total order $\prec$ and representing each
node $u_w$ by the binary representation of $|w|+1$ as before.
Given $S=\FF{branchTape}(\mathcal{I},u_w)$, let $\FF{startConf}(\mathcal{I},u_w)$ be the following interpretation:
\begin{align*}
\{&\PrLoad{1}(u_w, t_w^1, u_w), \PrHead{\StartingState}(t^1_w), \PrEndTape(t_w^{|S|+1})\} \cup{} \\
\{&\PrNext(t_w^{j-1}, t_w^{j}) \mid 2\leq j \leq |S|+1\} \cup{} \\
%
\{&\PrCell{a}(t_w^{j}) \mid  1\leq j\leq |S|, a=S[j]\}\cup\{\PrCell{\Blank}(t_w^{|S|+1})\}.
\end{align*}
Let $\mathcal{J}$ be the extension of $\mathcal{I}$ with $\FF{startConf}(\mathcal{I},u_w)$
for every $\PrLeaf(u_w)\in\mathcal{I}$.
We define $\FunInterpretation{4}(\mathcal{I})$ to be the smallest superset of $\mathcal{J}$
that satisfies all rules in $\R_4$ while including only the nulls in 
$\mathcal{J}$.
As in the case of $\FunInterpretation{3}$, the missing relations can easily be inferred
using the original rules or, for \eqref{rul_loadAllFacts} and \eqref{rul_encodeFact},
with simple rewritings thereof.

\begin{lemma}\label{lem:univ-model-set-4}
$\IS_4 = \{\FunInterpretation{4}(\mathcal{I}) \mid \mathcal{I}\in\IS_3\}$ is a universal model set of $\R_4$ and $\D$.
\end{lemma}

\begin{figure}
\setlength{\abovedisplayskip}{0pt}
\setlength{\belowdisplayskip}{4pt} 
\begin{align}
\PrHead{\AcceptingState}(x) & \to \PrGoal \label{rul_tm_accept}\\
\PrNext(x, y) &\to \PrNextTrans(x, y) \label{rul_tm_nexttrans_init}\\
\PrNextTrans(x, y) \wedge \PrNextTrans(y, z) &\to \PrNextTrans(x, z) \label{rul_tm_nexttrans_rec}\\
\PrNext(x, y) \wedge \PrStep(x, z) \wedge \PrStep(y, w) &\to \PrNext(z, w) \label{rul_tm_new_nxt}
\end{align}
\begin{align}
\!\!\!\!\PrEndTape(x)\,{\wedge}\,\PrStep(x, z) &\to \exists v.\PrNext(z, v) \,{\wedge}\, \PrCell{\Blank}(v) \,{\wedge}\, \PrEndTape(v)
\label{rul_tm_add_mem}
\end{align}
\begin{align}
\PrHead{q}(x) \wedge \PrCell{a}(x) &\to \exists z.\PrStep(x, z)\,{\wedge}\,\PrCell{b}(z) \label{rul_tm_new_symb}\\
%
\!\!\PrHead{q}(x) \,{\wedge}\, \PrNextTrans(x, y) \,{\wedge}\, \PrCell{c}(y) &\to \exists z.\PrStep(y, z) \,{\wedge}\,  \PrCell{c}(z) \label{rul_tm_right_mem}\\
\!\!\PrHead{q}(x) \,{\wedge}\, \PrNextTrans(y, x) \,{\wedge}\, \PrCell{c}(y) &\to \exists z.\PrStep(y, z) \,{\wedge}\,  \PrCell{c}(z)
\label{rul_tm_left_mem}
%
\end{align}
\begin{align}
\PrHead{q}(x) \wedge \PrCell{a}(x) \wedge \PrStep(x, z) \wedge \PrNext(z, w) &\to \PrHead{r}(w) \label{rul_tm_move_right}\\
%
\PrHead{q}(x) \wedge \PrCell{a}(x) \wedge \PrStep(x, z) \wedge \PrNext(w, z) &\to \PrHead{r}(w) \label{rul_tm_move_left}
%
\end{align}
\caption{The rule set $\RuleSet_5$ contains $\RuleSet_4$ (see Figure~\ref{figure:rule-set-4}) and the above rules,
where we instantiate rules \eqref{rul_tm_new_symb}--\eqref{rul_tm_left_mem} for all transitions $(q, a) \mapsto (r, b, X) \in \TransitionFunction$
and $c\in\Alphabet$;
rule \eqref{rul_tm_move_right} for all $(q, a) \mapsto (r, b, +1) \in \TransitionFunction$; and 
rule \eqref{rul_tm_move_left} for all $(q, a) \mapsto (r, b, -1) \in \TransitionFunction$.}
\label{figure:rule-set-5}
\end{figure}

\paragraph{$\mathbf{\R_5}$: TM Run}
The purpose of $\R_5$ is to simulate the run of the deterministic TM
$\Tuple{\States, \Alphabet, \TransitionFunction}$
on each of the initial tapes created by $\R_4$.
We continue to use predicate $\PrNext$ for neighbouring tape cells (augmented with its transitive
closure $\PrNextTrans$), $\PrCell{b}$ to encode tape symbols $b\in\Alphabet$, and
$\PrHead{q}$ to encode head position and current state $q\in\States$.
Predicate $\PrStep$ connects tape cells in each configuration to the
corresponding tape cells in the next configuration (provided the TM performs another step).
The rules in Figure~\ref{figure:rule-set-5} are a standard TM encoding, with the slight
exception of rule \eqref{rul_tm_add_mem}, which adds a new blank tape cell in each step (even if not used by the TM).
Our rules use the assumptions on TMs in Section~\ref{sec_prelims}.

Consider some $\mathcal{I}\in\IS_4$.
It is easy and only mildly laborious to define interpretations $\FF{Run}(u_w)$
that represent all successor configurations of the starting configuration
$\FF{startConf}(\mathcal{I},u_w)$, appropriately connected with $\PrStep$ and the
transitive closure $\PrNextTrans$.
Moreover, let $\mathcal{J}$ be the extension of $\mathcal{I}$ with
all $\PrNextTrans$ required to satisfy \eqref{rul_tm_nexttrans_init} and \eqref{rul_tm_nexttrans_rec}
(note that $\PrNext$ also occurs in encodings from $\R_3$).
We define $\FunInterpretation{5}(\mathcal{I})$ as the union of
$\mathcal{J}$ with the interpretations $\FF{Run}(u_w)$ for all 
$u_w$ with $\PrLeaf(u_w)\in\mathcal{I}$.

\begin{lemma}\label{lem:univ-model-set-5}
$\IS_5= \{\FunInterpretation{5}(\mathcal{I}) \mid \mathcal{I}\in\IS_4\}$ is a universal model set of $\R_5$ and $\D$.
\end{lemma}

\paragraph{Proving Lemma~\ref{lemma:expressivity-universal-model}}
To complete the proof of Lemma~\ref{lemma:expressivity-universal-model}, 
we set $\MS = \IS_5$. 
For $\mathcal{I}\in\MS$, let $\FunDatabase(\mathcal{I})=\{\PrP(\Vt)\mid \PrFreshIDBIn{\PrP}(\Vt)\in\mathcal{I}\}$ denote
the completed database created by $\R_1$.
Due to rule \eqref{rul_copyPrP} in Figure~\ref{figure:rule-set-1}, there is a
homomorphism $\D\to\FunDatabase(\mathcal{I})$.
Moreover, the representation tree constructed for $\mathcal{I}$ by $\R_2$ has a branch that is
maximal, i.e., has $|\FunOrder(\mathcal{I})|$ nodes; this branch has a leaf $u_w$ with $|w|=|\FunOrder(\mathcal{I})|$. 
We obtain a homomorphism $\FunDatabase(\mathcal{I})\to\FF{branchDb}(\mathcal{I},u_w)$.
Lemma~\ref{lemma:expressivity-universal-model} now follows from Lemma~\ref{lem:univ-model-set-5} and
Lemmas \ref{lemma:completeness} and \ref{lemma:soundness} below.

\begin{lemma}\label{lemma:soundness}
If $\D\in\Query$, then $\PrGoal\in\mathcal{I}$ for each $\mathcal{I}\in\MS$.
\end{lemma}
\begin{proof}
As shown above, there is a homomorphism $\D\to\FF{branchDb}(\mathcal{I},u_w)$
for the node $u_w$ where $|w|=|\FunOrder(\mathcal{I})|$.
Since $\Query$ is closed under homomorphisms, 
$\D\in\Query$ implies $\FF{branchDb}(\mathcal{I},u_w)\in\Query$.
By the correctness of our TM simulation, we obtain $\PrGoal\in\mathcal{I}$.
\end{proof}

\begin{lemma}
\label{lemma:completeness}
If $\D\notin\Query$, then $\PrGoal\notin\mathcal{I}$ for some $\mathcal{I}\in\MS$.
\end{lemma}
\begin{proof}
Consider some $\mathcal{I}\in\MS$ such that $\FunDatabase(\mathcal{I})=\D$
and $\PrNeq(t, u) \in \mathcal{I}$ for each $t, u \in\EI{\Nulls}{\Database}$ with $t \neq u$.
Let $u_w$ denote the leaf node with $|w|=|\FunOrder(\mathcal{I})|$ as before.
Then $\FF{branchDb}(\mathcal{I},u_w)$ is isomorphic to $\FunDatabase(\mathcal{I})=\D$.
By the correctness of our TM simulation, $\PrGoal$ is not derived from this maximal branch.
Moreover, for all other leaf nodes $u_v$ with $\PrLeaf(u_v)\in\mathcal{I}$, there is a
homomorphism $\FF{branchDb}(\mathcal{I},u_v)\to\FF{branchDb}(\mathcal{I},u_w)$.
Since $\Query$ is closed under homomorphisms, the TM does not accept any such $\FF{branchDb}(\mathcal{I},u_v)$,
so $\PrGoal\notin\mathcal{I}$.
\end{proof}

\section{Ensuring Chase Termination}\label{sec_brake}

While the rules in Section~\ref{sec_disRules} are semantically correct, the disjunctive chase may
not terminate on them.
Many known fragments of existential rules can guarantee chase termination, including for expressive cases
where termination might be exponential \cite{DBLP:conf/ijcai/CarralDKL19}, but they are not applicable
to our case, since the runtime of TMs that decide a query can in general not be bounded by any elementary function.
Indeed, we rely on the TM to stop ``naturally'', by virtue of being a decider.
Nevertheless, our rules lead to infinite chase trees, e.g., if the disjunctive guessing of
$\PrOrder$ leads to a cycle, which enables rule \eqref{rul_NextTreeRep} to create an infinite
path in the representation tree. We will now show that this can be avoided:

\begin{theorem}\label{theo_expressivity_terminating}
Every homomorphism-closed decidable query is expressed by a
set of disjunctive rules that is chase-terminating for all databases over the schema of the query.
\end{theorem}

To show this, we refine and generalise the ``emergency brake'' technique of
\citeauthor{DBLP:conf/icdt/KrotzschMR19} (\citeyear{DBLP:conf/icdt/KrotzschMR19}), and
re-formulate it as a general rule set transformation.
This not only yields a generic method that is of independent interest, but it also allows
us to address potential termination problems in our prior modelling.

\begin{definition}\label{def_brake}
Consider a rule set $\Sigma$ and a nullary predicate $\PrHalt$ that does not occur in $\Sigma$.
For every predicate $\PrP$ in $\Sigma$, let $\hat{\PrP}$ be a fresh predicate of the same arity,
and, for any formula $\psi$, let $\hat{\psi}$ be $\psi$ with all predicates $\PrP$ replaced by $\hat{\PrP}$.
%
%
%
%
Now the set $\FF{brake}(\Sigma,\PrHalt)$ consists of the following rules:
\begin{align}
 &\textstyle\to\exists v.\PrBrake(v) \label{rul_brake_init}
 \\
 \PrHalt\wedge\PrBrake(x)  &\to \PrReal(x)\label{rul_brake_halt}\\
\textstyle\hat{\PrP}(\vec{x})\wedge\bigwedge_{x\in\vec{x}} \PrReal(x) & \to \PrP(\vec{x}) \qquad \text{for all $\PrP$ in $\Sigma$}
\label{rul_brake_real}
\end{align}
For every rule $\rho: \Body[\vec{x}]\to\bigvee_{i=1}^k \exists\vec{y}_i.\Head_i[\vec{x}_i,\vec{y}_i]$:
\begin{align}
\begin{split}
\Body[\vec{x}]\,{\wedge}\,\PrBrake(v)
	&\textstyle\to \bigvee_{i=1}^k  \big(\PrBrakeBody_\rho^i(\vec{x}_i)\,{\wedge}\,\hat{\Head}_i[\vec{x_i},\vec{y}_i\mapsto v]
				\,{\wedge}\\
	&\textstyle\phantom{{}\to \bigvee_{i=1}^k (}\bigwedge_{x\in\vec{x}_i} \PrReal(x)\big) 
\end{split}\label{rul_brake_body}\\
\PrBrakeBody_\rho^i(\vec{x}_i) & \textstyle\to \exists\vec{y}_i.\hat{\Head}_i[\vec{x_i},\vec{y_i}]\wedge\bigwedge_{y\in\vec{y}_i} \PrReal(y)
\label{rul_brake_head}
\end{align}
where $\hat{\Head}_i[\vec{x_i},\vec{y}_i\mapsto v]$ is $\hat{\Head}_i$ with each variable $y\in\vec{y}_i$ replaced by $v$,
and $\PrBrake$, $\PrReal$, and all $\PrBrakeBody_\rho^i$ are fresh predicates with arities as indicated.
\end{definition}


Note that $\FF{brake}(\Sigma,\PrHalt)$ does not define rules to derive $\PrHalt$, and indeed the transformation largely
preserves the models of $\Sigma$ in the following sense:

\begin{lemma}\label{lemma_brake_conservative}
Consider a rule set $\Sigma$ and database $\D$ over predicates that occur in $\Sigma$.
For every model $\mathcal{I}$ of $\FF{brake}(\Sigma,\PrHalt)$ and $\D$, the set $\mathcal{I}^-=\{\PrP(\vec{n})\mid \PrP(\vec{n})\in\mathcal{I}, \PrP\text{ occurs in }\Sigma\}$ is a model of $\Sigma$ and $\D$, and every model $\mathcal{J}$ of $\Sigma$ and $\D$ is of this form.
\end{lemma}
\begin{proof}
Consider a rule $\rho\in\Sigma$ as in Definition~\ref{def_brake}, and let $\sigma$ be a
substitution such that $\sigma(\Body)\subseteq\mathcal{I}^-$. Then we can apply rules
\eqref{rul_brake_body}, \eqref{rul_brake_head}, and finally \eqref{rul_brake_real}
to derive $\sigma'(\Head_i)\subseteq\mathcal{I}^-$ for a suitable extension $\sigma'$ of $\sigma$.
Hence $\mathcal{I}^-\models\Sigma$.

Conversely, let $\mathcal{J}\models\Sigma$. A model $\mathcal{I}$ of $\FF{brake}(\Sigma,\PrHalt)$
can be found by adding, for each matching body $\sigma(\Body)\subseteq\mathcal{J}$ of rule $\rho$,
an atom $\sigma(\PrBrakeBody_\rho^i(\vec{x}_i))$ for some $i$ such that $\sigma'(\Head_i)\subseteq\mathcal{J}$
for an extension $\sigma'$ of $\sigma$. To obtain the required model $\mathcal{I}$ of $\FF{brake}(\Sigma,\PrHalt)$,
it remains to add facts $\PrBrake(b)$ for a fresh null $b$, 
$\sigma(\hat{\Head}_i[\vec{x_i},\vec{y}_i\mapsto b])$ as in \eqref{rul_brake_body} for every 
$\sigma(\PrBrakeBody_\rho^i(\vec{x}_i))\in\mathcal{I}$, and
$\PrReal(n)$ for every $\PrP(\vec{n})\in\mathcal{J}$ and $n\in\vec{n}$.
\end{proof}

For $\FF{brake}(\Sigma,\PrHalt)$ to be useful, we need to add rules that can ``pull the brake'' by deriving
$\PrHalt$. Doing so stops the chase in the following sense:

\begin{lemma}\label{lemma_brake_termination}
Consider a rule set $\Sigma$, a database $\D$ over predicates that occur in $\Sigma$, 
and a set $\Pi$ of rules of the form $\Body\to\PrHalt$ where $\Body$ only uses predicates in $\Sigma$.
If $\mathcal{I}$ is the label of a node in a chase tree for $\Sigma\cup\Pi$ and $\D$ such 
that $\PrHalt\in\mathcal{I}$, then the tree starting at the node of $\mathcal{I}$ is finite.
\end{lemma}
\begin{proof}
Since $\PrHalt\in\mathcal{I}$, there is a substitution $\sigma$ such that $\Tuple{\rho_{\eqref{rul_brake_halt}},\sigma}$ is
applicable (for 
$\rho_{\eqref{rul_brake_halt}}$ in \eqref{rul_brake_halt}).
By fairness, $\PrReal(\sigma(x))$ will be derived at some depth 
of the tree.
From this depth on, no rule of form \eqref{rul_brake_head} is applicable:
given $\PrReal(\sigma(x))$, the head of rules of form \eqref{rul_brake_body} already satisfies
the head of the rule \eqref{rul_brake_head} that could be applied to a newly derived atom for $\PrBrakeBody_\rho^i$.
Rules other than \eqref{rul_brake_head} do not contain existential quantifiers 
thus can 
only
be applied a finite number of times before the chase on this part of the tree terminates.
\end{proof}

If $\PrHalt$ is derived, the semantic correspondence of Lemma~\ref{lemma_brake_conservative} is weakened,
but 
suffices to preserve entailments:

\begin{lemma}\label{lemma_brake_preservation}
Consider $\Sigma$, $\D$, and $\Pi$ as in Lemma~\ref{lemma_brake_termination}.
For every model $\mathcal{I}$ of $\FF{brake}(\Sigma,\PrHalt)\cup\Pi$ and $\D$, 
$\mathcal{I}^-$ (as in Lemma~\ref{lemma_brake_conservative}) is a model of $\Sigma$ and $\D$.
\end{lemma}
\begin{proof}
This is immediate from Lemma~\ref{lemma_brake_conservative} and the fact that every model of
$\FF{brake}(\Sigma,\PrHalt)\cup\Pi$ and $\D$ is also a model of $\FF{brake}(\Sigma,\PrHalt)$ and $\D$.
\end{proof}

Having established the key properties of the emergency brake construction, we can now apply
it to show Theorem~\ref{theo_expressivity_terminating}.
Given the rule set $\R_5$ as defined for a query $\Query$ in Section~\ref{sec_disRules},
let $\R_6$ denote the extension of $\FF{brake}(\R_5,\PrHalt)$ with the following rules:
\begin{align}
\PrFreshIDBIn{\PrP}(\Vx) \wedge \PrFreshIDBNotIn{\PrP}(\Vx)&\to \PrHalt \label{rul_halt_inconsistentDb}\\
\PrOrder(x, x) &\to \PrHalt \label{rul_halt_cycle}\\
\PrLast(x) \wedge \PrOrder(x, y) &\to \PrHalt \label{rul_halt_afterLast}\\
\PrOrder(x, y) \wedge \PrFirst(y) &\to \PrHalt \label{rul_halt_beforeFirst}
\end{align}

\begin{lemma}\label{lemma_termDisjRules_correct}
$\R_6$ expresses the query $\Query$.
\end{lemma}
\begin{proof}
For a database $\D$ over $\mathcal{S}$, let $\MS$ be the universal model set constructed in Section~\ref{sec_disRules}.
If $\D\in\Query$, then $\R_5,\D\models\PrGoal$ by Theorem~\ref{theo_expressivity}. Then
$\R_6,\D\models\PrGoal$ since any model of $\R_6$ and $\D$ must contain $\PrGoal$ by Lemma~\ref{lemma_brake_preservation}.

Conversely, if $\D\notin\Query$, then there is $\mathcal{U}\in\MS$ with $\PrGoal\notin\mathcal{U}$.
By Lemma~\ref{lemma_brake_conservative}, there is a model $\mathcal{I}$ of $\FF{brake}(\R_5,\PrHalt)$
with $\mathcal{I}^-=\mathcal{U}$, and hence $\PrGoal\notin\mathcal{I}$.
By construction of $\MS$, none of the rules \eqref{rul_halt_inconsistentDb}--\eqref{rul_halt_beforeFirst} applies to
$\mathcal{U}$, and hence $\mathcal{I}$ is also a model of $\R_6$, i.e., $\R_6,\D\not\models\PrGoal$.
\end{proof}

\begin{lemma}\label{lemma_termDisjRules_terminating}
$\R_6$ is chase-terminating for all databases over the schema $\mathcal{S}$ of the query \Query.
\end{lemma}
\begin{proof}
Consider a chase over $\R_6$ and input database $\D$.
Chase branches where $\PrHalt$ is eventually part of a node label terminate by Lemma~\ref{lemma_brake_termination}.
Let $b$ denote any branch of the chase where $\PrHalt$ is not derived,
and let $\mathcal{I}$ be the union of all node labels on that branch.
We want to show that $\mathcal{I}$ (and hence $b$) is finite.

By Lemma~\ref{lemma_brake_conservative}, $\mathcal{I}\models \R_1$.
Moreover, since $\PrHalt\notin\mathcal{I}$, rules \eqref{rul_halt_inconsistentDb}--\eqref{rul_halt_beforeFirst} are
not applicable to $\mathcal{I}$. Both properties together suffice to show that the set
$\mathcal{I}_{\R_1}=\{\PrP(\vec{n})\in\mathcal{I}\mid \PrP$ is a predicate in $\R_1\}$
is an element of $\IS_1$ defined in Section~\ref{sec_disRules}.

Since the predicates in rule bodies of $\R_1$ do not occur in any rule head in $\R_5\setminus\R_1$,
we can assume without loss of generality (and without affecting chase termination), that the corresponding rules
of $\FF{brake}(\R_1,\PrHalt)\subseteq \R_6$ have been applied first.
This shows that $\mathcal{I}$ is equal to the result of a chase with non-disjunctive rules
$\R_6\setminus \FF{brake}(\R_1,\PrHalt)$ on a database $\mathcal{I}_{\R_1}\in\IS_1$.
The claim follows by noting that any such chase must terminate:
this was shown in Section~\ref{sec_disRules}, where we described a deterministic
process of defining the elements in the universal model set $\IS_5$ from those in $\IS_1$.
Each steps in this construction is fully determined and introduces isomorphic sets of nulls irrespectively of
the order of rule applications.
The only exception are application of rules \eqref{rul_RootCreation}, \eqref{rul_NextTreeRep}, and \eqref{rul_RepEq}.
For example, given facts $\PrFirst(n_1)$, $\PrFirst(n_2)$, and $\PrEq(n_1,n_2)$, the standard model
of Section~\ref{sec_disRules} contains one fact $\PrRoot(u_{1})$ with $\PrLink(n_1,u_{1})$ and $\PrLink(n_2,u_{1})$,
which can be obtained using \eqref{rul_RootCreation} (on $n_1$) and \eqref{rul_RepEq}.
If we apply \eqref{rul_RootCreation} to both $n_1$ and $n_2$ before applying \eqref{rul_RepEq},
we obtain two distinct $\PrRoot(u_{1})$ and $\PrRoot(u_{1}')$. Similar variations can occur with other tree nodes
if \eqref{rul_NextTreeRep} is applied before \eqref{rul_RepEq}. If is easy to see that this does not endanger termination,
but merely leads to several isomorphic paths in the representation tree.
%
\end{proof}

Together, Lemmas~\ref{lemma_termDisjRules_correct} and \ref{lemma_termDisjRules_terminating} show Theorem~\ref{theo_expressivity_terminating}.

\section{Removing Disjunctions}\label{sec:disjunction removal}
\begin{figure*}
\setlength{\belowdisplayskip}{0pt}\setlength{\abovedisplayskip}{0pt}
\begin{align}
\label{world-init}
  &\rightarrow \exists w.\PrInit(w) \wedge \PrDone(w) \wedge \PrEmpty(w) \\
\PrDone(w) \wedge \PrInit(w) \wedge \PrP(\vec{x})  &\rightarrow \exists w'.\PrIns{\PrP}(\vec{x},w,w') 
\wedge \PrSubs(w',w') \wedge \PrInit(w')
\label{collect-p}\\
\textstyle\PrDone(w) \wedge \bigwedge_{\PrP(\vec{x}) \in\Body  } \PrIns{\PrP}(\vec{x},w,w) &\rightarrow \exists w_1.\PrIns{\PrP_1}(\vec{x_1},w,w_1) 
\wedge \PrSubs(w_1,w_1)
 \label{complete-r1}\\
\textstyle \PrDone(w) \wedge \bigwedge_{\PrP(\vec{x}) \in \Body } \PrIns{\PrP}(\vec{x},w,w) &\rightarrow \exists w_2.\PrIns{\PrP_2}(\vec{x_2},w,w_2) 
\wedge \PrSubs(w_2,w_2)
\label{complete-r2} \\
 \PrIns{\PrP}(\vec{x},w_0,w_1) \wedge \PrSubs(w_1,w_2) &\rightarrow \PrIns{\PrP}(\vec{x},w_2,w_2) \wedge \PrHatted{\PrP}(\vec{x},w_2) 
 \wedge \PrSubs(w_0,w_2) 
\label{propagate-p} \\
\label{rule-done}
 \PrEmpty(w) \wedge \PrSubs(w,w') &\rightarrow \PrDone(w')
\end{align}
\caption{The rule set $\Sigma'_1$, where we instantiate \eqref{collect-p} and \eqref{propagate-p} for all $\PrP \in \Predicates$ in $\Sigma_1\cup\Sigma_2$, and \eqref{complete-r1} and \eqref{complete-r2} for all $\Body\rightarrow \PrP_1(\vec{x_1}) \vee \PrP_2(\vec{x_2})\in\Sigma_1$.}
\label{figure:rule-set-disjunction-removal-r1}
\end{figure*}
\begin{figure}
\setlength{\belowdisplayskip}{0pt}\setlength{\abovedisplayskip}{0pt}
\begin{align}
\setlength{\abovedisplayskip}{0pt}
\setlength{\belowdisplayskip}{4pt}
\PrDone(w) \wedge \textstyle\bigwedge_{\PrP(\vec{x}) \in \Body  } \PrHatted{\PrP}(\vec{x},w) &\textstyle\rightarrow \exists \vec{z}.\bigwedge_{\PrQ(\vec{y}) \in \Head } \PrHatted{\PrQ}(\vec{y},w)
\label{rule-sigma-2}
\end{align}
\caption{The rule set $\Sigma'_2$, where we instantiate 
\eqref{rule-sigma-2} for all $\Body   \rightarrow \exists \vec{z}. \Head \in \Sigma_2$.
}
\label{figure:rule-set-disjunction-removal-r2}
\end{figure}
\begin{figure}
\setlength{\belowdisplayskip}{0pt}\setlength{\abovedisplayskip}{0pt}
\begin{align}
\label{goal-rule-disj}
\PrHatted{\PrGoal}(w) &\to \PrWorldGoal(w)\\	
\begin{split}
\label{propagate-ri}
\PrIns{\PrP_1}(\vec{x}_1,w,w_1)\wedge \PrWorldGoal(w_1) \wedge{} \\
\PrIns{\PrP_2}(\vec{x}_2,w,w_2)\wedge \PrWorldGoal(w_2) \wedge{}\\
 \textstyle\bigwedge_{\PrP(\vec{x}) \in \Body  }  \PrIns{\PrP}(\vec{x},w,w) 
  &\to
  \PrWorldGoal(w)
\end{split}\\
%
\label{final-rule-disj}
\PrInit(w) \wedge \PrWorldGoal(w) &\to \PrGoal 
\end{align}
\caption{The rule set $\Sigma'_3$, where \eqref{propagate-ri} is instantiated for all rules $\Body\rightarrow \PrP_1(\vec{x_1}) \vee \PrP_2(\vec{x_2})\in\Sigma_1$}
\label{figure:rule-set-disjunction-removal-r3}
\end{figure}

Our main result is that any decidable homomorphism-closed query is expressible by a  chase-terminating existential rule set.
To conclude the proof of this statement, we remove the disjunction from the rule set $\R_6$ of Section~\ref{sec_brake}.
We present this as a general technique of expressing disjunctive Datalog using existential rules,
which is also of independent interest.

For a rule set $\Sigma$, the \emph{input schema} $\mathcal{S}_{\mathsf{in}}(\Sigma)$ is the set of 
all predicates in $\Sigma$ that do not occur in any rule head.
We focus on rule sets that can be split into a disjunctive part and an existential part,
such that it is admissible to completely apply the disjunctive rules first, and the existential
ones afterwards:

\begin{definition}\label{def_split}
A \emph{split} of a set $\Sigma$ of disjunctive existential rules
consists of a set $\Sigma_1$ of disjunctive Datalog rules and
a set $\Sigma_2$ of existential rules, such that $\Sigma=\Sigma_1\cup\Sigma_2$ and:
\begin{quote}
For every database $\D$ over $\mathcal{S}_{\mathsf{in}}(\Sigma)$,
and for every chase result $\MS$ over $\Tuple{\Sigma,\D}$,
there is a chase result $\MS_1$ over $\Tuple{\Sigma_1,\D}$, such that
$\MS=\{\mathcal{C}_{\mathcal{I}}\mid \mathcal{I}\in\MS_1\}$ where each
$\mathcal{C}_{\mathcal{I}}$ is the (unique) interpretation resulting from some
chase over $\Tuple{\Sigma_2,\mathcal{I}}$.
\end{quote}
\end{definition}

\begin{lemma}\label{lemma-disjunction-removal}
Consider a rule set $\Sigma$ with split $\Tuple{\Sigma_1,\Sigma_2}$.
There is a set $\Sigma'$ of existential rules, such that,
for every database $\D$ over $\mathcal{S}_{\mathsf{in}}(\Sigma)$, we have:
\begin{enumerate}
\item $\Database,\Sigma \models \PrGoal$ iff $\Database, \Sigma' \models \PrGoal$, and
\item if $\Tuple{\Sigma_2, \D_2}$ is chase-terminating for every database $\D_2$ over $\mathcal{S}_{\mathsf{in}}(\Sigma_2)$,
then $\Tuple{\Sigma',\D}$ is also chase-terminating.
\end{enumerate}
\end{lemma}

To construct this set $\Sigma'$, we assume w.l.o.g.\ that all disjunctive rules have exactly two
disjuncts in the head.
We define $\Sigma'$ as the union of sets $\Sigma'_1$, $\Sigma'_2$  and $\Sigma'_3$ as shown in
Figures~\ref{figure:rule-set-disjunction-removal-r1}, \ref{figure:rule-set-disjunction-removal-r2} and \ref{figure:rule-set-disjunction-removal-r3}, which we explain below.

$\Sigma'_1$ uses a technique for modelling sets with chase-terminating existential rules \cite[Fig.~2]{DBLP:conf/icdt/KrotzschMR19}.
We adapt this to sets of ground atoms, called \emph{worlds} and denoted by variables $w$ in the figures.
Facts $\PrIns{\PrP}(\vec{t},w,w')$ express that world $w'$ is obtained by adding $\PrP(\vec{t})$
to world $w$. In particular, $\PrIns{\PrP}(\vec{t},w,w)$ states that $\PrP(\vec{t})$ is in $w$,
and we define $\mathsf{world}(w)=\{\PrP(\Va) \mid \PrIns{\PrP}(\Va,w,w) \in \I\}$ for any interpretation $\I$.
Worlds are created by adding database facts \eqref{collect-p} or by applying
rules to existing worlds \eqref{complete-r1}--\eqref{complete-r2}.
Worlds containing only database facts are marked with $\PrInit$.
Predicate $\PrSubs$ defines the subset relation on worlds.
Rules \eqref{propagate-p}--\eqref{rule-done} copy all prior facts to a new world before marking it
$\PrDone$. 

%
%

\begin{proposition}
\label{proposition-r1-disjunction-removal}
 $\Sigma'_1$ is chase-terminating and for every $\Database$ over $\mathcal{S}_{\mathsf{in}}(\Sigma)$, 
 the (unique) interpretation $\I$ resulting from some chase over $\Sigma'_1$ and $\Database$ is such that:
 \begin{itemize}
  \item if $\PrP(\Va) \in \Database$ and $\PrDone(w)\in\I$, there exists $w'$ such that $\PrIns{\PrP}(\Va,w,w') \in\I$ and $\mathsf{world}(w') = \mathsf{world}(w) \cup \{\PrP(\Va)\}$;
  \item if $\Rule \in \Sigma_1$ is applicable to $\mathsf{world}(w)$, creating $\PrP_1(\Va)$ or $\PrP_2(\Vb)$, there exists $w_1$ and $w_2$ such that $\{\PrIns{\PrP_1}(\Va,w,w_1),\PrIns{\PrP_2}(\Vb,w,w_2)\} \subseteq \I$, $\mathsf{world}(w_1) = \mathsf{world}(w) \cup \{\PrP_1(\Va)\}$ and 
  $\mathsf{world}(w_2) = \mathsf{world}(w) \cup \{\PrP_2(\Vb)\}$.
 \end{itemize}

\end{proposition}

Note that we cannot distinguish worlds that are not containing all database facts,
and that some worlds may contain more facts than needed to satisfy all disjunctive heads.

$\Sigma'_2$ now simulates the application of rules from $\Sigma_2$ in any of the worlds.
Computations relative to different worlds are independent from each other. 
Finally, $\Sigma'_3$ aggregates results from all worlds:
a world is accepting ($\PrWorldGoal$) if either $\PrGoal$ was derived locally \eqref{goal-rule-disj}
or it has two successor worlds for a disjunctive rule that are both accepting \eqref{propagate-ri}.
$\PrGoal$ is a consequence if any initial world is accepting \eqref{final-rule-disj}.
This finishes the construction of $\Sigma'$ as the main ingredient for proving Lemma~\ref{lemma-disjunction-removal}.

Finally, we can apply Lemma~\ref{lemma-disjunction-removal} to $\R_6=\FF{brake}(\R_5,\PrHalt)\cup\Pi$ from Section~\ref{sec_brake},
where $\Pi$ denotes rules \eqref{rul_halt_inconsistentDb}--\eqref{rul_halt_beforeFirst}.
Intuitively, a possible split is $\FF{brake}(\R_1,\PrHalt)$ and $\FF{brake}(\R_5\setminus\R_1,\PrHalt)\cup\Pi$.
Formally, however, $\FF{brake}(\R_1,\PrHalt)$ is not disjunctive Datalog due to existential rules
\eqref{rul_createFirst}, \eqref{rul_createLast}, and \eqref{rul_brake_init}. However, our result easily extends to such
rules with empty body: we can just add them to $\Sigma_1'$ and treat their inferences like facts from the initial database.
The other properties of Definition~\ref{def_split} are easy to verify. The fact that both rule sets have some common rules, such as
\eqref{rul_brake_halt}, is no concern.
Finally, it remains to argue termination for $\Sigma_2$ as required for item (2) in Lemma~\ref{lemma-disjunction-removal}.
This is slightly stronger than Lemma~\ref{lemma_termDisjRules_terminating} since we must also consider databases
that use some inferred predicates of $\R_1$. However, the proof of Lemma~\ref{lemma_termDisjRules_terminating}
and the emergency brake technique in general served the main purpose of safeguarding against problematic
structures among the inferred predicates of $\R_1$, and it is not hard to see that this already showed what we require here.
Combining all of our insights, we finally obtain:

\begin{theorem}\label{theo_mainresult}
Chase-terminating existential rules capture the class of all decidable homomorphism-closed queries.
\end{theorem}

\newcommand{\FormulaSet}{\FormatFormulaSet{F}}
\newcommand{\FS}{\FormulaSet}

\section{Limitations of Semi-Decidable Languages}
\label{section:expressivity-limits}

A \emph{query language} \Fragment over a schema $\mathcal{S}$ is a function from a set $L$ to $2^{\mathfrak{D}_{\mathcal{S}} }$, where $\mathfrak{D}_{\mathcal{S}}$ is the set of all databases over schema $\mathcal{S}$.
We say that \Fragment is \emph{semi-decidable} if membership to $L$ is semi-decidable, and that its \emph{query answering problem is decidable} if there exists a TM $\TM_\Fragment$ that takes as input some $(l, \Database) \in (L \times \mathfrak{D}_{\mathcal{S}})$ and decides whether $\Database \in \Fragment(l)$.


The set of chase-terminating existential rule sets is a query language that is not semi-decidable \cite{DBLP:journals/fuin/GrahneO18} and for which the query answering problem is decidable (by running the chase).
In fact, we show that one cannot find a semi-decidable query language with similar properties.

\begin{theorem}
\label{theorem:expressivity-limits}
There are no semi-decidable query languages that \FirstItem express all decidable, homomorphism-closed queries and \SecondItem for which query answering is decidable.
\end{theorem}

To show this result, we define a set $\mathcal{M}$ of TMs (cf. Definition~\ref{definition:tm-class}), show that $\mathcal{M}$ can be enumerated up to equivalence if there is a semi-decidable language that satisfies \FirstItem and \SecondItem above (cf. Lemma~\ref{lemma:expressivity-limits-1}), and finally prove that the consequence of this implication does not hold (cf. Lemma~\ref{lemma:expressivity-limits-2}).

\begin{definition}
\label{definition:tm-class}
Consider the set $\mathcal{M}$ of all TMs $\TM$ such that:
\FirstItem The TM \TM halts on all inputs.
\SecondItem If $\TM$ accepts some word $w$, then $w$ corresponds to a database over schema $\{\PrEdge\}$.
\ThirdItem Consider some words $w$ and $v$ that correspond to some 
$\D$ and $\DA$ in $\mathfrak{D}_{\{\PrEdge\}}$, respectively.
If \TM accepts $w$ and there is a homomorphism $h : \D \to \DA$, then \TM accepts $v$.
\end{definition}

Intuitively, $\mathcal{M}$ is the set of all deciders that solve homomorphism-closed queries over databases in $\mathfrak{D}_{\{\PrEdge\}}$.

\begin{lemma}
\label{lemma:expressivity-limits-1}
If there is a semi-decidable query language \Fragment that satisfies \FirstItem and \SecondItem in Theorem~\ref{theorem:expressivity-limits}, then $\mathcal{M}$ is enumerable up to equivalence.
\end{lemma}
\begin{proof}
If there is a language such as \Fragment, then there is an enumerator \ETM for $L$ that prints out a sequence $l_1, \l_2, \ldots$ and a decider $\TM_{\Fragment}$ that can be used to check if $\Database \in \Fragment(l)$ for each $(l,\Database) \in (L \times \mathfrak{D}_{\{\PrEdge\}})$.
For each $i \geq 1$, let $\TM_i$ be the TM that, on input $w$, performs the following computation: if $w$ corresponds to a database $\D\in\mathfrak{D}_{\{\PrEdge\}}$ and $\TM_\Fragment$ accepts $(l_i,\D)$, then \emph{accept}; otherwise, \emph{reject}.
By modifying \ETM we can define an enumerator that prints out the sequence $\TM_1, \TM_2, \ldots$, which contains $\mathcal{M}$ up to equivalence.
\end{proof}

\begin{lemma}
\label{lemma:expressivity-limits-2}
The set $\mathcal{M}$ is not enumerable up to equivalence.
\end{lemma}
\begin{proof}[Proof Sketch]
Assume that there is an enumerator that outputs a sequence $\TM_1, \TM_2, \ldots$ that includes $\mathcal{M}$ up to equivalence.
We obtain a contradiction by defining a sequence $\D_1, \D_2, \ldots$ of databases and a TM $\TM_d \in \mathcal{M}$ that diagonalises over $\TM_1, \TM_2, \ldots$ and $\D_1, \D_2, \ldots$
Namely, for each $i \geq 1$, let $\D_i = \{\PrEdge(u_1, u_2)$, $\ldots, \PrEdge(u_{p_{i+1}}, u_1)\}$ where $p_{i+1}$ is the $(i+1)$-th prime.
Moreover, $\TM_d$ is the TM that, on input $w$, performs the computation: (1) \emph{Reject} if $w$ does not correspond to some $\D\in\mathfrak{D}_{\{\PrEdge\}}$.
(2) \emph{Reject} if $\D$ can be hom-embedded into a path over \PrEdge.
(3) \emph{Accept} if $\PrEdge(u, u) \in \D$ for some null $u$.
(4) If there is some $i \geq 1$ such that there are less nulls in $\D_i$ than in $\D$, the TM $\TM_i$ accepts some serialisation that corresponds to $\D_i$, and there is a homomorphism $h : \D \to \D_i$; then \emph{reject}.
Otherwise, \emph{accept}.\qedhere
\end{proof}

\section{Discussion and Conclusion}
\label{section:disc-conc}

In this paper, we have established a characterization of all decidable homomorphism-closed Boolean queries.
We showed that these are exactly the 
chase-terminating existential rule queries, that is, queries that can be expressed by a set of (non-disjunctive) existential rules for which the standard chase universally terminates irrespective of the 
order of rule applications (as long as it 
is fair).

By its nature, our result immediately shows that various extensions of our framework do not increase its expressivity: 

\begin{theorem}
Chase-terminating existential rule queries have the same expressivity as	
\begin{enumerate}	
\item existential rule queries with guaranteed existence of some finite chase tree (for every database),
\item existential rule queries for which the chase terminates according to some fair strategy (such as datalog-first),
\item core-chase-terminating existential rule queries,
\item disjunctive chase-terminating existential rule queries.
\end{enumerate}	
\end{theorem}

\begin{proof}  
(3) Standard-chase termination implies core-chase termination. 
On the other hand, core chase termination implies decidability and 
thus our result applies. (1) and (2) Standard-chase termination implies these weaker form of guarantees, which themselve imply core chase termination.
(4) Obviously, every (non-disjunctive) existential rule set is a special case of a disjunctive one and for this special case, disjunctive chase termination coincides with termination of the (non-disjunctive) standard chase.
On the other hand, disjunctive existential rule queries are also closed under homomorphisms, and disjunctive universal chase termination obviously implies decidability. So our result applies.
\end{proof}

However, the applicability of our result does not stop at (syntactic) extensions of our framework, as it applies to arbitrary query languages and querying formalisms of different types. In particular we would like to stress the relationship to the very comprehensive existential rules fragment of \emph{bounded treewidth sets (bts) of rules} \cite{BagetLMS11} that is \emph{not} chase-terminating and encompasses a plethora of well-known existential rule fragments with decidable query entailment, including guarded \cite{CaliGK08}, frontier-guarded \cite{BagetLMS11}, and glut-guarded existential rules \cite{KR11:jointacyc}, as well as greedy bts \cite{BagetMRT11}:

\begin{theorem}
Let $\Sigma$ be a bounded-treewidth set of rules and $Q$ a conjunctive query. There is a chase-terminating set $\Sigma_Q$ of existential rules such that
$\mathcal{D},\Sigma \models Q$ iff 
$\mathcal{D},\Sigma_Q \models \PrGoal$. 
\end{theorem}

While possibly surprising, this is a straightforward consequence of decidability of conjunctive query entailment from bts and of homomorphism-closedness of existential rule queries in general. Note, however, that every $Q$ would give rise to a different $\Sigma_Q$. In fact, asking for a ``uniform'' chase-terminating existential rules set $\Sigma'$ satisfying $\mathcal{D},\Sigma' \models Q$ iff $\mathcal{D},\Sigma \models Q$ would change the game \cite{ZZY15:finitechase}. Such a set will not exist in all cases.

While our result addresses many of the open questions regarding \emph{expressivity} of the terminating chase \cite{DBLP:conf/icdt/KrotzschMR19} an important avenue for future work is to investigate potential differences when it comes to the corresponding computational \emph{complexities}. We deem it likely that not all of the discussed chase variants give rise to worst-case optimal computations.

\paragraph*{Acknowledgements}
This work is partly supported by DFG
in project number 389792660 (TRR 248, \href{https://www.perspicuous-computing.science/}{Center for Perspicuous Systems}), by BMBF in the \href{https://www.scads.de}{Center for Scalable Data Analytics and Artificial Intelligence} (ScaDS.AI), by the
\href{https://cfaed.tu-dresden.de/}{Center for Advancing Electronics Dresden} (cfaed), by the ERC Consolidator Grant DeciGUT (project number 771779), and by the ANR project CQFD (ANR-18-CE23-0003).

\bibliographystyle{kr}
\bibliography{bib}

\newpage
\newgeometry{left=35mm,right=35mm}
\onecolumn
\appendix
\section{Proofs for Section \ref{sec_disRules}
}

In this section, we prove the claims made in Section \ref{sec_disRules}. In particular, besides proving Lemmas \ref{lem:univ-model-set-1}, \ref{lem:univ-model-set-2}, \ref{lem:univ-model-set-3}, \ref{lem:univ-model-set-4}, \ref{lem:univ-model-set-5},  \ref{lemma:soundness}, and \ref{lemma:completeness}, we show Lemma~\ref{app-lem:various-claims-structure} from which the next claims (given in the order they occur in the section) follow.
\begin{enumerate}
\item $|\IS_1| = |\IS_2| = |\IS_3| = |\IS_4|= |\IS_5|$.

\item For each $1\leq i<j\leq 5$ and each $\I\in\IS_i$, there is exactly one $\mathcal{J}\in\IS_j$ with $\I\subseteq\mathcal{J}$. 

\item Every $\I\in\IS_i$ with $1 \leq i \leq 5$ contains a unique interpretation $\FunSeed(\I)\in\IS_1$,
and there is a unique ordered partition $\FunOrder(\I)\in\FunOrders$ such that
$\FunSeed(\I) \in \FunCompletions(\FunOrder(\I))$. 
\end{enumerate}

Recall that $\Delta=\EI{\Nulls}{\D} \cup \{\FirstNull,  \LastNull\}$. 
We will call a \D-\emph{order} a database over schema $\{\PrFirst, \PrLast, \PrOrder, \PrEq, \PrNeq\}$ and nulls from $\Delta$, and a \D-\emph{completion} a database over schema $\{\PrFreshIDBIn{\PrP}, \PrFreshIDBNotIn{\PrP}\mid\PrP \in \mathcal{S}\}$ and nulls from $\Delta$. 
We say that a \D-completion is \emph{complete} if (i) it includes $\{\PrFreshIDBIn{\PrP}(\Vt)\mid \PrP(\Vt)\in \D\}$ and (ii) 
 for every $\PrP \in \mathcal{S}$ and tuple $\Vt\subseteq \Delta$ of matching arity, it contains exactly one of $\{\PrFreshIDBIn{\PrP}(\Vt),\PrFreshIDBNotIn{\PrP}(\Vt)\}$.

\begin{lemma}\label{app-lem:various-claims-structure}
There is a one-to-one correspondence between the interpretations in $\IS_i$ ($1\leq i\leq 5$) and the pairs of the form $(\VT,\FactSet_\D)$, where  $\VT = T_1, \ldots, T_k \in \FunOrders$ is an ordered partition of $\Delta$  with $\FirstNull \in T_1$ and $ \LastNull \in T_k$, and $\FactSet_\D$ is a complete \D-completion. 
More precisely: 
\begin{itemize}
\item An interpretation $\I\in\IS_1$ corresponds to a pair $(\VT,\FactSet_\D)$ iff $\I$ is the union of the following databases:
\begin{itemize}
\item $\D\cup\{\PrDomain(t)\mid t\in\Delta\}$,
\item the \D-order corresponding to $\VT$: $\{\PrFirst(t) \mid t \in T_1\}\cup\{\PrLast(t) \mid t \in T_k\} \cup\{\PrEq(t, u) \mid 1 \leq i \leq k; t, u \in T_i\}\cup\{\PrOrder(t, u), \PrNeq(t, u), \PrNeq(u, t) \mid 1 \leq i < j \leq k$, $t \in T_i$, $u \in T_j\} $, and
\item $\FactSet_\D$.
\end{itemize}
\item An interpretation $\I\in\IS_i$ ($2\leq i\leq 5$) corresponds to $(\VT,\FactSet_\D)$ iff it includes some $\I_1\in\IS_1$ that corresponds to $(\VT,\FactSet_\D)$.
\end{itemize}
In particular, it follows that every $\I\in\IS_i$ with $1 \leq i \leq 5$ contains a unique interpretation $\FunSeed(\I)\in\IS_1$,
and there is a unique ordered partition $\FunOrder(\I)\in\FunOrders$ such that
$\FunSeed(\I) \in \FunCompletions(\FunOrder(\I))$. 
\end{lemma}
\begin{proof}
The correspondence between interpretations in $\IS_1$ and pairs of the form $(\VT,\FactSet_\D)$ is well defined because $\IS_1=\bigcup_{\VT \in \FunOrders}\FunCompletions(\VT)$ where $\FunCompletions(\VT)$ is the set of all minimal models of $\RuleSet_1$ and $\Database$ that contain the \D-order corresponding to $\VT$. In particular, every $\I\in\IS_1$ must include a  complete \D-completion to satisfy \eqref{rul_DbCompletion} and \eqref{rul_copyPrP}, and by minimality $\I$ does not include any set of the form $\{\PrFreshIDBIn{\PrP}(\Vt), \PrFreshIDBNotIn{\PrP}(\Vt)\}$. 
Moreover, by minimality, $\I$ does not contain any other atom over predicates $\{\PrFirst, \PrLast, \PrOrder, \PrEq, \PrNeq\}$. 

Each $\IS_{i}$ ($2\leq i\leq 5$) is obtained from $\IS_{i-1}$ by applying the function $\FunInterpretation{i}$ to each of its elements and by construction of $\FunInterpretation{i}$, $\I\subseteq\FunInterpretation{i}(\I)$ and $\FunInterpretation{i}(\I)\setminus\I$ does not contain any atoms on predicates from $\mathcal{S}\cup\{\PrDomain,\PrFirst, \PrLast, \PrOrder, \PrEq, \PrNeq\}\cup\{\PrFreshIDBIn{\PrP}, \PrFreshIDBNotIn{\PrP}\mid\PrP \in \mathcal{S}\}$. 
Hence the one-to-one correspondence between elements of $\IS_1$ and pairs of the form $(\VT,\FactSet_\D)$ is preserved when applying $\FunInterpretation{i}$ to elements of $\IS_{i-1}$ to obtain the elements of $\IS_i$ for each $2 \leq i \leq 5$.
\end{proof}

\noindent\textbf{Lemma \ref{lem:univ-model-set-1}.} 
\textit{$\IS_1=\bigcup_{\VT \in \FunOrders}\FunCompletions(\VT)$ is a universal model set of $\R_1$ and $\D$.}

\begin{proof}
Let $\mathfrak{K}$ be the result of some chase tree for $\R_1$ and $\D$ that prioritises the application of
rules in the order of their appearance in Figure~\ref{figure:rule-set-1}. 
By Fact~\ref{fact_chase-univ-model-sets}, $\mathfrak{K}$ is a universal model set for $\R_1$ and $\D$ and we show that each $\mathcal{K}\in\mathfrak{K}$ is isomorphic to a unique $\I\in\IS_1$.
 
We use the one-to-one correspondence between the interpretations in $\IS_1$ and the pairs of the form $(\VT,\FactSet_\D)$ described in Lemma~\ref{app-lem:various-claims-structure}.  
Since $\mathfrak{K}$ is a result of the chase on $\R_1$ and $\D$, the domain of every $\mathcal{K}\in\mathfrak{K}$ contains all nulls in $\EI{\Nulls}{\D}$ 
and exactly two fresh nulls introduced to satisfy rules $\eqref{rul_createFirst}$ and $\eqref{rul_createLast}$, respectively. We name them $\FirstNull$ and $\LastNull$ respectively, so that $\Delta$ is the domain of $\mathcal{K}$, and define a one-to-one correspondence between the interpretations in $\mathfrak{K}$ and the pairs of the form $(\VT,\FactSet_\D)$ as in Lemma \ref{app-lem:various-claims-structure}. 
This can be done because of the following points:
\begin{itemize}
\item Every $\mathcal{K}\in\mathfrak{K}$ includes $\D$ and does not contain any other atom on predicates from $\mathcal{S}$ because such predicates do not occur in rule heads. 

\item  Every $\mathcal{K}\in\mathfrak{K}$ includes $\{\PrDomain(t)\mid t\in\Delta\}$ and does not contain any other atom on predicate $\PrDomain$ (since there are no more nulls).

\item Every $\mathcal{K}\in\mathfrak{K}$ includes a \D-order corresponding to some $\VT  \in \FunOrders$ and does not contain any other atom on predicates from $\{\PrFirst, \PrLast, \PrOrder, \PrEq, \PrNeq\}$: 
\begin{itemize}

\item Let $\sim$ be the relation over $\Delta$ defined by $t\sim t'$ iff $\PrEq(t,t')\in \mathcal{K}$. This is an equivalence relation because (i) $\{\PrDomain(t)\mid t\in\Delta\}\subseteq\mathcal{K}$ and $\mathcal{K}$ satisfies \eqref{rul_DomToEq} (reflexivity); (ii) $\mathcal{K}$ satisfies  \eqref{rul_ReflEq} (symmetry);
and (iii)  $\mathcal{K}$ satisfies \eqref{rul_EqCongruence} instantiated for $\PrR=\PrEq$ (transitivity). The equivalence classes form a partition $T_1,\dots, T_k$ of~$\Delta$ such that $\{\PrEq(t, u) \mid 1 \leq i \leq k; t, u \in T_i\}\subseteq\mathcal{K}$.

\item For every $u,t\in\Delta$, $\mathcal{K}$ contains at least one of $\{\PrEq(u, t), \PrNeq(u, t)\}$ because it satisfies \eqref{rul_EqOrNeq}. 
Moreover, $\mathcal{K}$ cannot contain both $\PrEq(u, t)$ and $\PrNeq(u, t)$ because the rules \eqref{rul_DomToEq}, \eqref{rul_ReflEq}
and \eqref{rul_EqCongruence} have a higher priority than \eqref{rul_EqOrNeq}. 
Hence $\{ \PrNeq(t, u), \PrNeq(u, t) \mid 1 \leq i < j \leq k$, $t \in T_i$, $u \in T_j\} \subseteq\mathcal{K}$.

\item $\mathcal{K}$ includes $\{\PrFirst(u)\mid u\sim \FirstNull\}$ and $\{\PrLast(u)\mid u\sim \LastNull\}$ by definition of $\FirstNull,\LastNull$ and because it satisfies \eqref{rul_EqCongruence}. Assuming w.l.o.g. that $T_1$ is the equivalence class of $\FirstNull$ and $T_k$ that of $\LastNull$: $\{\PrFirst(t) \mid t \in T_1\}\cup\{\PrLast(t) \mid t \in T_k\}\subseteq\mathcal{K}$. 

\item $\mathcal{K}$ does not include any other atom over $\PrFirst$ or $\PrLast$ because rules \eqref{rul_createFirst} and \eqref{rul_createLast} are applied only once.

\item For every $u,t\in\Delta$ such that $ \PrNeq(u, t)\in \mathcal{K}$, $\mathcal{K}$ contains at least one of $\{\PrOrder(u, t),\PrOrder(t, u) \}$ because it satisfies \eqref{rul_LTchoice}.

\item For every $u,t\in\Delta$, if $\PrFirst(u)\in \mathcal{K}$ (i.e. $u\sim\FirstNull$) and $\PrFirst(t)\notin \mathcal{K}$ (i.e. $t\not\sim \FirstNull$) then $\PrOrder(u, t)\in \mathcal{K}$ because of \eqref{rul_FirstLessThan}. 

\item For every $u,t\in\Delta$, if $\PrLast(u)\in \mathcal{K}$ (i.e. $u\sim\LastNull$) and $\PrLast(t)\notin \mathcal{K}$ (i.e. $t\not\sim \LastNull$) then $\PrOrder(t,u)\in \mathcal{K}$ because of \eqref{rul_LessThanLast}.

\item If $\PrOrder(u, t)\in \mathcal{K}$ then $\PrOrder(u', t')\in \mathcal{K}$ for all $u'\sim u$ and $t\sim t'$ because $\mathcal{K}$ satisfies \eqref{rul_EqCongruence}.   

\item $\mathcal{K}$ contains the transitive closure of $\PrOrder$ because it satisfies \eqref{rul_TransLT}. 

\item Since the rules \eqref{rul_FirstLessThan}, \eqref{rul_LessThanLast}, \eqref{rul_DomToEq}, \eqref{rul_ReflEq}, \eqref{rul_EqCongruence}, and \eqref{rul_TransLT} have a higher priority than \eqref{rul_LTchoice}, then $\mathcal{K}$ cannot contain both $\PrOrder(u, t)$ and $\PrOrder(t, u)$. 

\item Hence $T_1,\dots, T_k$ can be ordered in a way such that $\{\PrOrder(t, u) \mid 1 \leq i < j \leq k$, $t \in T_i$, $u \in T_j\}  \subseteq\mathcal{K}$ and the ordered partition corresponds to some $\VT  \in \FunOrders$ as required.
\end{itemize}

\item Every $\mathcal{K}\in\mathfrak{K}$ contains a complete \D-completion $\FactSet_\D$ and does not contain any other atom on predicates from $\{\PrFreshIDBIn{\PrP}, \PrFreshIDBNotIn{\PrP}\mid\PrP \in \mathcal{S}\}$: 
\begin{itemize}
\item $\mathcal{K}$ includes $\{\PrFreshIDBIn{\PrP}(\Vt)\mid \PrP(\Vt)\in \D\}$ because it satisfies \eqref{rul_copyPrP}.
\item For every $\PrP \in \mathcal{S}$ and tuple $\Vt\subseteq\Delta$ of matching arity, $\mathcal{K}$ contains at least one of $\{\PrFreshIDBIn{\PrP}(\Vt),\PrFreshIDBNotIn{\PrP}(\Vt)\}$ because it satisfies  \eqref{rul_DbCompletion}. We show that it does not contain both. 
The rules of the form \eqref{rul_DbCompletion} have the lowest priority and applying such a rule can only trigger the application of some rules of the form \eqref{rul_EqCongruence}, which themselves cannot trigger the application of any rules but \eqref{rul_EqCongruence}. Hence, (i) the application of these rules in the chase tree happens after the construction of the \D-order in $\mathcal{K}$, and (ii) when a rule of the form \eqref{rul_DbCompletion} is applied and adds some $\PrFreshIDBIn{\PrP}(\Vt)$ (resp. $\PrFreshIDBNotIn{\PrP}(\Vt)$), the exhaustive application of  rules \eqref{rul_EqCongruence} adds all $\PrFreshIDBIn{\PrP}(\Vt')$ (resp. $\PrFreshIDBNotIn{\PrP}(\Vt')$) such that $\Vt'\sim\Vt$, where $\Vt'\sim\Vt$ iff $t_i\sim t'_i$ for every $i$. 
Suppose for a contradiction that there exists $\Vu$ such that $\PrFreshIDBIn{\PrP}(\Vu)\in\mathcal{K}$ and $\PrFreshIDBNotIn{\PrP}(\Vu)\in\mathcal{K}$. Since the rules \eqref{rul_DbCompletion} and \eqref{rul_EqCongruence} are the only ones that have predicates from $\{\PrFreshIDBIn{\PrP}, \PrFreshIDBNotIn{\PrP}\mid\PrP \in \mathcal{S}\}$ in the head, 
it follows that there exist $\PrFreshIDBIn{\PrP}(\Vu_1)\in\mathcal{K}$ and $\PrFreshIDBNotIn{\PrP}(\Vu_2)\in\mathcal{K}$ that have both been added by some application of \eqref{rul_DbCompletion} and are such that $\Vu_1\sim\Vu$ and $\Vu_2\sim\Vu$. Assume that $\PrFreshIDBIn{\PrP}(\Vu_1)$ was added first. Then by (ii) $\PrFreshIDBIn{\PrP}(\Vu_2)$ has been added before any new application of \eqref{rul_DbCompletion}, and the rule $\bigwedge_{x \in \Vx} \PrDomain(x) \to \PrFreshIDBIn{\PrP}(\Vx) \vee \PrFreshIDBNotIn{\PrP}(\Vx)$ with the substitution that maps $\Vx$ to $\Vu_2$ was already satisfied, contradicting the introduction of $\PrFreshIDBNotIn{\PrP}(\Vu_2)$. 
\end{itemize}

\item For every $(\VT,\FactSet_\D)$, there exists $\mathcal{K}\in\mathfrak{K}$ that corresponds to $(\VT,\FactSet_\D)$. The existence of $\mathcal{K}$ is witnessed by the path of the chase tree such that for every vertex labelled with $\mathcal{E}$ where a disjunctive rule \eqref{rul_EqOrNeq}, \eqref{rul_LTchoice}, or \eqref{rul_DbCompletion} is applied, the path chooses the child labelled with $\mathcal{C}_i=\mathcal{E}\cup\sigma_i(\Head_i)$ where $\sigma_i(\Head_i)$ is included in $\FactSet_\D$ or in the \D-order corresponding to $\VT$. 
\end{itemize}
 
 Hence, we can define a one-to-one correspondence between interpretations in $\mathfrak{K}$ and interpretations in $\IS_1$ using the corresponding pairs of the form $(\VT,\FactSet_\D)$. 
By construction, every $\mathcal{K}\in\mathfrak{K}$ and $\I\in\IS_1$ that are in correspondence are isomorphic (note that since we choose to name the two fresh nulls $\FirstNull$ and $\LastNull$, they are actually identical). 
\end{proof}

\begin{lemma}\label{app:int2-datalog-first-chase}
Given some $\I\in \IS_1$, $\FunInterpretation{2}(\I)$ is isomorphic to the unique interpretation in a result of the chase over $\R_2$ and $\I$.
\end{lemma}
\begin{proof}
We will let $\FunOrder(\FactSet)= T_1, \ldots, T_k$. 
The interpretation $\FunInterpretation{2}(\I)$ is equal to:
\begin{align*}
& \I \cup \{\PrRoot(u_1)\} \cup
\{\PrHasChild(u_w,u_{wz})\mid wz\in\mathbf{Z}, 2\leq z\leq k\} \cup{}\\
&\{\PrLeaf(u_w)\mid w\in\mathbf{Z}; \FormatFunction{end}(w)=k\} \cup
\{\PrLink(x,u_w)\mid w\in\mathbf{Z}; x\in T_{\FormatFunction{end}(w)}\}\cup{}\\
&
\{\PrFreshIDBInTape{\PrP}(u_{w_1},\dots, u_{w_{\Arity(\PrP)}})\mid \PrFreshIDBIn{\PrP}(t_1,\dots, t_{\Arity(\PrP)})\in\I; w_i\in\mathbf{Z};t_i\in T_{\FormatFunction{end}(w_i)}; 1\leq i\leq \Arity(\PrP)\}\cup{}
\\&
 \{\PrFreshIDBNotInTape{\PrP}(u_{w_1},\dots,u_{w_{\Arity(\PrP)}})\mid \PrFreshIDBNotIn{\PrP}(t_1,\dots, t_{\Arity(\PrP)})\in\I; w_i\in\mathbf{Z};t_i\in T_{\FormatFunction{end}(w_i)}; 1\leq i\leq \Arity(\PrP)\}.
\end{align*}

We describe a chase over $\R_2$ and $\I$ that prioritises rule \eqref{rul_RepEq} step by step (note that if we do not prioritise \eqref{rul_RepEq}, we could get a result which is not isomorphic to $\FunInterpretation{2}(\I)$). Instead of constructing an isomorphism by giving for each fresh null introduced during the chase its counterpart in $\FunInterpretation{2}(\I)$, we directly set the the names of the fresh nulls so that the chase result coincides with $\FunInterpretation{2}(\I)$. Note that no rule in Figure \ref{figure:rule-set-1} is applicable to \I and that applying a rule in Figure \ref{figure:rule-set-2} cannot make a rule in Figure \ref{figure:rule-set-1} applicable.

\begin{itemize}
\item Apply $\Tuple{\Rule, \sigma}$ with $\Rule = \eqref{rul_RootCreation}$ and $\sigma$ the substitution that maps $x$ to some $t \in T_1$ (recall that $\{\PrFirst(t) \mid t \in T_1\}\subseteq\I$). This introduces a fresh null that we call $u_1$, and adds $\PrRoot(u_1)$ and $\PrLink(t,u_1) $.

\item Apply exhaustively \eqref{rul_RepEq}: add $\PrLink(t, u_1)$ for all $t \in T_1$.

\item For $2\leq i \leq k$: 
\begin{itemize}
\item Choose one $t_1\in T_1$ and one $t_i\in T_i$ and 
apply $\Tuple{\Rule, \sigma}$ with $\Rule = \eqref{rul_NextTreeRep}$ and $\sigma$ the substitution that maps $x$ to $t_1$, $v$ to $u_1$, and $z$ to $t_i$ (recall that the atoms on predicate $ \PrOrder$ in $\I$ are $\{\PrOrder(t,u) \mid 1 \leq j < j' \leq k$, $t \in T_j$, $u \in T_j'\}$). This introduces a fresh null that we call $u_{1\,i}$ and adds $\PrHasChild(u_1,u_{1\,i})$, $\PrLink(t_i,u_{1\,i}) $ (note that this is the first fact of the form $\PrLink(t_i,x)$ with $t_i\in T_i$ that we introduce).
\item Apply exhaustively \eqref{rul_RepEq}: add $\PrLink(t, u_{1\,i})$ for all $t\in T_ i$.
\end{itemize}

\item Continue applying \eqref{rul_NextTreeRep} and \eqref{rul_RepEq} exhaustively as done in the previous step, introducing all nulls of the form $u_{w}$ with $w \in\mathbf{Z}$, and adding facts $\{\PrHasChild(u_w,u_{wz})\mid wz\in\mathbf{Z}, 2\leq z\leq k\}$ and $\{\PrLink(x,u_w)\mid w\in\mathbf{Z}; x\in T_{\FormatFunction{end}(w)}\}$.

\item Apply exhaustively \eqref{rul_LeafCreation} and add $\{\PrLeaf(u_w)\mid w\in\mathbf{Z}; \FormatFunction{end}(w)=k\}$  (recall that $\{\PrLast(t) \mid t \in T_k\}\subseteq\I$).

\item Apply exhaustively \eqref{rul_makeTreeIn}: for all $\PrFreshIDBIn{\PrP}(t_1,\dots, t_{\Arity(\PrP)})\in\I$,  add all $\PrFreshIDBInTape{\PrP}(u_{w_1},\dots, u_{w_{\Arity(\PrP)}})$ where $t_i\in T_{\FormatFunction{end}(w_i)}$.

\item Apply exhaustively \eqref{rul_makeTreeNotIn}: for all $\PrFreshIDBNotIn{\PrP}(t_1,\dots, t_{\Arity(\PrP)})\in\I$,  add all $\PrFreshIDBNotInTape{\PrP}(u_{w_1},\dots, u_{w_{\Arity(\PrP)}})$ where $t_i\in T_{\FormatFunction{end}(w_i)}$.
\qedhere
\end{itemize}
\end{proof}

\noindent\textbf{Lemma \ref{lem:univ-model-set-2}.} 
\textit{$\IS_2 = \{\FunInterpretation{2}(\I) \mid \I \in \IS_1\}$ is a universal model set of $\R_2$ and $\D$.}

\begin{proof}
By Lemmas~\ref{lem:univ-model-set-1} and \ref{app:int2-datalog-first-chase}, $\IS_2$ only contains models of $\R_2$ and $\D$. We show that for every model $\M$ of $\R_2$ and $\D$,  there exists $\I_2\in\IS_2$ and a homomorphism 
$\Hom_2 : \I_2 \to \M$. 
\begin{enumerate}
\item Let $\M$ be a model of $\R_2$ and $\D$. \label{lab-Mmodel}
\item By Lemma~\ref{lem:univ-model-set-1}: since $\M$ is a model of $\R_2\supseteq \R_1$ and $\D$, there exist $\I_1\in\IS_1$ and a homomorphism $\Hom_1 : \I_1 \to \M$. \label{lab-hom1}
\item By Lemma~\ref{app:int2-datalog-first-chase}: since $\I_1\in\IS_1$,  there exists a chase over $\R_2$ and $\I_1$ such that the unique interpretation \J in its result is such that there exists an isomorphism $\HomA_2 :\FunInterpretation{2}(\I_1)\to\J$. \label{lab-isog2}
\item By definition of $\J$, and since no rule from $\R_1$ is applied in the chase of over $\R_2$ and $\I_1$ (see proof of Lemma~\ref{app:int2-datalog-first-chase}) and $\R_2\setminus\R_1$ contains no disjunctive rules, it follows that $\J$ is a universal model of $\R_2 \setminus \R_1$ and $\I_1$.\label{lab-univJ}
\item By \eqref{lab-Mmodel}, \eqref{lab-hom1} and \eqref{lab-univJ}, there exists a homomorphism $\HomA_1 :  \J\to\M$ that extends $\Hom_1$. \label{lab-homg2}
\item Let $\I_2=\FunInterpretation{2}(\I_1)\in\IS_2$.  By \eqref{lab-isog2} and \eqref{lab-homg2}, there is an homomorphism 
$\Hom_2 = \HomA_2\circ \HomA_1 : \I_2 \to \M$. \qedhere
\end{enumerate}
\end{proof}

In the next proofs, we will use the following notation. 
Given a chase tree node labelled $\mathcal{C}_i=\mathcal{E}\cup\sigma_i(\Head_i)$ and obtained by applying some $\Tuple{\Rule, \sigma}$ to its parent node labelled by $\mathcal{E}$, 
then for each variable $y$ of $\Rule$ not in the domain of $\sigma$, we denote by $u_{\sigma, \Rule, y}$ the fresh null such that $\sigma_i(y)=u_{\sigma, \Rule, y}$ where $\sigma_i$ is as introduced in Definition~\ref{definition_chase}. 

\begin{lemma}\label{app:int3-datalog-first-chase}
Given some $\I\in \IS_2$, $\FunInterpretation{3}(\I)$ is well-defined and $\FunInterpretation{3}(\I)$ is isomorphic to the unique interpretation in a result of the chase over $\R_3$ and $\I$.
\end{lemma}
\begin{proof}
We will let $\FunOrder(\I)= T_1, \ldots, T_k$. 
The interpretation $\FunInterpretation{3}(\I)$ is equal to the union of the following databases:

\begin{enumerate}
\item \I.
\item For each $w\in\mathbf{Z}$ of length $|w|$,
with $b_1\cdots b_\ell$ the binary representation of $|w|+1$, the database $\FormatFunction{EncPos}(w)$: 
$$\{\PrEncoding(u_w, e_w^1, e_w^\ell)\} \cup \{\PrCell{b_i}(e_w^{i}) \mid 1 \leq i \leq \ell \} \cup{} 
\{\PrNext(e_w^{i-1}, e_w^i) \mid 2 \leq i \leq \ell\}.$$
\item 
For each $w\in\mathbf{Z}$ of length $|w|<k$, and $z\in\{1,\dots, k\}$ such that $wz\in\mathbf{Z}$, 
with $b_1\cdots b_\ell$ the binary representation of $|w|+2$ and $b'_1\cdots b'_{\ell'}$ the binary representation of $|w|+1$, the database $\FormatFunction{CpyPlusOne}(w, wz)$ defined as follows: 
 \begin{align*}
& \{\PrCopyPlusOne(e^i_{w}, e^{\ell'}_{w}, e^j_{wz}, e^\ell_{wz}) \mid \\&\qquad\qquad\text{ the number represented by }b_j \cdots b_{\ell}\text{ is equal to the number represented by }b'_i\cdots b'_{\ell'}\text{ plus one} \}\\
 &\cup\{\PrCopy(e^i_{w}, e^{\ell'}_{w}, e^i_{wz}, e^\ell_{wz}) \mid b_i\cdots b_{\ell}=b'_i\cdots b'_{\ell'}\}.
 \end{align*}
\end{enumerate}

We describe a (Datalog-first) chase over $\R_3$ and $\I$ step by step. We directly set the the names of the fresh nulls so that the chase result coincides with $\FunInterpretation{3}(\I)$. 
Note that no rule in Figures \ref{figure:rule-set-1} or \ref{figure:rule-set-2} is applicable  to \I and that applying a rule in Figure \ref{figure:rule-set-3} cannot make a rule in Figures \ref{figure:rule-set-1} or \ref{figure:rule-set-2} applicable.
\begin{itemize}
\item Apply $\Tuple{\Rule, \sigma}$ with $\Rule =\eqref{rul_encRoot}$ 
and $\sigma(u)=u_1$. This introduces two fresh nulls, $u_{\sigma, \Rule, y_1}=e_1^1$ and $u_{\sigma, \Rule, y_2}=e_1^2$ and adds $\PrEncoding(u_1, e_1^1,e_1^2)$, $\PrCell{0}(e_1^1) $, $\PrNext(e_1^1,e_1^2)$ and $\PrCell{1}(e_1^2) $.
\end{itemize}
We show by induction on $p$ that for every $w\in\mathbf{Z}$ of length $|w|=p$ and $z\in\{1,\dots, k\}$ such that $wz\in\mathbf{Z}$, our chase sequence adds $\FormatFunction{EncPos}(wz)$, and $\FormatFunction{CpyPlusOne}(w, wz)$.
\smallskip

\noindent\emph{Case $p=1$.} In this case, $w=1$ and $z\in\{2,\dots, k\}$.
\begin{itemize}
\item For $2\leq i\leq k$: 
\begin{itemize}
\item Apply $\Tuple{\Rule, \sigma}$ with $\Rule =\eqref{rul_encChild}$ 
and $ \sigma(u)=u_1$, $ \sigma(y_1)=e_1^1$, $ \sigma(y_\_)=e_1^2$, $ \sigma(v)=u_{1\,i}$. This introduces two fresh nulls, $u_{ \sigma, \Rule, z_1}=e_{1\,i}^1$ and $u_{ \sigma, \Rule,  z_\_}=e_{1\,i}^2$ and adds $\PrEncoding(u_{1\,i}, e_{1\,i}^1, e_{1\,i}^2)$ and $\PrCopyPlusOne(e_1^1, e_1^2, e_{1\,i}^1, e_{1\,i}^2)$.

\item Apply \eqref{rul_CpyPlusOneCell0}: add $\PrCell{1}(e_{1\,i}^1)$, $\PrNext(e_{1\,i}^1, e_{1\,i}^2)$, $\PrCell{1}(e_{1\,i}^2)$. 
\end{itemize}
\end{itemize}

\noindent\emph{Induction step:} Assume that our chase sequence already added all required facts for every $w'\in\mathbf{Z}$ of length $|w'|<p$, and $z'\in\{1,\dots, k\}$ such that $w'z'\in\mathbf{Z}$, and consider $w\in\mathbf{Z}$ of length $|w|=p$, and $z\in\{1,\dots, k\}$ such that $wz\in \mathbf{Z}$. Let $b_1\cdots b_\ell$ the binary representation of $p+2$ and $b'_1\cdots b'_{\ell'}$ the binary representation of $p+1$.

\begin{itemize}
\item By induction hypothesis, $\PrEncoding(u_{w}, e_{w}^1, e_{w}^{\ell'})$ was added by the chase sequence. Apply $\Tuple{\Rule, \sigma}$ with $\Rule =\eqref{rul_encChild}$ 
and $ \sigma(u)=u_{w}$, $ \sigma(y_1)=e_{w}^1$, $ \sigma(y_\_)=e_{w}^{\ell'}$, $ \sigma(v)=u_{wz}$. This introduces two fresh nulls, $u_{ \sigma, \Rule, z_1}=e_{wz}^1$ and $u_{ \sigma, \Rule,  z_\_}=e_{wz}^{\ell}$ and adds $\PrEncoding(u_{wz}, e_{wz}^1, e_{wz}^\ell)$ and $\PrCopyPlusOne(e_{w}^1, e_{w}^{\ell'},e_{wz}^1, e_{wz}^{\ell})$.
\item By induction hypothesis, the chase sequence added $\PrCell{b'_i}(e_{w}^{i})$, $1 \leq i \leq \ell'$, and $\PrNext(e_{w}^{i-1}, e_{w}^i)$, $2 \leq i \leq \ell'$. We distinguish two cases, depending on the value of $b'_1$.

(1) In the case where $b'_1=0$, we have $\ell=\ell'$ and $b_1 \cdots b_\ell=1 b'_2\cdots b'_{\ell}$. Moreover, note that $\ell>2$ (otherwise $b'_1 b'_2=0 1$ and $p+1= 2$). 
\begin{itemize}
\item Apply $\Tuple{\Rule, \sigma}$ with $\Rule =\eqref{rul_CpyPlusOneCell0Next}$ 
and $\sigma(y_1)=e_{w}^1$, $\sigma(y_\_)=e_{w}^{\ell}$, $\sigma(y_2)=e_{w}^2$,  $\sigma(y_3)=e_{w}^3$, $\sigma(z_1)=e_{wz}^1$, $\sigma(z_\_)= e_{wz}^\ell$. This introduces a fresh null, $u_{\sigma, \Rule, z_2}=e_{wz}^2$, and adds the atoms $\PrCopy(e_{w}^2, e_{w}^{\ell}, e_{wz}^2, e_{wz}^\ell)$, $\PrCell{1}(e_{wz}^1)$, and $\PrNext(e_{wz}^1, e_{wz}^2)$. 

\item While $\PrCopy(e_{w}^{\ell-1}, e_{w}^{\ell}, e_{wz}^{\ell-1}, e_{wz}^\ell)$ has not been introduced, apply \eqref{rul_copyRec}, introducing fresh nulls $e_{wz}^i$ for $2<i<\ell$ and adding atoms $\PrCopy(e_{w}^i, e_{w}^{\ell}, e_{wz}^i, e_{wz}^\ell)$, $\PrCell{b_{i-1}}(e_{wz}^{i-1})$, and $\PrNext(e_{wz}^{i-1}, e_{wz}^i)$.  

\item Apply  \eqref{rul_copyBase} and add $\PrCell{b_{\ell-1}}(e_{wz}^{\ell-1})$, $\PrNext(e_{wz}^{\ell-1}, e_{wz}^\ell)$, $\PrCell{1}(e_{wz}^\ell)$.
\end{itemize}

(2) In the case where $b'_1=1$, we have $b_1=0$ and the number represented by $b_2\cdots b_\ell$ is equal to the number represented by $b'_2 \cdots b'_{\ell'}$ plus one. 
\begin{itemize}
\item While no atom of the form $\PrCopy(e_{w}^{j}, e_{w}^{\ell'}, e_{wz}^{j}, e_{wz}^\ell)$ or $\PrCopyPlusOne(e_{w}^{\ell'-1}, e_{w}^{\ell'}, e_{wz}^{\ell'-1}, e_{wz}^\ell)$ have been introduced, and for $i$ starting from $2$ and growing by one at each cycle, 
apply the applicable  $\Tuple{\Rule, \sigma}$ among:
\begin{itemize}
\item $\Rule=\eqref{rul_CpyPlusOneCell0Next}$ and 
$\sigma(y_1)=e_{w}^{i-1}$, $\sigma(y_\_)=e_{w}^{\ell'}$, $\sigma(y_2)=e_{w}^i$,  $\sigma(y_3)=e_{w}^{i+1}$, $\sigma(z_1)=e_{wz}^{i-1}$, $\sigma(z_\_)= e_{wz}^\ell$. 
This introduces a fresh null, $u_{\sigma, \Rule, z_2}=e_{wz}^i$, and adds the atoms $\PrCopy(e_{w}^i, e_{w}^{\ell'}, e_{wz}^i, e_{wz}^\ell)$, $\PrCell{1}(e_{wz}^{i-1})=\PrCell{b_{i-1}}(e_{wz}^{i-1})$, and $\PrNext(e_{wz}^{i-1}, e_{wz}^i)$. 

\item $\Rule =\eqref{rul_CpyPlusOneCell1Next}$  and 
$\sigma$ as above. This introduces a fresh null, $u_{\sigma, \Rule, z_2}=e_{wz}^i$, and adds the atoms $\PrCopyPlusOne(e_{w}^i, e_{w}^{\ell'}, e_{wz}^i, e_{wz}^\ell)$, $\PrCell{0}(e_{wz}^{i-1})=\PrCell{b_{i-1}}(e_{wz}^{i-1})$, and $\PrNext(e_{wz}^{i-1}, e_{wz}^i)$. 
\end{itemize}

\item If an atom of the form $\PrCopy(e_{w}^{j}, e_{w}^{\ell'}, e_{wz}^{j}, e_{wz}^\ell)$ is introduced at some point (i.e. the applicable $\Tuple{\Rule, \sigma}$ in the previous step is such that $\Rule=\eqref{rul_CpyPlusOneCell0Next}$), we use a sequence of chase steps similar to the case $b'_1=0$, and obtain the required atoms.

\item Otherwise, $\PrCopyPlusOne(e_{w}^{\ell'-1}, e_{w}^{\ell'}, e_{wz}^{\ell'-1}, e_{wz}^\ell)$ is introduced at the end of the loop. Apply the applicable  $\Tuple{\Rule, \sigma}$ among: 
\begin{itemize}
\item $\Rule=\eqref{rul_CpyPlusOneCell0}$ 
and $\sigma(y_1)= e_{w}^{\ell'-1}$, $\sigma(y_2)= e_{w}^{\ell'}$, $\sigma(z_1)=e_{wz}^{\ell'-1}$, $\sigma(z_\_)=e_{wz}^\ell$ with $\ell'=\ell$. 
This adds $\PrCell{1}(e_{wz}^{\ell-1})$, $\PrNext(e_{wz}^{\ell-1}, e_{wz}^\ell)$, and $\PrCell{1}(e_{wz}^\ell)$.

\item $\Rule=\eqref{rul_CpyPlusOneCell1}$ and $\sigma$ is as above but $\ell=\ell'+1$ so that $\sigma(z_1)=e_{wz}^{\ell'-1}=e_{wz}^{\ell-2}$. 
This introduces a new null $u_{\sigma, \Rule, z_2}=e_{wz}^{\ell-1}$ and 
adds $\PrCell{0}(e_{wz}^{\ell-2})$, $\PrNext(e_{wz}^{\ell-2}, e_{wz}^{\ell-1})$, $\PrCell{0}(e_{wz}^{\ell-1})$, $\PrNext(e_{wz}^{\ell-1}, e_{wz}^\ell)$ and $\PrCell{1}(e_{wz}^\ell)$.\qedhere
\end{itemize}
\end{itemize}
\end{itemize}
\end{proof}

\noindent\textbf{Lemma \ref{lem:univ-model-set-3}.} 
\textit{$\IS_3 = \{\FunInterpretation{3}(\I)\mid\I\in\IS_2\}$ is a universal model set of $\R_3$ and $\D$.}
\begin{proof}
Similar to the proof of Lemma \ref{lem:univ-model-set-2}, using Lemmas~\ref{lem:univ-model-set-2} and \ref{app:int3-datalog-first-chase}. 
\end{proof}

\begin{lemma}\label{app:int4-datalog-first-chase}
Given some $\I\in \IS_3$, $\FunInterpretation{4}(\I)$ is well-defined and $\FunInterpretation{4}(\I)$ is isomorphic to the unique interpretation in a result of the chase over $\R_4$ and $\I$.
\end{lemma}
\begin{proof}
Before giving the extension of $\FunInterpretation{4}(\I)$, we introduce some notation. 
Consider some fact of the form $\PrLeaf(u_w)\in\I$ and let $S=\FF{branchTape}(\I,u_w)$. By construction, $S=e_1\cdots e_{|\FF{branchDb}(\mathcal{I},u_w)|}$ where each $e_i$ corresponds to the serialisation of some $\PrP^x_g(u_{w_1}, \dots, u_{w_x})\in \FF{branchDb}(\mathcal{I},u_w)$ and is of the form $\PrP^x_g\Separator b^{1}_1\cdots b^{1}_{\ell_{1}}\Separator\cdots\Separator b^{x}_1\cdots b^{x}_{\ell_{x}}\Separator$ where $b^{i}_1\cdots b^{i}_{\ell_{i}}$ is the binary representation of $|w_i|+1$.
For every $\PrP^x_g\in\mathcal{S}$ and $u_{w_1}, \dots, u_{w_x}\subseteq  \FF{branch}(u_w)$, let 
$j^s_{S}(\PrP^x_g,u_{w_1}, \dots, u_{w_x})$ and $j^e_{S}(\PrP^x_g,u_{w_1}, \dots, u_{w_x})-1$ be the indexes where the serialisation of $\PrP^x_g(u_{w_1}, \dots, u_{w_x})$ starts and ends respectively if  $\PrP^x_g(u_{w_1}, \dots, u_{w_x})\in\FF{branchDb}(\mathcal{I},u_w)$, and otherwise $j^s_{S}(\PrP^x_g,u_{w_1}, \dots, u_{w_x})=j^e_{S}(\PrP^x_g,u_{w_1}, \dots, u_{w_x})$ be the index where starts the serialisation of the first fact of $\FF{branchDb}(\mathcal{I},u_w)$ that follows $\PrP^x_g(u_{w_1}, \dots, u_{w_x})$ according to $\prec$ (and $|S|+1$ if there is no such fact). (Recall that $\prec$ is defined before Lemma \ref{lem:univ-model-set-4}.) 
It is easy to verify that:
\begin{itemize}
\item $j^s_{S}(\PrP^1_1,u_{w})=1$.

\item If $\PrFreshIDBInTape{\PrP^x_g}(u_{w_1}, \dots, u_{w_x})\notin\I$, then  $j^e_{S}(\PrP^x_g,u_{w_1}, \dots, u_{w_x})=j^s_{S}(\PrP^x_g,u_{w_1}, \dots, u_{w_x})$.

\item If $\PrFreshIDBInTape{\PrP^x_g}(u_{w_1}, \dots, u_{w_x})\in\I$, then  $j^e_{S}(\PrP^x_g,u_{w_1}, \dots, u_{w_x})=j^s_{S}(\PrP^x_g,u_{w_1}, \dots, u_{w_x})+\Sigma_{i=1}^x  \ell_{i}+x+2$, where $\ell_i$ is the length of the binary encoding of $|w_i|+1$.

\item If $g<\PredsPerArity$, then $j^e_{S}(\PrP^x_{g},u_{w_1}, \dots, u_{w_x})=j^s_{S}(\PrP^x_{g+1},u_{w_1}, \dots, u_{w_x})$.

\item  If $u_{w_1}, \dots, u_{w_x}\neq u_{1}, \dots, u_{1}$, then $j^e_{S}(\PrP^{x}_{\PredsPerArity},u_{w_1}, \dots, u_{w_x})=j^s_{S}(\PrP^x_1,u_{w_1'}, \dots, u_{w_x'})$, where $u_{w'_1}, \dots, u_{w_x'}$ comes right after $u_{w_1}, \dots, u_{w_x}$ in the lexicographic order according to $\prec$. 

\item  If $x<\MaxArity$, then $j^e_{S}(\PrP^{x}_{\PredsPerArity},u_{1}, \dots, u_{1})=j^s_{S}(\PrP^{x+1}_1,u_{w}, \dots, u_{w})$.
\end{itemize}
Note that $j^e_{S}(\PrP^\MaxArity_\PredsPerArity,u_{1}, \dots, u_{1})=|S|+1$ and that for every $\PrP^x_g(u_{w_1}, \dots, u_{w_x})\in\FF{branchDb}(\mathcal{I},u_w)$, if $b^i_1\cdots b^i_{\ell_i}$ is the binary encoding of $|w_i|+1$ ($1\leq i\leq x$), then :
\begin{itemize}
\item $S[j^s_{S}(\PrP^x_g,u_{w_1}, \dots, u_{w_x})]=\PrP^x_g$, 
\item $S[j^s_{S}(\PrP^x_g,u_{w_1}, \dots, u_{w_x})+1]=\Separator$, 
\item and for $1\leq i\leq x$:
\begin{itemize}
\item $S[j^s_{S}(\PrP^x_g,u_{w_1}, \dots, u_{w_x})+i+\Sigma_{j=1}^{i-1}\ell_j+1]=b^i_1$,
\item$ \dots$, 
\item $S[j^s_{S}(\PrP^x_g,u_{w_1}, \dots, u_{w_x})+i+\Sigma_{j=1}^{i-1}\ell_j+\ell_i]=b^i_{\ell_i}$,
\item $S[j^s_{S}(\PrP^x_g,u_{w_1}, \dots, u_{w_x})+i+\Sigma_{j=1}^{i}\ell_j+1]=\Separator$.
\end{itemize}
\end{itemize}
Finally, let $J_{S}(u_{w'})$ be the set of indexes of $S$ where starts the sequence $\Separator b'_1\cdots b'_{\ell'} \Separator$ where $b'_1\cdots b'_{\ell'}$ is the binary representation of $|w'|+1$. 
The interpretation $\FunInterpretation{4}(\I)$ is equal to the union of the following databases:
\begin{enumerate}
\item  \I;
\item For every $\PrLeaf(u_w)\in\I$ with serialisation $S=\FF{branchTape}(\I,u_w)$, the database $\FF{startConf}(\mathcal{I},u_w)$:
\begin{align*}
&\{\PrLoad{1}(u_w, t_w^1, u_w), \PrHead{\StartingState}(t^1_w), \PrEndTape(t_w^{|S|+1}),\PrCell{\Blank}(t_w^{|S|+1})\} \cup{}\\& \{\PrNext(t_w^{j-1}, t_w^{j}) \mid 2\leq j \leq |S|+1\} \cup
\{\PrCell{a}(t_w^{j}) \mid  1\leq j\leq |S|, a=S[j]\}.
\end{align*}

 \item  For every $\PrLeaf(u_w)\in\I$ with serialisation $S=\FF{branchTape}(\I,u_w)$, the database $\FF{loadS}(\mathcal{I},u_w)$ defined as follows: 
\begin{align*}
\{\PrLoad{x}(u_{w}, t^{j}_{w},  u_{w_1}, \dots, u_{w_x}) \mid& 1\leq x \leq \MaxArity; u_{w_i}\in \FF{branch}(u_w); j=j^s_{S}(\PrP^x_1,u_{w_1}, \dots, u_{w_x})\}
\\\cup
\{\PrReady{x}(u_{w}, t^{j}_{w},  u_{w_1}, \dots, u_{w_x})\mid& 1\leq x \leq \MaxArity;u_{w_i}\in \FF{branch}(u_w); j=j^e_{S}(\PrP^x_\PredsPerArity,u_{w_1}, \dots, u_{w_x})\}
\\\cup
\{ \PrLoad{\PrP^x_g}(u_{w}, t^{j}_{w},  u_{w_1}, \dots, u_{w_x}) \mid& 1\leq g\leq \PredsPerArity;1\leq x \leq \MaxArity; u_{w_i}\in \FF{branch}(u_w); j=j^s_{S}(\PrP^x_g,u_{w_1}, \dots, u_{w_x})\}
\\\cup
\{\PrReady{\PrP^x_g}(u_{w}, t^{j}_{w},  u_{w_1}, \dots, u_{w_x}) \mid& 1\leq g\leq \PredsPerArity;1\leq x \leq \MaxArity; u_{w_i}\in \FF{branch}(u_w);  j=j^e_{S}(\PrP^x_g,u_{w_1}, \dots, u_{w_x})\}\\
\cup
\{\PrLoadEncoding(u_{w'}, t^{j+1}_{w}, t^{j+\ell'+2}_{w}) \mid& 
u_{w'}\in \FF{branch}(u_w); b_1'\cdots b_{\ell'}' \text{ bin. rep. of }|w'|+1;
j\in J_{S}(u_{w'})
\} 
\\\cup 
\{
\PrCopy(e_{w'}^r, e_{w'}^{\ell'}, t^{j+r+1}_{w}, t^{j+\ell'+1}_{w})
\mid& u_{w'}\in \FF{branch}(u_w); b_1'\cdots b_{\ell'}' \text{ bin. rep. of }|w'|+1;
\\&
\qquad 1\leq r\leq \ell'-1;
j-1\in J_{S}(u_{w'})\}
\end{align*}
\end{enumerate}
We describe a (Datalog-first) chase over $\R_4$ and $\I$ step by step. We directly set the the names of the fresh nulls so that the chase result coincides with $\FunInterpretation{4}(\I)$. 
Note that no rule in Figures \ref{figure:rule-set-1}, \ref{figure:rule-set-2}, or \ref{figure:rule-set-3} is applicable  to \I but that applying rule \eqref{rul_encodeFact} may trigger the application of the ``copy'' rules \eqref{rul_copyBase} and \eqref{rul_copyRec} in Figure\ref{figure:rule-set-3}.
\smallskip

For every $\PrLeaf(u_w)\in\I$ with serialisation $S=\FF{branchTape}(\I,u_w)$: 
\begin{itemize}
\item Apply $\Tuple{\Rule, \sigma}$ with $\Rule = \eqref{rul_startTape}$ and $\sigma$ that maps $u$ to $u_{w}$. 
This introduces a fresh null that we call $t^{1}_{w}$ and adds atoms $\PrLoad{1}(u_{w}, t^{1}_{w}, u_{w})$ and $\PrHead{\StartingState}(t^{1}_{w})$. Note that $j^s_{S}(\PrP^1_1,u_{w})=1$ so that we have indeed added $\PrLoad{1}(u_{w}, t^{j^s_{S}(\PrP^1_1,u_{w})}_{w}, u_{w})$.

\item Let $\FF{branch}(u_w)=\{u_{w_1},\dots,u_{w_p}\}$ with $u_{w}=u_{w_1}\prec \dots\prec u_{w_p}=u_1$ (i.e. $|w|=|w_1|>\dots>|w_p|=1$). For $1\leq i\leq p$:
\begin{itemize}
\item Apply \eqref{rul_firstPredEll} to  $\PrLoad{1}(u_{w}, t^{j^s_{S}(\PrP^1_1,u_{w_i})}_{w}, u_{w_i})$ and add $\PrLoad{\PrP^1_1}(u_{w}, t^{j^s_{S}(\PrP^1_1,u_{w_i})}_{w}, u_{w_i})$.

\item For $1\leq g\leq \PredsPerArity$ and $j=j^s_{S}(\PrP^1_g,u_{w_i})$:
\begin{itemize}

\item If $\PrFreshIDBNotInTape{\PrP^1_g}(u_{w_i})\in\I$, apply $\eqref{rul_copyNotInPred}$ and add $ \PrReady{\PrP^1_g}(u_{w}, t^{j}_{w}, u_{w_i})$. Note that $j=j^s_{S}(\PrP^1_g,u_{w_i})=j^e_{S}(\PrP^1_g,u_{w_i})$.

\item Else, $\PrFreshIDBInTape{\PrP^1_g}(u_{w_i})\in\I$: apply $\Tuple{\Rule, \sigma}$ with $\Rule $ the instantiation of $\eqref{rul_loadAllFacts} $ for $\PrP^1_g$, and $\sigma(u)=u_{w}$, $\sigma(t)=t^{j}_{w}$ and $\sigma(v)=u_{w_i}$. This introduces three fresh nulls $u_{\Rule,\sigma,x_1}=t^{j+1}_{w}$, $u_{\Rule,\sigma,x_2}=t^{j+\ell_i+2}_{w}$ and $u_{\Rule,\sigma,y}= t^{j+\ell_i+3}_{w}$, where $b_1^i\cdots b_{\ell_i}^i$ is the binary representation of $|w_i|+1$. 
This adds atoms $\PrCell{\PrP^1_g}(t^{j}_{w})$, $\PrNext(t^{j}_{w}, t^{j+1}_{w})$, $\PrLoadEncoding(u_{w_i}, t^{j+1}_{w}, t^{j+\ell_i+2}_{w})$, $\PrNext(t^{j+\ell_i+2}_{w}, t^{j+\ell_i+3}_{w})$ and $\PrReady{\PrP^1_g}(u_{w},  t^{j+\ell_i+3}_{w}, u_{w_i})$. 
 Note that $j+\ell_i+3=j^s_{S}(\PrP^1_g,u_{w_i})+\ell_i+3=j^e_{S}(\PrP^1_g,u_{w_i})$. 

\item If $g<\PredsPerArity$: 
Apply \eqref{rul_nextPredEll} and add $\PrLoad{\PrP_{g+1}^{1}}(u_{w},  t^{j'}_{w}, u_{w_i})$ with $j'=j$ if $\PrFreshIDBNotInTape{\PrP^1_g}(u_{w_i})\in\I$, $j'=j+\ell_i+3$ otherwise. 
Note that $j'=j^e_{S}(\PrP^1_g,u_{w_i})=j^s_{S}(\PrP^1_{g+1},u_{w_i})$,  independently of whether $\PrFreshIDBInTape{\PrP^1_{g}}(u_{w_i})$ is in $\I$.

\end{itemize}

\item Apply \eqref{rul_lastPredEll} and add $\PrReady{1}(u_{w},  t^{j^e_{S}(\PrP^1_\PredsPerArity,u_{w_i})}_{w}, u_{w_i})$. 

\item If $i<p$ (i.e. $u_{w_i}\neq u_1$): Apply the instantiation of \eqref{rul_tapeNextVector} 
$\PrReady{1}(u, t, v) \wedge \PrHasChild(w,v) \to \PrLoad{1}(u, t, w)$
 and add  $\PrLoad{1}(u_{w}, t^{j'}_{w},u_{w_{i+1}})$ with $j'=j^e_{S}(\PrP^1_\PredsPerArity,u_{w_i})=j^s_{S}(\PrP^1_1,u_{w_{i+1}})$. Recall that $u_{w_{i+1}}$ comes right after $u_{w_{i}}$ in the lexicographic order according to $\prec$.

\end{itemize}

\item At the end of the above loop, we obtain an atom $\PrReady{1}(u_{w},  t^{j^e_{S}(\PrP^1_\PredsPerArity,u_{1})}_{w}, u_1)$. 
Apply the instantiation of \eqref{rul_tapeNextLevel}  $\PrReady{1}(u, t, v) \wedge \PrRoot(v) \to \PrLoad{2}(u, t, u,u)$ and add $\PrLoad{2}(u_{w}, t^{j^e_{S}(\PrP^1_\PredsPerArity,u_{1})}, u_{w}, u_{w}) $. Note that $j^e_{S}(\PrP^1_\PredsPerArity,u_{1})=j^s_{S}(\PrP^2_1,u_{w},u_{w})$. 

\item Repeat the process, so that we obtain the atoms of the following form for $2\leq x\leq \MaxArity$ and $u_{w_1},\dots, u_{w_x}\in \FF{branch}(u_w)$:
\begin{itemize}
\item $\PrLoad{x}(u_{w}, t^{j}_{w}, u_{w_1}, \dots, u_{w_x})$ where $j=j^s_{S}(\PrP^x_1,u_{w_1}, \dots, u_{w_x})$;
\item $\PrReady{x}(u_{w},  t^{j}_{w},  u_{w_1}, \dots, u_{w_x})$ where  $j=j^e_{S}(\PrP^x_\PredsPerArity,u_{w_1}, \dots, u_{w_x})$;
\item  for $1\leq g\leq \PredsPerArity$: 
\begin{itemize}
\item $\PrLoad{\PrP^x_g}(u_{w}, t^{j}_{w},  u_{w_1}, \dots, u_{w_x})$ where $j=j^s_{S}(\PrP^x_g,u_{w_1}, \dots, u_{w_x})$;
\item $ \PrReady{\PrP^x_g}(u_{w}, t^{j}_{w},  u_{w_1}, \dots, u_{w_x})$ where $j=j^e_{S}(\PrP^x_g,u_{w_1}, \dots, u_{w_x})$;
\item for $\PrFreshIDBInTape{\PrP^x_g}(u_{w_1}, \dots, u_{w_x})\in\I$ and $j=j^s_{S}(\PrP^x_g,u_{w_1}, \dots, u_{w_x})$: 
\begin{itemize}
\item $\PrCell{\PrP^x_g}(t^{j}_{w})$, 
\item $\PrNext(t^{j}_{w}, t^{j+1}_{w}),\dots, \PrNext(t^{j+\Sigma_{l=1}^x  \ell_{l}+x+1}_{w}, t^{j+\Sigma_{l=1}^x  \ell_{l}+x+2}_{w})$, 
\item $\PrLoadEncoding(u_{w_1}, t^{j+1}_{w}, t^{j+\ell_{1}+2}_{w}),\dots, \PrLoadEncoding(u_{w_x}, t^{j+\Sigma_{l=1}^{x-1}  \ell_{l}+x}_{w}, t^{j+\Sigma_{l=1}^x  \ell_{l}+x+1}_{w})$.
\end{itemize}
\end{itemize}
\end{itemize}
\item We have obtained an atom $\PrReady{\MaxArity}(u_{w},  t^{j}_{w}, u_1, \dots, u_1)$ with $j=j^e_{S}(\MaxArity,u_{1}, \dots, u_{1})=|S|+1$. Apply \eqref{rul_endTape} and add $\PrCell{\Blank}(t^{|S|+1}_{w})$ and $\PrEndTape(t^{|S|+1}_{w})$.

\item For all $u_{w_i}\in\FF{branch}(u_w)$ and atoms of the form $\PrLoadEncoding(u_{w_i}, t^{j+1}_{w}, t^{j+\ell_i+2}_{w})$: 
\begin{itemize}
\item Apply $\Tuple{\Rule, \sigma}$ with $\Rule =\eqref{rul_encodeFact}$ and $\sigma(v)=u_{w_i}$, $\sigma(x_s)=t^{j+1}_{w}$, $\sigma(x_e)=t^{j+\ell_i+2}_{w}$, $\sigma(y_1)=e_{w_i}^1$ and $\sigma(y_\_)=e_{w_i}^{\ell_i}$ where $b_1^i \cdots b_{\ell_i}^i$ is the binary representation of $|w_i|+1$. 
This introduces two fresh nulls $u_{\Rule,\sigma,z_1}= t^{j+2}_{w}$ and $u_{\Rule,\sigma,z_\_}= t^{j+\ell_i+1}_{w}$ and adds the atoms: $\PrCell{\Separator}(t^{j+1}_{w})$, $\PrNext(t^{j+1}_{w}, t^{j+2}_{w})$, $\PrCopy(e_{w_i}^1, e_{w_i}^{\ell_i}, t^{j+2}_{w}, t^{j+\ell_i+1}_{w})$, $\PrNext(t^{j+\ell_i+1}_{w}, t^{j+\ell_i+2}_{w})$ and $ \PrCell{\Separator}(t^{j+\ell_i+2}_{w})$.

\item Apply exhaustively rules \eqref{rul_copyBase} and \eqref{rul_copyRec}. This process introduces nulls $ t^{j+3},\dots,  t^{j+\ell_i} $ and adds the following atoms (cf. proof of Lemma \ref{app:int3-datalog-first-chase}):
\begin{itemize}
\item $\PrCell{b^i_1}(t^{j+2}_{w}), \dots, \PrCell{b^i_{\ell_i}}(t^{j+\ell_i+1}_{w})$, 
\item $\PrNext(t^{j+2}_{w}, t^{j+3}_{w}), \dots,\PrNext(t^{j+\ell_i}_{w}, t^{j+\ell_i+1}_{w}) $  
\item $\PrCopy(e_{w_i}^2, e_{w_i}^{\ell_i}, t^{j+3}_{w}, t^{j+\ell_i+1}_{w})$, $\dots$, $\PrCopy(e_{w_i}^{\ell_i-1}, e_{w_i}^{\ell_i}, t^{j+\ell_i}_{w}, t^{j+\ell_i+1}_{w})$.
\end{itemize}
\end{itemize}
\end{itemize}
At this point, for every null $u_{w}$ such that $\PrLeaf(u_{w})\in\I$, we have added $\FF{startConf}(\mathcal{I},u_w)$ and $\FF{loadS}(\mathcal{I},u_w)$ and no rule is applicable. 
\end{proof}

\noindent\textbf{Lemma \ref{lem:univ-model-set-4}.} 
\textit{$\IS_4 = \{\FunInterpretation{4}(\I) \mid \I\in\IS_3\}$ is a universal model set of $\R_4$ and $\D$.}

\begin{proof}
Similar to the proof of Lemma \ref{lem:univ-model-set-3}, using Lemmas~\ref{lem:univ-model-set-3} and \ref{app:int4-datalog-first-chase}. 
\end{proof}

Before proving Lemma \ref{lem:univ-model-set-5}, we introduce some notation about TM runs. 
A \emph{configuration} for a TM $\TM = \Tuple{\States, \Alphabet, \TransitionFunction}$ is a finite sequence $\VC = C_1, \ldots, C_k$ where  $C_i = \{q, a\}$ for some $i \leq k$, $q \in \States$, and $a \in \Alphabet$, while for each $j \neq i$, $C_j = \{a_j\}$ for some $a_j \in \Alphabet$. 
The \emph{successor configuration} of $\VC$ under $\TM$, denoted $\TM(\VC)$, is the  configuration $\VC'=C'_1, \ldots, C'_{k+1}$ satisfying, assuming $\TransitionFunction(q,a) = (q',a',l)$, 
\[
C'_j \cap \Alphabet = \left\{ 
\begin{array}{ll}
\{a'\} & \text{ if } j = i ,\\
\{\Blank\} & \text{ if } j = k+1, \\
C_j \cap \Alphabet & \text{ otherwise,}
\end{array}
\right.
\]
while $C'_j \cap \States = \{q'\}$ if $j =  \max(1,i+l)$, otherwise $\emptyset$. 
For a word $\vec{\gamma} = \gamma_1, \ldots, \gamma_k$, let $\FunStartingConf(\TM, \vec{\gamma})$ denote the sequence $\{\StartingState, \gamma_1\}$, $\{\gamma_2\}$, $\ldots, \{\gamma_k\}, \{\Blank\}$.
Let $\FunComputation(\TM, \vec{\gamma}) = \VC_1, \VC_2, \ldots$ be the (possibly finite) configuration sequence satisfying
$\VC_1 = \FunStartingConf(\TM, \vec{\gamma})$ as well as $\VC_{i+1}=\TM(\VC_i)$, whenever defined (otherwise -- and only then -- the sequence ends at $i$). 
Note that \TM halts on $\vec{\gamma}$ if $\FunComputation(\TM, \vec{\gamma})$ is finite, and accepts $\vec{\gamma}$ if the last configuration in $\FunComputation(\TM, \vec{\gamma})$ features \AcceptingState.

\begin{lemma}\label{app:int5-datalog-first-chase}
Given some $\I\in \IS_4$, $\FunInterpretation{5}(\I)$ is isomorphic to the unique interpretation in a result of the chase over $\R_5$ and $\I$.
\end{lemma}
\begin{proof}
The interpretation $\FunInterpretation{5}(\I)$ is equal to the union of the following interpretations:
\begin{enumerate}
\item  \I;

\item  For each $w\in\mathbf{Z}$ with $b_1 \cdots b_\ell$ the binary representation of $|w| + 1$: 
$\{\PrNextTrans(e_{w}^{i}, e_{w}^j) \mid 1\leq i<j\leq \ell \}$.

\item For each $u_{w} \in \EI{\Nulls}{\I}$ with $S=\FF{branchTape}(\I,u_w)$, $\FunComputation(\TM, S) = \VC_1, \ldots, \VC_\ell$, and $\VC_i = C_{i, 1}, \ldots, C_{i, \vert S \vert + i}$ for each $1 \leq i \leq \ell$, the database $\FF{Run}(u_w)$ defined as follows: 
\begin{align*}
\{&\PrNextTrans(t_{w}^{i}, t_{w}^{j}) \mid 1 \leq i < j \leq \vert S \vert + 1\}\cup~ \\
\{&\PrEndTape(t_{w}^{i, \vert S  \vert + i}) \mid  2 \leq i \leq \ell\} \cup~ \\
\{&\PrNext(t_{w}^{i, j-1}, t_{w}^{i, j}) \mid 2 \leq i \leq \ell, 2 \leq j \leq \vert S  \vert + i\} \cup~ \\
\{&\PrNextTrans(t_{w}^{i, j}, t_{w}^{i, j'}) \mid 1 \leq i \leq \ell, 1 \leq j < j' \leq \vert S \vert + i\}\cup~ \\
\{&\PrStep(t_{w}^{i-1, j}, t_{w}^{i, j}) \mid 3 \leq i \leq \ell, 1 \leq j \leq \vert S \vert + i -1
 \} \cup~ \\
\{&\PrCell{a}(t_{w}^{i, j}) \mid  2 \leq i \leq \ell, 1 \leq j \leq \vert S \vert + i, a \in C_{i, j} \cap \Alphabet\} \cup~\\
\{&\PrHead{q}(t_{w}^{i, j}) \mid  2 \leq i \leq \ell, 1 \leq j \leq \vert S \vert + i, q \in C_{i, j} \cap \States\} \cup~ \\
\{&\PrGoal \mid 1 \leq j \leq \vert S \vert + \ell, \AcceptingState \in C_{\ell, j} \};
\end{align*}
\end{enumerate}

We describe a (Datalog-first) chase over $\R_5$ and $\I$ step by step. We directly set the the names of the fresh nulls so that the chase result coincides with $\FunInterpretation{5}(\I)$. Note that no rule in Figures \ref{figure:rule-set-1}, \ref{figure:rule-set-2}, \ref{figure:rule-set-3}, \ref{figure:rule-set-4}, or \ref{figure:rule-set-4-part} is applicable  to \I and that applying a rule in Figure~\ref{figure:rule-set-5} cannot make a rule in Figures \ref{figure:rule-set-1}, \ref{figure:rule-set-2}, \ref{figure:rule-set-3},  \ref{figure:rule-set-4}, or \ref{figure:rule-set-4-part} applicable.

We first apply exhaustively \eqref{rul_tm_nexttrans_init} and \eqref{rul_tm_nexttrans_rec} which compute a transitive closure on the atoms of the form $\PrNext(e_{w}^{i}, e_{w}^j)$ and $\PrNext(t^{i}_{w}, t^{j}_{w})$ in $\I$. This adds the atoms $\{\PrNextTrans(e_{w}^{i}, e_{w}^j) \mid 1\leq i<j\leq \ell\}$ described in point (2) and $\{\PrNextTrans(t_{w}^{i}, t_{w}^{j}) \mid 1 \leq i < j \leq \vert S \vert + 1\}$ of the first line in point (3).

For every null $u_{w}$ such that $\PrLeaf(u_{w})\in\I$, $S=\FF{branchTape}(\I,u_w)$ and $\FunComputation(\TM, S) = \VC_1, \ldots, \VC_\ell$, we show by induction that for every $2\leq i\leq \ell$ the following atoms belong to the chase result: 
\begin{itemize}
\item $\PrEndTape(t^{i,\vert S\vert+i}_{w})$, 
\item $\PrNext(t^{i,j-1}_{w}, t^{i,j}_{w})$ for $2\leq j\leq \vert S\vert+i$, 
\item $\PrNextTrans(t^{i,j}_{w}, t^{i,j'}_{w})$ for $1\leq j<j'\leq \vert S\vert+i$, 
\item $\PrStep(t_{w}^{i-1, j}, t_{w}^{i, j})$ for $1 \leq j \leq \vert S \vert + i -1$,
\item $\PrCell{a}(t^{i,j}_{w})$ for $1 \leq j \leq \vert S \vert + i, a \in C_{i, j} \cap \Alphabet$,
\item $\PrHead{q}(t^{i,j}_{w})$ for $1 \leq j \leq \vert S \vert + i, q \in C_{i, j} \cap \States$.
\end{itemize}
\noindent\emph{Base case $i=2$.} 
Recall that for every $\PrLeaf(u_w)\in\I$ with serialisation $S=\FF{branchTape}(\I,u_w)$, $\I$ includes:
\begin{align*}
&\{\PrLoad{1}(u_w, t_w^1, u_w), \PrHead{\StartingState}(t^1_w), \PrEndTape(t_w^{|S|+1}),\PrCell{\Blank}(t_w^{|S|+1})\} \cup \{\PrNext(t_w^{j-1}, t_w^{j}) \mid 2\leq j \leq |S|+1\} \cup{}
\\ 
&
\{\PrCell{a}(t_w^{j}) \mid  1\leq j\leq |S|, a=S[j]\}.
\end{align*} 
Moreover $\VC_1=\FunStartingConf(\TM, S)=\{\StartingState, S[1]\}$, $\{S[2]\}$, $\ldots, \{S[|S|]\}, \{\Blank\}$. 
Let $(\StartingState, S[1]) \mapsto (r, b, +1) \in\TransitionFunction$  (recall that we assume that the TM will never attempt to move left on the first
position of the tape). 
By definition, we have $C_{2, 1} \cap\Alphabet=\{b\}$, $C_{2, j} \cap\Alphabet=C_{1, j} \cap\Alphabet =\{S[j]\}$ for every $1<j\leq |S|+1$, and $C_{2, 2} \cap\States=\{r\}$. 
\begin{itemize}
\item For every $1<j\leq |S|+1$, let $a \in C_{1, j} \cap \Alphabet =  C_{2, j} \cap \Alphabet$. 
\begin{itemize}
\item Apply $\Tuple{\Rule, \sigma}$ with $\Rule$ the instantiation of \eqref{rul_tm_right_mem} for transition $(\StartingState, S[1]) \mapsto (r, b, +1)$ and tape symbol $S[j]$, and $\sigma(x)=t^{1}_{w}$, $\sigma(y)=t^{j}_{w}$. This introduces a fresh null $u_{\Rule,\sigma,z}=t^{2,j}_{w}$ and adds $\PrStep(t^{j}_{w}, t^{2,j}_{w})$ and $\PrCell{S[j]}(t^{2,j}_{w})$. 
\end{itemize}

\item Apply $\Tuple{\Rule, \sigma}$ with $\Rule$ the instantiation of $\eqref{rul_tm_new_symb}$ for $(\StartingState, S[1]) \mapsto (r, b, +X)$ and $\sigma(x)=t^{1}_{w}$ . This introduces a fresh null $u_{\Rule,\sigma,z}=t^{2,1}_{w}$ and adds $\PrStep(t^{1}_{w}, t^{2,1}_{w})$ and $\PrCell{b}(t^{2,1}_{w})$. 

\item Apply $\Tuple{\Rule, \sigma}$ with $\Rule =\eqref{rul_tm_add_mem}$ with $\sigma(x)=t^{\vert S\vert+1}_{w}$ and $\sigma(z)=t^{2,\vert S\vert+1}_{w}$. This introduces a fresh null $u_{\Rule,\sigma,v}=t^{2,\vert S\vert+2}_{w}$ and adds $\PrNext(t^{2,\vert S\vert+1}_{w}, t^{2,\vert S\vert+2}_{w})$,  $\PrCell{\Blank}(t^{2,\vert S\vert+2}_{w})$ and $\PrEndTape(t^{2,\vert S\vert+2}_{w})$.

\item Apply exhaustively \eqref{rul_tm_new_nxt}, adding all atoms $\PrNext( t^{2,j-1}_{w},  t^{2,j}_{w})$ for $2\leq j\leq \vert S\vert+2$.

\item Apply exhaustively  \eqref{rul_tm_nexttrans_init} and \eqref{rul_tm_nexttrans_rec}, adding all atoms $\PrNextTrans( t^{2,j}_{w},  t^{2,j'}_{w})$ for $1\leq j<j'\leq \vert S\vert+2$.

\item  Apply \eqref{rul_tm_move_right} and add $\PrHead{r}(t^{2,2}_{w})$.
\end{itemize}

\noindent\emph{Induction step:} Assume the property is true for some $2\leq i<\ell$ and let $j_0$ be the unique $ C_{i, j_0}$ such that $ C_{i, j_0} \cap \States\neq\emptyset$, $\{q_0\}=C_{i, j_0} \cap\States$, $\{a_0\}=C_{i, j_0} \cap\Alphabet$, and $(q_0, a_0) \mapsto (r, b, +X) \in \TransitionFunction$. By definition, we have $C_{i+1, j_0} \cap\Alphabet=\{b\}$, $C_{i+1, j} \cap\Alphabet=C_{i, j} \cap\Alphabet$ for every $j\neq j_0$, and $C_{i+1, j_0+X} \cap\States=\{r\}$. 
\begin{itemize}
\item For every $j> j_0$ (resp. $j<j_0$), let $a \in C_{i, j} \cap \Alphabet =  C_{i+1, j} \cap \Alphabet$. 
\begin{itemize}
\item Apply $\Tuple{\Rule, \sigma}$ with $\Rule$ the instantiation of $\eqref{rul_tm_right_mem}$ (resp. \eqref{rul_tm_left_mem}) for transition $(q_0, a_0) \mapsto (r, b, +X)$ and tape symbol $a$, and $\sigma(x)=t^{i,j_0}_{w}$, $\sigma(y)=t^{i,j}_{w}$. This introduces a fresh null $u_{\Rule,\sigma,z}=t^{i+1,j}_{w}$ and adds $\PrStep(t^{i,j}_{w}, t^{i+1,j}_{w})$ and $\PrCell{a}(t^{i+1,j}_{w})$. 
\end{itemize}

\item Apply $\Tuple{\Rule, \sigma}$ with $\Rule$ the instantiation of $\eqref{rul_tm_new_symb}$ for $(q_0, a_0) \mapsto (r, b, +X)$ and $\sigma(x)=t^{i,j_0}_{w}$ . This introduces a fresh null $u_{\Rule,\sigma,z}=t^{i+1,j_0}_{w}$ and adds $\PrStep(t^{i,j_0}_{w}, t^{i+1,j_0}_{w})$ and $\PrCell{b}(t^{i+1,j_0}_{w})$. 

\item Apply $\Tuple{\Rule, \sigma}$ with $\Rule =\eqref{rul_tm_add_mem}$ with $\sigma(x)=t^{i,\vert S\vert+i}_{w}$ and $\sigma(z)=t^{i+1,\vert S\vert+i}_{w}$. This introduces a fresh null $u_{\Rule,\sigma,v}=t^{i+1,\vert S\vert+i+1}_{w}$ and adds $\PrNext(t^{i+1,\vert S\vert+i}_{w}, t^{i+1,\vert S\vert+i+1}_{w})$,  $\PrCell{\Blank}(t^{i+1,\vert S\vert+i+1}_{w})$ and $\PrEndTape(t^{i+1,\vert S\vert+i+1}_{w})$.

\item Apply exhaustively \eqref{rul_tm_new_nxt}, adding all atoms $\PrNext( t^{i+1,j-1}_{w},  t^{i+1,j}_{w})$ for $2\leq j\leq \vert S\vert+i+1$.

\item Apply exhaustively  \eqref{rul_tm_nexttrans_init} and \eqref{rul_tm_nexttrans_rec}, adding all atoms $\PrNextTrans( t^{i+1,j}_{w},  t^{i+1,j'}_{w})$ for $1\leq j<j'\leq \vert S\vert+i+1$.

\item  If $X=+1$, apply \eqref{rul_tm_move_right} and add $\PrHead{r}(t^{i+1,j_0+X}_{w})$.

\item If $X=-1$, then $j_0>1$ (since the TM never attempts to move left on the first position of the tape): apply \eqref{rul_tm_move_left} and add $\PrHead{r}(t^{i+1,j_0+X}_{w})$.
\end{itemize}
We have shown that all atoms of $\FF{Run}(u_w)$ but $\PrGoal$ belong to the chase result. 
Finally, for every null $u_{w}$ such that $\PrLeaf(u_{w})\in\I$, $S=\FF{branchTape}(\I,u_w)$ and $\FunComputation(\TM, S) = \VC_1, \ldots, \VC_\ell$, if there exists $j$ such that $\PrHead{\AcceptingState}(t^{\ell,j}_{w})\in\FF{Run}(u_w)$, apply \eqref{rul_tm_accept} and add $\PrGoal$.
\end{proof}

\noindent\textbf{Lemma \ref{lem:univ-model-set-5}.} 
\textit{$\IS_5 = \{\FunInterpretation{5}(\I) \mid \I \in \IS_4\}$ is a universal model set of $\R_5$ and $\D$.}

\begin{proof}
Similar to the proof of Lemma \ref{lem:univ-model-set-4}, using Lemmas~\ref{lem:univ-model-set-4} and \ref{app:int5-datalog-first-chase}. 
\end{proof}

\noindent\textbf{Lemma \ref{lemma:soundness}.}
\textit{If $\D\in\Query$, then $\PrGoal\in\mathcal{I}$ for each $\mathcal{I}\in\MS$.}

\begin{proof}
Let $\mathcal{I}\in\MS$. 
We have shown that there is a homomorphism $\D\to\FF{branchDb}(\mathcal{I},u_w)$ for the node $u_w$ where $|w|=|\FunOrder(\mathcal{I})|$. Since $\Query$ is closed under homomorphisms, $\D\in\Query$ implies $\FF{branchDb}(\mathcal{I},u_w)\in\Query$.

Let $\FunComputation(\TM, S) = \VC_1, \ldots, \VC_\ell$ be the computation of \TM on the serialisation $S=\FF{branchTape}(\I,u_w)$ of the database $\FF{branchDb}(\mathcal{I},u_w)$. 
Since $\FF{branchDb}(\mathcal{I},u_w)\in\Query$ and \TM decides $\Query$, $\VC_1, \ldots, \VC_\ell$ is such that $\VC_\ell$ features $\AcceptingState$. 
By construction of $\FF{Run}(u_w)$ (see proof of Lemma \ref{app:int5-datalog-first-chase}), it follows that $\PrGoal \in \FF{Run}(u_w)$. 
Hence $\PrGoal \in \I$.
\end{proof}

\noindent\textbf{Lemma \ref{lemma:completeness}.}
\textit{If $\D\notin\Query$, then $\PrGoal\notin\mathcal{I}$ for some $\mathcal{I}\in\MS$.}

\begin{proof}
Consider some $\mathcal{I}\in\MS$ such that $\FunDatabase(\mathcal{I})=\D$
and $\PrNeq(t, u) \in \mathcal{I}$ for each $t, u \in\EI{\Nulls}{\Database}$ with $t \neq u$ (such a \I exists by Lemma \ref{app-lem:various-claims-structure}).
Let $u_w$ denote the leaf node with $|w|=|\FunOrder(\mathcal{I})|$.
Then $\FF{branchDb}(\mathcal{I},u_w)$ is isomorphic to $\FunDatabase(\mathcal{I})=\D$ and $\FF{branchDb}(\mathcal{I},u_w)\notin\Query$.

Let $\FunComputation(\TM, S) = \VC_1, \ldots, \VC_\ell$ be the computation of \TM on the serialisation $S=\FF{branchTape}(\I,u_w)$ of the database $\FF{branchDb}(\mathcal{I},u_w)$. 
Since $\FF{branchDb}(\mathcal{I},u_w)\notin\Query$ and \TM decides $\Query$, $\VC_1, \ldots, \VC_\ell$ is such that $\VC_\ell$ does not feature $\AcceptingState$. 
By construction of $\FF{Run}(u_w)$ (see proof of Lemma \ref{app:int5-datalog-first-chase}), it follows that $\PrGoal \notin \FF{Run}(u_w)$. 

Moreover, for all other leaf nodes $u_v$ with $\PrLeaf(u_v)\in\mathcal{I}$, there is a
homomorphism $\FF{branchDb}(\mathcal{I},u_v)\to\FF{branchDb}(\mathcal{I},u_w)$.
Since $\Query$ is closed under homomorphisms, \TM does not accept any such $\FF{branchDb}(\mathcal{I},u_v)$, and $\PrGoal\notin\FF{Run}(u_v)$.  
By construction of \MS, \PrGoal occurs in \I iff it occur in some $\FF{Run}(u_v)$, so $\PrGoal\notin\mathcal{I}$.
\end{proof}

\section{Proofs for Section \ref{sec:disjunction removal}}
We fix for this whole appendix part a split $(\Sigma_1,\Sigma_2)$ of $\Sigma$ fulfilling the condition of Item 2 of Lemma~\ref{lemma-disjunction-removal}, and $\Sigma'_1, \Sigma'_2, \Sigma'_3$ built from $\Sigma_1,\Sigma_2$ as described in the body of the paper. 

\subsection*{Proof of Proposition~\ref{proposition-r1-disjunction-removal}}
\begin{definition}[World Structure]
 Let $\I$ be an interpretation of $\Sigma'_1$. The world structure of $\I$ is the graph $(V,E)$ where:
 \begin{itemize}
  \item $V = \{w \mid \PrDone(w) \in \I\}$
  \item $E = \{(w,w')\ \mid\ w \not = w' \wedge \exists \PrP \exists \Vx\  \PrIns{\PrP}(\Vx,w,w') \in \I\}$
 \end{itemize}
\end{definition}

\begin{lemma}
\label{lemma-finite-structure}
For every $\Database$ over $\mathcal{S}_{\mathsf{in}}(\Sigma)$, 
the world structure of a model $\I$ of $\Database$ and $\Sigma'_1$ generated by a chase sequence is a finite tree, such that:
 \begin{itemize}
  \item its root is the unique element $w$ such that $\PrEmpty(w) \in \I$;
  \item if $(w,w') \in E$, then $\mathsf{world}(w') = \mathsf{world}(w) \cup \{\PrP(\Vx)\}$ for some atom $\PrP(\Vx)$, and $\PrP(\Vx) \not \in \mathsf{world}(w)$;
  \item if $\PrInit(w) \in \I$ and $\PrP(\Va) \in \Database \setminus \mathsf{world}(w)$, there exists $w'$ such that $\PrIns{\PrP}(\Va,w,w') \in \I$.
 \end{itemize}
\end{lemma}

\begin{proof}
Let us first notice that Rule~(\ref{world-init}) is applicable exactly once, as its frontier is empty. The null it creates, say $w_\emptyset$, is then the only one such that $\PrEmpty(w_\emptyset) \in \I$. Moreover, there is no edge incoming in $w_\emptyset$ in the world structure, as the only rules Rules\ (\ref{collect-p}), (\ref{complete-r1}) and (\ref{complete-r2}) that introduce an atom of the shape $\PrIns{\PrP}(\Vx,w,w')$ with $w \not = w'$ are such that $w'$ is existentially quantified. For the same reason, note that there is at most one edge incoming in any vertex in the world structure. 

Next notice that no new null is created unless Rules\ (\ref{collect-p}), (\ref{complete-r1}) or (\ref{complete-r2}) are applied, and they all must be applied by mapping $w$ to an element of the world structure, as $\PrDone(w)$ is in the body of each of these rules. Hence the world structure is connected, and is thus a tree of root $w_\emptyset$.

Now let us notice that in any interpretation $\I$ generated by a sequence of rule applications, if $\PrDone(w) \in \I$, then $\mathsf{world}_\I(w) = \mathsf{world}_{\I'}(w)$ for any $\I'$ obtained by extending the sequence of rule applications that generated $\I$. Indeed, the only way to derive $\PrDone(w)$ is to apply Rule~(\ref{rule-done}), which requires the atoms $\PrEmpty(w_\emptyset)$ and $\PrSubs(w_\emptyset,w)$, as $w_\emptyset$ is the only element for which $\PrEmpty$ holds. The only way to create an atom of the shape $\PrSubs(w_0,w_2)$ with $w_0 \not = w_2$ is by applying an instantiation of Rule~(\ref{propagate-p}), which can be applied if $w_0$ is a parent of $w_1$ in the world structure, and $w_1$ is such that $\PrSubs(w_1,w_2)$ holds. Hence $\PrDone(w)$ is entailed only when Rule~(\ref{propagate-p}) as been applied by mapping $w_0$ to every ancestor of $w$ and $w_2$ to $w$, which has a effect to ensure that $\mathsf{world}_\I(w)$ contains all the atoms possibly present in $\mathsf{world}_{\I'}(w)$. 

If $(w,w') \in E$, then $w'$ has been created by the application of Rule (\ref{collect-p}), (\ref{complete-r1}) or (\ref{complete-r2}). In all cases, $\PrDone(w)$ must hold at the time of the rule application. If it is by Rule (\ref{collect-p}) and substitution $\sigma$, then $\PrP(\sigma(\Vx))$ cannot belong to $\mathsf{world}(w)$, as this would make Rule~\eqref{collect-p} not applicable with $\sigma$. 
By the sequence of rule applications described before, $\mathsf{world}(w') = \mathsf{world}(w) \cup \{\PrP(\sigma(\Vx))\}$. Rules~(\ref{complete-r1}) and (\ref{complete-r2}) are treated in a similar way.

Both the depth and the arity of the world structure is thus upper bounded by the number of atoms using a predicate appearing in $\Sigma_1 \cup \Sigma_2$ and nulls from $\Database$. 

For the last item, let us notice that if $\PrInit(w) \in \I$, and $\PrP(\Va) \in \Database \setminus \mathsf{world}(w)$, then Rule~(\ref{collect-p}) instantiated for $\PrP$ is applicable by mapping $\Vx$ to $\Va$. 
Applying that rule will create a fresh null $w'$ which will fulfill the conditions stated in the lemma. 
%
%
\end{proof}

\begin{lemma}
\label{lemma-disjunctive-to-hat}
Let $\Database$ be a database over $\mathcal{S}_{\mathsf{in}}(\Sigma)$,
and let $\I$ be a model of $\Database$ and $\Sigma_1'$ generated by a chase sequence. 
Let $v\in\EI{\Nulls}{\I}$ such that $\PrDone(v)\in\I$. 
If a disjunctive rule $\Rule$ is applicable to $\mathsf{world}(v)$ by $\sigma$, then there exist in $\I$ two elements $w_1$ and $w_2$ such that $\PrIns{\PrP_1}(\Vu_1,v,w_1)$ and $\PrIns{\PrP_2}(\Vu_2,v,w_2)$ such that $\PrP_1(\Vu_1)$ and $\PrP_2(\Vu_2)$ are the two atoms created by the application of $\Tuple{\Rule,\sigma}$.
\end{lemma}

\begin{proof}
Let us first notice that $\sigma$, extended by mapping $w$ to $v$, is a substitution that maps the bodies of Rules\ (\ref{complete-r1}) and (\ref{complete-r2}) to $\I$, by definition of $\mathsf{world}(v)$ and since $\PrDone(v)\in\I$. As $\I$ is a model of $\Sigma'_1$, the head of these two rules should be also mappable in $\I$, hence there exists $w_1$ and $w_2$ fulfilling the conditions of the lemma. 
\end{proof}

Note that Proposition \ref{proposition-r1-disjunction-removal} is a direct consequence of Lemmas\ \ref{lemma-finite-structure} and \ref{lemma-disjunctive-to-hat}.

\subsection*{Proof of Lemma~\ref{lemma-disjunction-removal}}
%

\begin{definition}[Saturated World of a Null in an Interpretation]
Let $w$ be a null in an interpretation $\I$ such that $\PrSubs(w,w) \in \I$. The saturated world of $w$ is defined by $\mathsf{saturatedWorld}_\I(w) = \{\PrP(\Vx) \mid \PrHatted{\PrP}(\Vx,w) \in \I\}$.
\end{definition}

\begin{lemma}
\label{lemma-context-to-existential}
 If there is a chase tree w.r.t. $(\Sigma'_1 \cup \Sigma'_2)$ and $\Database$ whose unique leaf $\I$ is such that there exists $w^*\in\EI{\Nulls}{\I}$ s.t. $\PrGoal \in \mathsf{saturatedWorld}_\I(w^*)$ and $\PrDone(w^*)\in\I$, then there exists a chase tree w.r.t $\Sigma_2$ and $\mathsf{world}(w^*)$ whose unique leaf contains $\PrGoal$.
\end{lemma}

\begin{proof}

Let us consider a chase tree $T'$ w.r.t. $(\Sigma'_1 \cup \Sigma'_2)$ and $\Database$, and let $w^*$ be such that $\PrDone(w^*)$ is in the result $\I$ of $T'$.
We show that there exists a chase tree $T$ w.r.t $\Sigma_2$ and $\mathsf{world}(w^*)$ such that $\mathsf{saturatedWorld}_\I(w^*)$ is mapped to the label of the (unique) leaf of $T$ by a homomorphism $\psi$. 

For every label $\J$ of a node in $T'$ such that $\PrDone(w^*)\in\J$, we show that there exists a prefix of a chase tree w.r.t. $\Sigma_2$ and $\mathsf{world}(w^*)$ such that $\mathsf{saturatedWorld}_\J(w^*)$ is mapped to the label of the leaf of that prefix by a homomorphism $\psi$. 
We do this by induction on the number of rule applications $\Tuple{\Rule,\Hom}$ with $\Rule$ of the form Rule~(\ref{rule-sigma-2}) 
and $\Hom(w) = w^*$ that are done  in $T'$ before the node labelled by $\J$.

\begin{itemize}
 \item If no rule application of the shape $\Tuple{\Rule,\Hom}$ with $\Hom(w) = w^*$ is performed in $T'$ before $\J$, then $\mathsf{saturatedWorld}_\J(w^*) = \mathsf{world}(w^*)$. We thus define $\psi$ as the identity, and the chase tree prefix consisting of the root labelled by $\mathsf{world}(w^*)$ fulfills the property. 
 
 \item Assume that the property is true for every $\J$ obtained in $T'$ after at most $i-1$ applications of the shape $\Tuple{\Rule',\Hom}$ with $\Rule'$ of the form Rule~(\ref{rule-sigma-2}) 
 and $\Hom(w) = w^*$. 
Let $\J_{i}$ be obtained in $T'$ after $i$ such applications and let $\Tuple{\Rule',\Hom}$ be the last one, with $\Rule'$ being the instantiation of Rule~(\ref{rule-sigma-2}) for $\Rule \in \Sigma_2$. Let $\J_{i-1}$ be the label of the node in $T'$ on which $\Tuple{\Rule',\Hom}$ is applied. By induction assumption, there exist $\psi$ and a prefix of a chase tree for $\Sigma_2$ and $\Database$, resulting in $S_{i-1}$, such that $\psi(\mathsf{saturatedWorld}_{\J_{i-1}}(w^*)) \subseteq S_{i-1}$. 
 Hence $\Tuple{\Rule,\psi\circ\Hom_{\mid \Terms(B_\Rule)}}$, where $B_\Rule$ is the body of $\Rule$, is applicable on $S_{i-1}$, creating a new leaf, labeled by $S_i$, and we extend $\psi$ by mapping every null created by the instantiation of $z_i\in\Vz$ in the application of $\Tuple{\Rule',\Hom}$ to the null created by the instantiation of $z_i\in\Vz$ in the application of $\Tuple{\Rule,\psi\circ\Hom_{\mid \Terms(B_\Rule)}}$. We obtain $\psi(\mathsf{saturatedWorld}_{\J_{i}}(w^*)) \subseteq S_i$. \qedhere
 \end{itemize}
\end{proof}

\begin{lemma}
\label{lemma-hat-reordering}
 Any restricted chase w.r.t $\Sigma'_1 \cup \Sigma'_2 \cup \Sigma'_3$ can be transformed into an equivalent restricted chase in which rules from ${\Sigma'_1}$ are applied before  rules from ${\Sigma'_2}$, which are applied before rules from ${\Sigma'_3}$.
\end{lemma}

\begin{proof}
Notice that head predicates of $\Sigma'_3$ do not appear in do not appear in ${\Sigma'_j}$ with $j<3$. Hence rules of $\Sigma'_3$ can be applied last. Note that head predicates of $\Sigma'_2$ do not appear as body predicates of $\Sigma'_1$. Moroever, a rule of $\Sigma'_2$ is applicable by mapping $w$ to $w^*$ only if $\PrDone(w^*)$ as been derived. Further applications of rules of $\Sigma'_1$ cannot add an atom of the shape $\PrHatted{\PrP}(\vec{x},w^*)$, and thus applying them before do not prevent the application of a rule of $\Sigma'_2$.
\end{proof}

%

\begin{proposition}
\label{proposition-completeness-disj-removal}
For every database $\Database$ over $\mathcal{S}_{\mathsf{in}}(\Sigma)$, 
 if $\Sigma_1 \cup \Sigma_2, \Database \models \PrGoal$, then $\Sigma'_1 \cup \Sigma'_2 \cup \Sigma'_3, \Database \models \PrGoal$.
\end{proposition}

\begin{proof}
 Let us consider a finite a chase tree $T$ proving that $\Sigma_1 \cup \Sigma_2, \Database \models \PrGoal$ and such that rules of $\Sigma_1$ are applied before rules of $\Sigma_2$ (this is possible by definition of $\Sigma_1$ and $\Sigma_2$ that form a split of $\Sigma$). We build a 
 chase tree proving that $\Sigma'_1 \cup \Sigma'_2 \cup \Sigma'_3, \Database \models \PrGoal$ (we actually describe a sequence of rule applications, as $\Sigma'_1 \cup \Sigma'_2 \cup \Sigma'_3$ contains only deterministic rules).
 \begin{itemize}
  \item Apply Rule\ (\ref{world-init}), creating a fresh null $w_\emptyset$ and facts $\PrInit(w_\emptyset)$, $\PrDone(w_\emptyset)$ and $ \PrEmpty(w_\emptyset)$.
  \item Until creation of $w_\Database$, perform the following:
  \begin{itemize}
   \item let $w_{D'}$ be the last introduced null;
   \item let $\PrP(\Va)\in\Database\setminus D'$:
  \begin{itemize}
   \item apply Rule\ (\ref{collect-p}), instantiated for $\PrP$, by mapping $\Vx$ to $\Va$ and $w$ to $w_{D'}$, creating a fresh null $w_{D'\cup\{\PrP(\Va)\}}$;
   \item apply Rule\ (\ref{propagate-p}) as many times as necessary until  {$\PrSubs(w_\emptyset,w_{D'\cup\{\PrP(\Va)\}})$ is created, which makes Rule~(\ref{rule-done}) become applicable by mapping $w'$ to $w_{D'\cup\{\PrP(\Va)\}}$. Then apply Rule~(\ref{rule-done}) and add} $\PrDone(w_{D'\cup\{\PrP(\Va)\}})$.
     \end{itemize}
  \end{itemize}
  \item At the end of the above loop, note that we have created null $w_{\Database}$  and added in particular $\PrDone(w_{\Database})$ as well as $\PrIns{\PrP}(\Va,w_{\Database},w_{\Database})$ and for every $\PrP(\Va)\in\Database$. 
 \item Let $r$ be the root node of $T$. Note that $r$ is labelled by $\Database$. In what follows, we will introduce some nulls $w_{\mathcal{N}}$ such that $\mathcal{N}$ is the label of some node $n$ in $T$ and it will remain true that $\PrDone(w_{\mathcal{N}})$ as well as $\PrIns{\PrP}(\Va,w_{\mathcal{N}},w_{\mathcal{N}})$ 
 for every $\PrP(\Va)\in\mathcal{N}$ are introduced between the creation of $w_{\mathcal{N}}$ and the creation of the next $w_{\mathcal{N}'}$.
 \item Define $\psi(r) = w_{\Database}$ and let $N = \{r\}$ and $L = \emptyset$. Until $N = \emptyset$, perform the following.
 \begin{itemize}
 \item Consider the case where $n$ is a node of $T$ labelled by $\mathcal{N}$ such that $\psi(n) = w_\mathcal{N}$ and $n$ has $2$ children $c_1$ and $c_2$ labelled by $\mathcal{N}_1$ and $\mathcal{N}_2$ respectively, which correspond to the application of $\Tuple{\Rule,\sigma}$ with $\Rule \in \Sigma_1$ of the form $\bigwedge_{\PrP(\vec{x}) \in \Body} \PrP(\vec{x}) \rightarrow \PrP_1(\vec{x_1}) \vee \PrP_2(\vec{x_2})$. 
We have $\mathcal{N}_1=\mathcal{N}\cup\{\PrP_1(\Va_1)\}$ and $\mathcal{N}_2=\mathcal{N}\cup\{\PrP_2(\Va_2)\}$ where $\Va_1=\sigma(\Vx_1)$ and $\Va_2=\sigma(\Vx_2)$.
The instantiations of \ (\ref{complete-r1}) and (\ref{complete-r2}) corresponding to $\Rule$ are applicable through $\sigma'$, which maps $w$ to $w_\mathcal{N}$ and $\Vx$ to $\sigma(\Vx)$.

\begin{itemize}
\item Apply Rule\ (\ref{complete-r1}), introducing a fresh null that we call $w_{\mathcal{N}_1}$ and adding in particular $\PrIns{\PrP_1}(\Va_1,w_{\mathcal{N}_1},w_{\mathcal{N}_1})$. 
Then apply Rule\ (\ref{propagate-p}) as many times as necessary, then Rule~(\ref{rule-done}) so that $\PrDone(w_{\mathcal{N}_1})$ is created. Define $\psi(c_1) = w_{\mathcal{N}_1}$. 
 
\item If Rule\ (\ref{complete-r2}) is not applicable, it implies that $\PrP_1=\PrP_2$ and $\Va_1=\Va_2$, so that 
$\mathcal{N}_1=\mathcal{N}_2$. Set $N = (N \cup \{c_1\}) \setminus \{n\}$, and set $\psi(c_2) = w_{\mathcal{N}_1}$. 

\item Otherwise apply Rule\ (\ref{complete-r2}), introducing a fresh null that we call $w_{\mathcal{N}_2}$ and adding in particular $\PrIns{\PrP_2}(\Va_2,w_{\mathcal{N}_2},w_{\mathcal{N}_2})$.
Then apply Rule\ (\ref{propagate-p}) as many times as necessary, then Rule~(\ref{rule-done}) so that $\PrDone(w_{\mathcal{N}_2})$ is created. 
 Define $\psi(c_2) = w_{\mathcal{N}_2}$ and set $N = (N \cup \{c_1,c_2\}) \setminus \{n\}$. 
 \end{itemize}
 \item If $n$ is a node of $T$ having only one child, set $L = L \cup \{n\}$ and $N = N \setminus \{n\}$. 
 \end{itemize}
 \end{itemize}
At this point, for every $n\in L$, $\psi(n)=w_\mathcal{N}$ and if $\I$ denotes the current set of facts built by our derivation, $\mathsf{saturatedWorld}_\I(\psi(n))=\mathcal{N}$ by construction.

 For all $n \in L$, the next rule applied on $n$ is a rule from $\Sigma_2$, hence 
 no more rule of $\Sigma_1$ are applied on a descendant of $n$ in $T$ (by assumption on $T$). All rule applications below $n$ are thus deterministic, and we ``copy'' that derivation. Let $\psi_n$ be the identity mapping from the label $\mathcal{N}$ of $n$ to $\mathsf{saturatedWorld}_\I(\psi(n))$. By construction, $\psi_n$ is a homomorphism. While we extend $\I$ by applying new rules, we will extend $\psi_n$ into a homomorphism from the label of any descendant $n'$ of $n$ in $T$ to $\mathsf{saturatedWorld}_\I(\psi(n))$. 
 
 \begin{itemize}
 \item For $n'=n$, this is already done.
\item Assume that we have built a derivation such that $\psi_n$ is a homomorphism from the label $\mathcal{N'}$ of some descendant $n'$ of $n$ in $T$ to $\mathsf{saturatedWorld}_\I(\psi(n))$. Let $n''$ be the child of $n'$ in $T$, labelled by $\mathcal{N''}$, and let $\Tuple{\Rule, \sigma}$ be the rule application creating $n''$, where $\Rule\in\Sigma_2$ . 
  Then $\Tuple{\Rule, \psi_n\circ\sigma}$ can be applied to $\mathsf{saturatedWorld}_\I(\psi(n))$, and thus Rule\ \eqref{rule-sigma-2}, 
  instantiated for $\Rule$ is applicable in $\I$ by extending $\psi_n\circ\sigma$ by mapping $w$ to $\psi(n)$. We then define $\psi(n'') = \psi(n)$, and extend $\psi_n$ by mapping each null created in the label of $n''$ by a variable $z_i$ to the null created by variable $z_i$ in the application of Rule \eqref{rule-sigma-2}.
 \end{itemize}

 Let $T'$ be the tree structure having $\psi(r) = w_{\Database}$ as root and where $w_p$ is parent of $w_c$ if $w_c$ has been created by an application of Rule\ (\ref{complete-r1}) or Rule\ (\ref{complete-r2}), 
 that mapped $w$ to $w_p$. 
 Note that all the $w_n$ in this tree structure are exactly the nulls of $\I$ such that there exists a node $n$ in $T$ with $\psi(n)=w_n$, $c_1, c_2$ are children of $n$ in $T$ iff $\psi(c_1), \psi(c_2)$ are children of $\psi(n)$ in $T'$.
  
  Let $\ell$ be a leaf of $T$. By assumption, $\PrGoal$ belongs to the label of $\ell$. Let $n$ be the unique node in $L$ such that $\ell$ is a descendant of $n$. 
 Since $\psi_n$ is a homomorphism from the label of $\ell$ to $\mathsf{saturatedWorld}_\I(\psi(n))$, it holds that $\mathsf{saturatedWorld}_\I(\psi(n))$ contains $\PrGoal$. By definition of $\mathsf{saturatedWorld}_\I$,  it follows that $\PrHatted{\PrGoal}(\psi(n))\in \I$. 
 Since any leaf $w_n$ of $T'$ has an antecedent $n$ by $\psi$ which has a descendant leaf $\ell$ in $T$, then $\PrHatted{\PrGoal}(w)$ holds for each leaf $w$ of $T'$.  
 Extend the derivation as follows:
 \begin{itemize}
  \item Apply Rule\ (\ref{goal-rule-disj}) to each $w$ leaf of $T'$.
  \item If $\PrWorldGoal(w)$ has already been derived for the children $w$ of a node $\psi(n)$ in $T'$, let $\Rule$ be the rule that has been applied by $\sigma$ on $n$ to create its children. Apply Rule (\ref{propagate-ri}) instantiated for $\Rule$ by mapping $w_1$ and $w_2$ to $\psi(c_1)$ and $\psi(c_2)$, where $c_1$ and $c_2$ are the children of $n$ corresponding to the application of $\Rule$. $\PrWorldGoal(\psi(n))$ is thus derived.
  \item By induction on the depth, one derives $\PrWorldGoal(w_\Database)$. As $\PrInit(w_\Database)$ holds, one can apply Rule\ (\ref{final-rule-disj}) and obtain $\PrGoal$.
 \end{itemize}
 We did not apply every possible rule: hence, to build a chase tree, we must ensure fairness (point 3. of the definition of chase tree), by applying all possible remaining rule applications, which concludes the proof.
\end{proof}

To ease that proof, we will actually replace Rule \eqref{propagate-ri} by the following rule:
\begin{align}
\begin{split}
\label{propagate-ri-with-goal}
\PrIns{\PrP_1}(\vec{x}_1,w,w_1)\wedge \PrWorldGoal(w_1) \wedge{} \\
\PrIns{\PrP_2}(\vec{x}_2,w,w_2)\wedge \PrWorldGoal(w_2) \wedge{}\\
 \textstyle\bigwedge_{\PrP(\vec{x}) \in \Body  }  \PrIns{\PrP}(\vec{x},w,w) \\
 &\rightarrow \PrWorldGoal_{\Rule}(w,w_1,w_2)\wedge \PrWorldGoal(w)
\end{split}
\end{align}\\
which has an additional ternary fresh predicate describing the rule $\Rule \in \Sigma_1$ (and the worlds generated by the corresponding rule application) allowing to derive $\PrWorldGoal(w)$. Note that since $\PrWorldGoal_{\Rule}$ not being appearing anywhere else, all the properties shown so far are still valid.

\begin{proposition}
\label{proposition-soundness-disj-removal}
For every database $\Database$ over $\mathcal{S}_{\mathsf{in}}(\Sigma)$, 
if $\Sigma'_1 \cup \Sigma'_2 \cup \Sigma'_3, \Database \models \PrGoal$ then $\Sigma_1 \cup \Sigma_2, \Database \models \PrGoal$.
\end{proposition}

\begin{proof}
 Let $\I$ be a model of $\Database$ and $\Sigma'_1 \cup \Sigma'_2 \cup \Sigma'_3$ obtained through a restricted chase sequence such that $\PrGoal \in \I$. We build a chase tree $T$ for $\Sigma_1 \cup \Sigma_2$ and $\Database'$ for some $\Database' \subseteq \Database$ such that $\PrGoal$ 
  is in the label of every leaf of $T$.
 
 Let $w_{\Database'}\in\EI{\Nulls}{\I}$ be such that Rule\ (\ref{final-rule-disj}) is applied on it, with $\mathsf{world}(w_{\Database'}) = \Database'$. 
 Set $W = \{w_{\Database'}\}$, and perform the following until $W = \emptyset$.
 \begin{itemize}
  \item Let $w \in W$. If $\PrHatted{\PrGoal}(w) \in \I$, set $W = W \setminus \{w\}$.
  \item Otherwise, there must be some $\PrWorldGoal_{\Rule}(w,w_1,w_2) \in \I$, as the only way to derive $\PrWorldGoal(w)$ if $\PrHatted{\PrGoal}(w)$ does not hold is to apply some instantiation for some $\Rule$ of Rule\ (\ref{propagate-ri-with-goal}) by some mapping $\sigma$. Moreover, there must be such an atom for which $w_1$ and $w_2$ are different from $w$. Consider $\Tuple{\Rule,\sigma_{\Terms(\Rule)}}$: it is applicable, as $\sigma$ maps the body of $\Rule$ to $\mathsf{world}(w)$, and none of the two atoms in the head are present in $\mathsf{world}(w)$ (otherwise, $w$ would be equal to $w_1$ or to $w_2$). 
Performing this rule application  adds two children 
to the node of $T$ labelled by  $\mathsf{world}(w)$, 
of respective labels $\mathsf{world}(w_1)$ and $\mathsf{world}(w_2)$. 
Set $W = (W \setminus w) \cup \{w_1,w_2\}$.
 \end{itemize}
As the tree structure of $\I$ is finite, and that at each step, either the size of $W$ is decreasing, or the null $w$ considered is replaced by two nulls $w_1$ and $w_2$ such that $\mathsf{world}(w_1)$ and $\mathsf{world}(w_2)$ are labels of nodes of strictly greater depth than $\mathsf{world}(w)$ in that tree, the above process terminates. Moreover, for each leaf of the considered prefix of $T$, it holds that its label and $\Sigma_2$ entail $\PrGoal$, by Lemma\ \ref{lemma-context-to-existential} and the fact that each leaf is labelled by some 
$\mathsf{world}(w)$ such that $\PrHatted{\PrGoal}(w) \in \I$ (so that $\PrGoal \in \mathsf{saturatedWorld}_\I(w)$). Hence, one can expand $T$ in such a way that all its leaves are labeled by sets containing $\PrGoal$, which concludes the proof.
\end{proof}

\begin{lemma}
 \label{lemma-termination}
For every database $\Database$ over $\mathcal{S}_{\mathsf{in}}(\Sigma)$, $\Tuple{\Sigma'_1\cup\Sigma'_2\cup\Sigma'_3,\Database}$ is chase-terminating. 
\end{lemma}

\begin{proof}
 By Lemma~\ref{lemma-hat-reordering}, chase sequences can be reordered without changing the number of rule applications such that rules from $\Sigma'_1$ are applied first, then rules from $\Sigma'_2$, then rules from $\Sigma'_3$. Lemma~\ref{lemma-finite-structure} implies that rules from $\Sigma'_1$ cannot be triggered indefinitely. As the tree structure is finite, if there are infinitely many rule applications performed with a rule from $\Sigma'_2$, there must be one world null $w^*$ such that there are infinitely many rules of $\Sigma'_2$ that are performed mapping $w$ to $w^*$. By the construction of Proposition~\ref{proposition-soundness-disj-removal}, it implies that there exists a database over $\mathcal{S}_{\mathsf{in}}(\Sigma_2)$ for which $\Sigma_2$ does not terminate, which is against our assumptions on $\Sigma_2$. Finally, $\Sigma'_3$ does not contain any existentially quantified variable, hence only finitely many rules can be triggered.
\end{proof}

Lemma~\ref{lemma-disjunction-removal} is a direct consequence of Propositions~\ref{proposition-completeness-disj-removal}, \ref{proposition-soundness-disj-removal} and Lemma \ref{lemma-termination}.

\section{Proofs of Section 6}

\noindent\textbf{Lemma \ref{lemma:expressivity-limits-2}.} 
\textit{The set $\mathcal{M}$ is not enumerable up to equivalence.}
\begin{proof}[Proof by Contradiction]~
\begin{enumerate}
\item Suppose for a contradiction that the lemma does not hold.
Then, there is an enumerator \ETM that outputs a sequence of TMs that includes $\mathcal{M}$ up to equivalence.
That is, the enumerator \ETM outputs an infinite sequence $\TM_1, \TM_2, \ldots$ of TMs such that:
\begin{itemize}
\item For each $i \geq 1$, we have that $\TM_i \in \mathcal{M}$.
\item For each $\TM \in \mathcal{M}$, there is some $i \geq 1$ such that $\TM$ and $\TM_i$ are equivalent.
\end{itemize}
Two TM $\TM$ and $\TM'$ are equivalent if, for each word $w$, we have that $\TM$ accepts $w$ iff $\TM'$ accepts $w$.
\item Consider the sequence $p_1, p_2, \ldots$ of natural numbers such that $p_1 = 1$ and $p_i$ is the smallest prime with $p_i > p_{i-1}$ for each $i \geq 2$.
Note that this sequence is infinite by Euclid's theorem.
\item By (2): there is an infinite sequence $\Database_1, \Database_2, \ldots$ of databases such that $\D_i = \{\PrEdge(u_1, u_2)$, $\ldots, \PrEdge(u_{p_{i+1}}, u_1)\}$ for each $i \geq 1$.
\item Consider the TM $\TM_d$ that, on input $w$, performs the computation:
\begin{enumerate}
\item Check if $w$ corresponds to a database $\D$ that only contains facts defined over \PrEdge.
If this is not the case, then \emph{reject.}
\item If $\D$ can be homomorphically embedded into a database such as $\{\PrEdge(u_1, u_2), \ldots, \PrEdge(u_{k-1}, u_k)\}$ where $k$ is smaller or equal than the number of nulls in \D, then \emph{reject}.
\item If $\PrEdge(u, u) \in \D$ for some null $u$, then \emph{accept}.
\item If there is some $i \geq 1$ such that \FirstItem there are less or the same number of nulls in $\D_i$ than in $\D$, \SecondItem $\TM_i$ accepts some serialisation that corresponds to $\D_i$, and \ThirdItem there is a homomorphism $h : \D \to \D_i$; then \emph{reject}.
Otherwise, \emph{accept}.
\end{enumerate}
\item By (4): the TM $\TM_d$ halts on all inputs.
Note the following remarks about instruction (4.d):
\begin{itemize}
\item The TM $\TM_i$ accepts some serialisation that corresponds to $\D_i$ if $\TM_i$ accepts any such serialisation.
Therefore, for each instantiation of $i$, the TM $\TM_d$ only needs to check one (arbitrarily chosen) serialisation of $\D_i$ when executing (4.d).
\item When executing (4.d), the TM $\TM_d$ only needs to check a finite amount of instantiations of $i \geq 1$ before rejecting with confidence due to condition (4.d.i).
\end{itemize}
\item After this enumeration, we prove that $\TM_d$ does satisfy \ThirdItem in Definition~\ref{definition:tm-class} by contradiction.
\item By (4.a), (5), and (6): the TM $\TM_d$ is in $\mathcal{M}$.
\item By (4): $\TM_d$ diagonalises over $\TM_1, \TM_2, \ldots$ and $\D_1, \D_2, \ldots$
That is, for any given $i \geq 1$, $\TM_d$ accepts $\D_i$ iff $\TM_i$ rejects $\D_i$.
\begin{itemize}
\item If $\TM_i$ accepts $\D_i$, then $\TM_d$ rejects $\D_i$ in (4.d).
\item If $\TM_i$ rejects $\D_i$, then $\TM_d$ accepts $\D_i$ in (4.d).
Note that, if we assume that $\TM_d$ rejects $\D_i$ in this case we obtain a contradiction; namely, we can conclude that $p$ evenly divides $q$ for some prime numbers $p, q > 1$ with $p \neq q$.
\end{itemize}
\item Contradiction by (1), (7), and (8): the enumerator \ETM is incomplete since it fails to print out a TM that is equivalent to $\TM_d$, which is in $\mathcal{M}$.
\end{enumerate}

Suppose for a contradiction that there are some words $w$ and $v$ that correspond to some databases \D and \DA such that $\D$ and $\DA$ only contain facts defined over \PrEdge, $\TM_d$ accepts \D, $\TM_d$ rejects \DA, and there is some homomorphism $\Hom : \Database \to \DA$.
We conduct a case-by-case analysis to show that this assumption results in a contradiction:
\begin{itemize}
\item Assume that $\D$ is accepted due to instruction (4.c).
Then, $\PrEdge(u, u) \in \D$ for some null $u$ and hence, $\PrEdge(t, t) \in \DA$ for some null $t$ since $\Hom : \D \to \DA$.
Therefore, $\TM_d$ accepts $\DA$ due to (4.c) ($\lightning$).
\item Assume that $\D$ is accepted due instruction (4.d).
Two possible cases arise:
\begin{itemize}
\item By definition \DA is a database that only contains facts defined over the predicate \PrEdge.
Hence, \DA cannot be rejected due to (4.a).
\item If \DA is rejected due to (4.b), then $\DA$ can be hom-embedded into a path over $\PrEdge$.
Therefore, $\D$ can also be hom-embedded into the same path since $\Hom : \D \to \DA$ and $\TM_d$ rejects $\D$ due to (4.b) ($\lightning$).
\item Assume that \DA is rejected due to instruction (4.d).
\begin{enumerate}
\item There is some $i \geq 1$ and a homomorphism $\HomA$ such that \FirstItem the number of nulls in $i$ is smaller than the number of nulls in \DA, \SecondItem $\HomA : \DA \to \D_i$, and \ThirdItem the TM $\TM_i$ accepts the database $\D_i$.
\item By (1): $\HomA \circ \Hom : \D \to \D_i$.
\item If we assume that $p_{i+1}$ is smaller or equal than the number of nulls in \D, then $\TM_d$ rejects $\D$ due to (4.d) ($\lightning$).
Therefore, we conclude that $p_{i+1}$ is strictly greater than the number of nulls in \D.
\item If we assume that \D can be homomorphically embedded into a path over $\PrEdge$, then $\TM_d$ rejects $\D$ due to (4.b).
Hence, we assume that this is not the case.
\item We obtain a contradiction from (2), (3), and (4).\qedhere
\end{enumerate}
\end{itemize}
\end{itemize}
\end{proof}

\end{document}